\documentclass[11 pt,final]{article}

\usepackage{srsWorkingPaper}

\usepackage[utf8]{inputenc}
\usepackage{amsmath} % for align
\usepackage{amsthm} % for proof pkg
\usepackage{amssymb}
\usepackage{mathrsfs}
\usepackage{bbm}
\usepackage{bm}
\usepackage{bookmark}
\usepackage{hyperref}

\usepackage{algorithm}
\usepackage{algorithmicx}

\usepackage{algpseudocode}
\usepackage{graphicx}
\usepackage{subcaption}

\usepackage{calc}
\usepackage{enumitem}

\usepackage{placeins} % floatbarrier

\newtheorem{theorem}{Theorem}
\newtheorem{corollary}{Corollary}
\newtheorem{lemma}{Lemma}
\newtheorem{proposition}{Prop.}
\theoremstyle{definition}

\theoremstyle{remark}\newtheorem{remark}{Remark}

\theoremstyle{remark}

\newcommand{\balpha}{\overline\alpha}
\newcommand\bzpo{\bar{x}_{U,1}}
\newcommand\balphao{\balpha_1}
\DeclareMathOperator{\bzp}{\bar x_U}
\DeclareMathOperator{\hzp}{\hat x_U}
\DeclareMathOperator{\hzph}{\hat x_{U,h}}
\DeclareMathOperator{\hzpl}{\hat x_{U,l}}
\DeclareMathOperator{\halpha}{\hat\alpha}
\DeclareMathOperator{\btheta}{\overline\theta}
\DeclareMathOperator{\itheta}{\theta}
\newcommand{\rhou}{p^U}

\newcommand{\de}{\mathrm{d}}
\newcommand{\Iff}{\Leftrightarrow}
\newcommand{\Then}{\Rightarrow}
\newcommand{\eps}{\varepsilon}

\title{Designing Autonomous Markets for Stablecoin Monetary Policy\footnote{%
		The design from this paper is implemented as part of the upcoming Gyroscope stablecoin system under the name \emph{Dynamic Stability Mechanism (DSM)}. The source code will be made available later.
}}

\author{Ariah Klages-Mundt\thanks{Researcher at Superluminal Labs; in a separate capacity, a PhD student at Cornell University.}   \ \ \ \ \ \ \ \
	Steffen Schuldenzucker\thanks{Researcher at Superluminal Labs}
}

\date{December 2022}

\begin{document}
	
% For some reason these can only be done after \begin{document}.
\renewcommand\algorithmicrequire{\textbf{Input:}}
\renewcommand\algorithmicensure{\textbf{Output:}}
	
\maketitle

\begin{abstract}
	We develop a new type of automated market maker (AMM) that helps to maintain stability and long-term viability in a stablecoin.
This primary market AMM (P-AMM) is an autonomous mechanism for pricing minting and redemption of stablecoins in all possible states and is designed to achieve several desirable properties.
We first cover several case studies of current ad hoc stablecoin issuance and redemption mechanisms, several of which have contributed to recent stablecoin de-peggings, and formulate desirable properties of a P-AMM that support stability and usability.
We then design a P-AMM redemption curve and show that it satisfies these properties, including bounded loss for both the protocol and stablecoin holders.
We further show that this redemption curve is path independent and has properties of path deficiency in extended settings involving trading fees and a separate minting curve.
This means that system health weakly improves relative to the path independent setting along any trading curve and that there is no incentive to strategically subdivide redemptions.
Finally, we show how to implement the P-AMM efficiently on-chain.

\end{abstract}

\section{Introduction}

The design of non-custodial stablecoins has faced several recent turning points, both in the Black Thursday crisis in Dai and the recent churn of algorithmic stablecoins. These have both pointed toward the importance of designing good primary markets for stablecoins ---i.e., mechanisms for pricing minting and redeeming of stablecoins. The terminology here is borrowed from ETF market structure and contrasts ``primary market'', where shares of a fund are minted and redeemed for underlying assets, and ``secondary market'', where existing shares are traded for other assets (and where ordinary exchange trading takes place).
A primary market for a stablecoin helps to maintain peg by allowing users to exchange stablecoins with the protocol near the peg price, providing a means to arbitrage other markets back toward the peg value, should they deviate.
In a blockchain protocol, primary markets are automated; we call the mechanism that controls mint and redeem prices the \emph{primary algorithmic market maker (P-AMM)} of the stablecoin. Every stablecoin design has a P-AMM, be it intentionally designed or not.

 Black Thursday in March 2020 saw a $\sim\!50\%$ crash in ETH in the day. This triggered a deleveraging spiral, a short squeeze effect that causes the price of Dai to increase as borrowers need to buy it to reduce exposure. This was shown to amplify collateral and liquidity drawdown and cause instability in Dai \cite{klages2019stability,klages2020while}.
This demonstrated fundamental problems around deleveraging, liquidity, and scaling in stablecoins like Dai, in which supply depends on an underlying market for leverage as opposed to a primary market.\footnote{
	In this market, a speculator can post collateral and borrow Dai against this collateral to achieve a risky leveraged position. As a result, the supply of Dai will depend on the demand for leverage, which can and does plummet in a crisis.
} 

\paragraph{PSM: a primary market for Dai}
Patching the deleveraging problem has been a major topic since Black Thursday. Several approaches have been pursued, the most prominent of which is the tethering of Dai to the custodial stablecoin USDC.\footnote{
	Two other notable approaches are using negative rates to equilibrate supply and demand at the target (e.g., Rai) and using dedicated liquidity pools to smooth the effects of deleveraging (e.g., Liquity and the solution proposed in \cite{klages2020while}). The former leads to questions of liquidity and equilibrium participation under negative rate regimes, and the latter is not a full fix as it smooths but postpones spirals.
}
This takes the form of Maker's peg stability module (PSM), which maintains exchangeability of Dai with USDC via a protocol-held USDC reserve. The PSM in this way effectively becomes a primary market for minting and redeeming Dai, backed by USDC reserves. The PSM has greatly enhanced the liquidity around Dai's peg and its resilience to deleveraging spirals, as evidenced in Figure~\ref{fig:dai-spiral-vs-psm}, which plots Dai price for days $t$ since the major ETH shocks of 12 March 2020 (w/o PSM) and 19 May 2021 (w/ PSM).
However, this has further exposed the scaling problem of the original Dai mechanism: the leverage market doesn't necessarily scale with demand. Since the May 2021 crisis, Dai backed by USDC has grown from 17\% to now over 60\% of the Dai supply (Figure~\ref{fig:dai-psm}). This arguably compromises the decentralization of Dai by importing the custodial and regulatory risks of USDC.

\begin{figure*}
	\centering
	\begin{subfigure}[b]{0.49\textwidth}
		\centering
		\includegraphics[width=\textwidth]{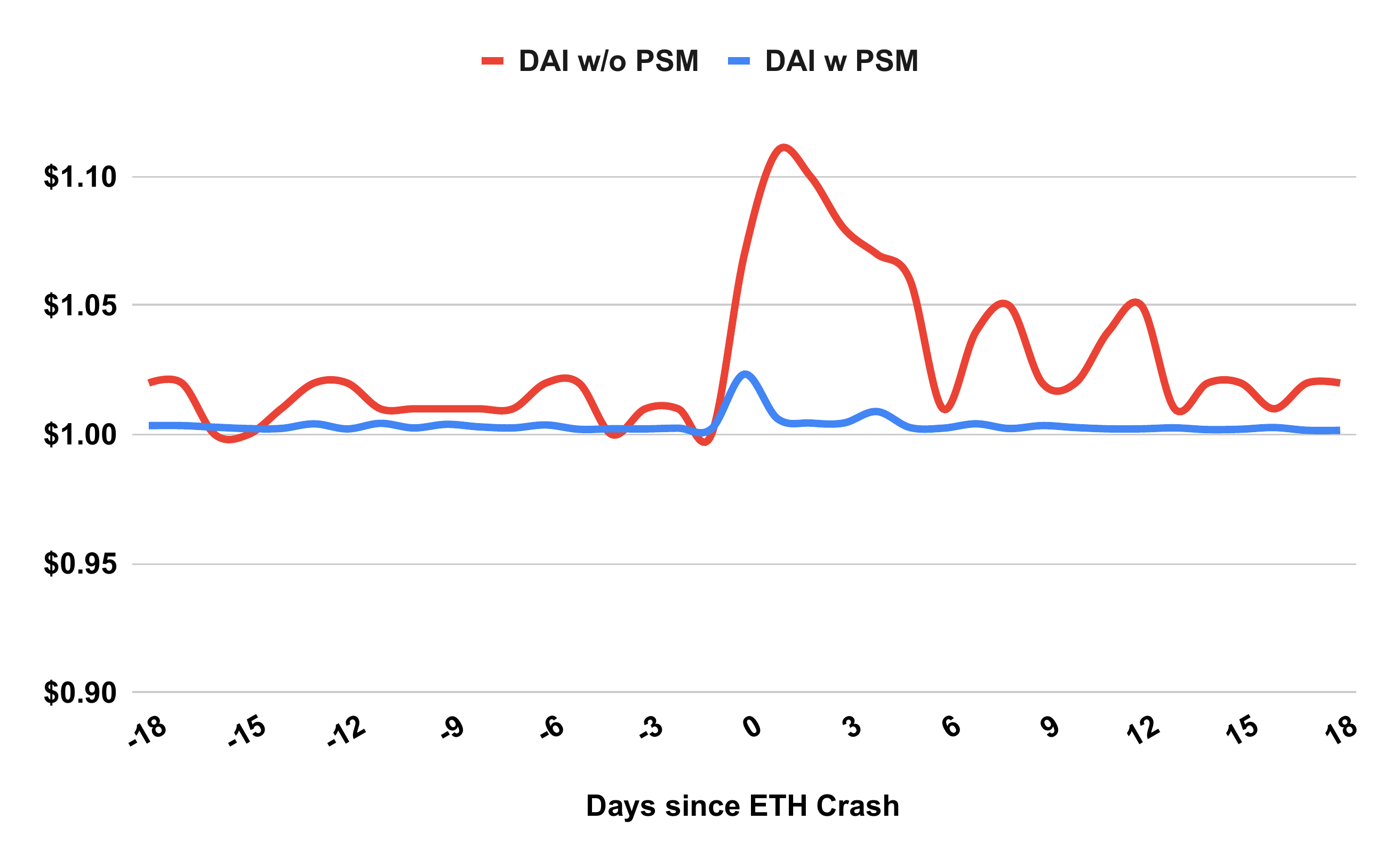}
		\caption{}\label{fig:dai-spiral-vs-psm}
	\end{subfigure}
	\begin{subfigure}[b]{0.49\textwidth}
		\centering
		\includegraphics[width=\textwidth]{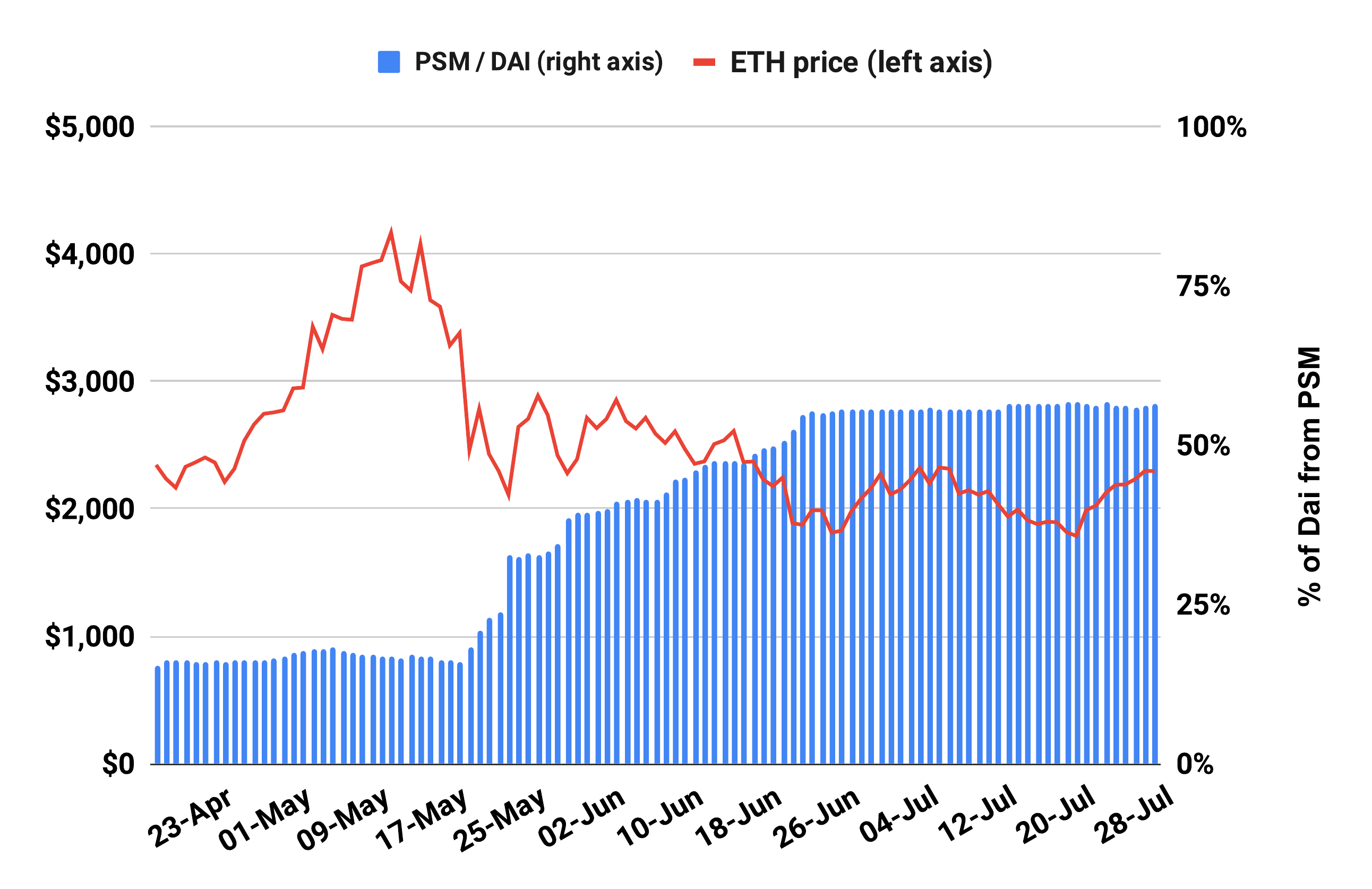}
		\caption{}\label{fig:dai-psm}
	\end{subfigure}
	\caption{Effects of the Dai PSM. (a) Dai price for days $t$ since major ETH shocks w/ and w/o the PSM, (b) The portion of Dai issued through the PSM has grown $>3\times$ since the May 2021 ETH shock.}\label{fig:dai-study}
\end{figure*}

\paragraph{Algorithmic Stablecoins}
The problems with Dai, including its most recent USDC centralization issue, has also motivated a wave of algorithmic stablecoins, which aim to keep the stablecoin supply in line with demand algorithmically. 
Mostly, these designs have either no or dubious degrees of asset backing. These designs have almost universally experienced depegging events, as depicted in Figure~\ref{fig:algo-montage} due to susceptibility to downwards spirals arising from the lack of asset backing and ad hoc primary market structure when under-reserved.

\begin{figure}
	\centering
	\includegraphics[width=0.8\textwidth]{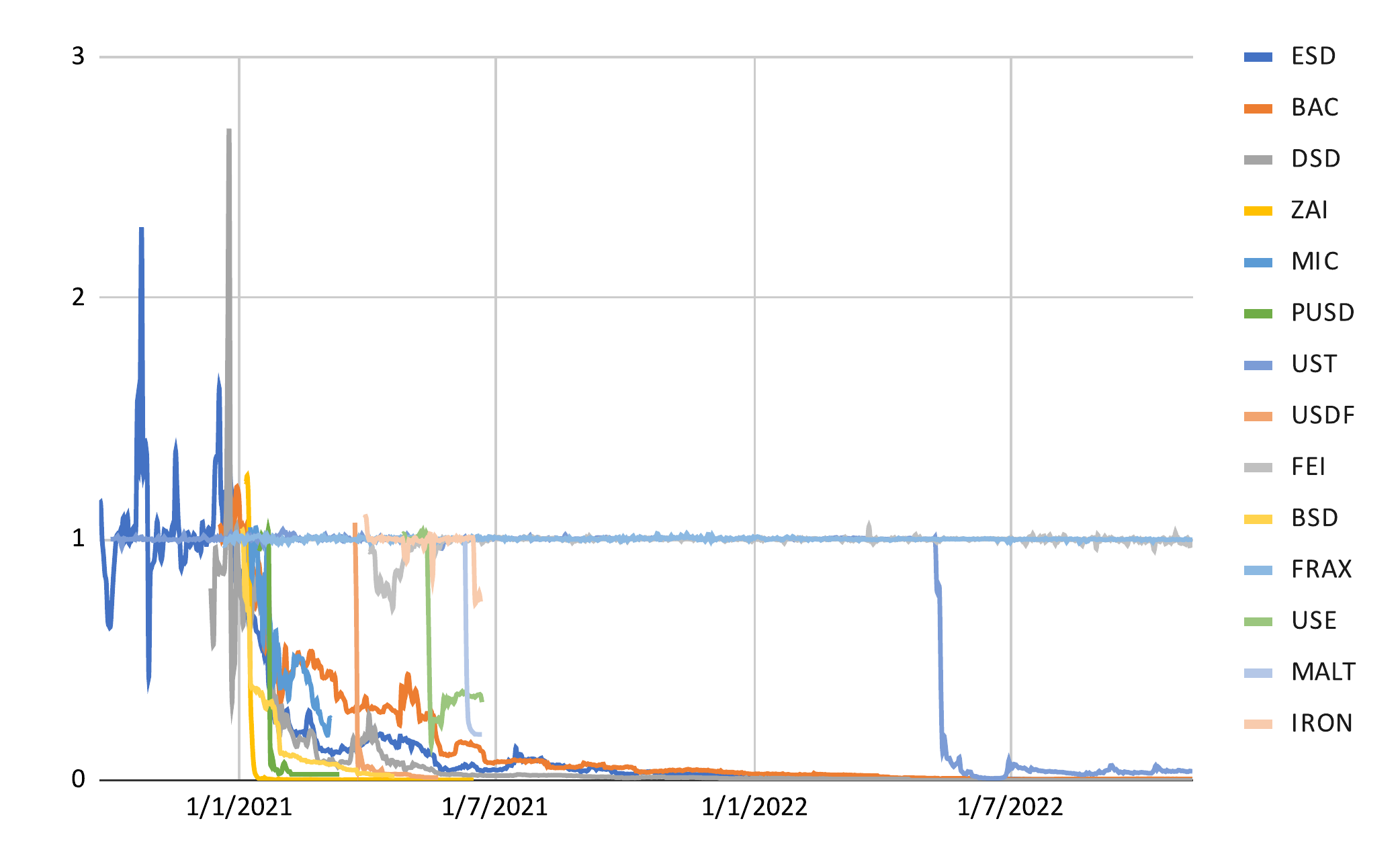}
	\caption{The recent churn of algorithmic stablecoins.}\label{fig:algo-montage}
\end{figure}

Algorithmic stablecoins are best understood in the context of currency peg models, such as \cite{morris1998unique}. In a simplified sense, these systems are backed by two sources of value: (i) asset backing in a currency reserve and (ii) economic usage, an intangible value that represents the demand to hold the currency as it unlocks access to an underlying economy.
Supposing these two values are together great enough, a currency peg is maintainable; otherwise, it is susceptible to breaking.
A peg break can also be triggered by a speculative attack that is profitable for the attacker, akin to the attack on the British pound on Black Wednesday.

Algorithmic stablecoins have encountered several fundamental problems, which contribute challenges to primary market design. Many are under-reserved by design while having no native economic usage, leading to many observed depeggings through downwards spirals, often exhausting the asset backing.
Further, the composition of reserve assets that can be held on-chain are inherently risky. 
In some cases, these assets are non-existent (e.g., Basis). In seigniorage shares--style designs (e.g., Terra and Iron), the backing is effectively the value of ``equity shares'', which have an endogenous/circular price with the expected growth of the system and can spiral to zero \cite{klages2020stablecoins}.
A better equipped limit case is a stablecoin backed fully by a portfolio of exogenous, but risky and possibly illiquid, assets, which could enter an under-reserve state if assets experience shocks.
Any design that could experience under-reserved states requires a good policy for how the protocol applies reserve assets to maintain liquidity near the peg in sustainable ways, factoring in the state of reserve backing.
This challenge effectively becomes the problem of designing a primary market for the stablecoin.

%%%%%%%%%%%%%%%%%%%%%%%%%%%%%%%%%
\subsection{This Paper}

In this paper, we study the rigorous design of stablecoin primary markets with these desirable properties.
In particular, a well-designed curve will be able to adapt shape/pricing autonomously to achieve these properties.
Such a formulation will require minimal intervention by governance, further limiting risks from governance extractable value \cite{lee2021gov}.
%In the remainder of this paper, we confront this challenge in designing our P-AMM.
%the Gyroscope P-AMM, initially envisioned in \cite{gyroscope2020}.
We make the following contributions.
\begin{itemize}
	\item We introduce our analysis framework of the \emph{redemption curve} of a stablecoin and conduct a case study of different existing P-AMM designs (Section~\ref{sec:case-studies}).
	
	\item We formulate the P-AMM desiderata (Section~\ref{sec:desiderata}).
	
	\item We design a new P-AMM redemption curve implicitly (Section~\ref{sec:design}). The design is parameterized by a virtual \emph{anchor point} that captures system health and redemption pressure. We specify the curve explicitly as a function of the anchor point (Section~\ref{sec:calc-params}) and we show that the anchor point is uniquely determined by the current system state (Section~\ref{sec:reconstruction-uniqueness}), implying that the overall design is well-defined.
	We also establish that the shape of the redemption curve satisfies several of our desiderata directly.
	Furthermore, we formulate a simpler redemption curve that satisfies many, but not all, desiderata (Section~\ref{sec:discrete-decay} and Appendix~\ref{apx:discrete-decay}).
	
	\item We show that our P-AMM redemption curve is path independent and we prove several path deficiency properties in an extended setting with trading fees and minting. In particular, we show that system health weakly improves relative to the path independent setting along any trading curve and there is no incentive to strategically subdivide redemptions (Section~\ref{sec:path-properties}).
	
	\item Finally we show that our P-AMM can be implemented efficiently on-chain. Specifically, the implementation only requires a constant number of basic arithmetic operations and at most two square roots to evaluate. The implementation makes use of additional structural results  (Section~\ref{sec:implementation}).
\end{itemize}

The result of this paper is an autonomously adapting P-AMM that satisfies the desired properties throughout the possible state space. The P-AMM formulation contains a few select hyperparameters, which can in principle be tuned by governance; however, the desired properties stand over the entire parameter space. In particular, if parameters are tuned, it does not need to happen on-the-fly, thus still minimizing reliance on governance intervention.

Note that the stablecoin peg target is implicitly \$1. However, the mechanics remain fully functional under any arbitrary target within a block. Through adapting this implicit parameter, the system could also implement much more arbitrary monetary policy while retaining desirable properties.

We make two important, but not contentious, assumptions in this work.
%First, since we are focusing on the Gyroscope stablecoin, which does not incorporate endogenous/circularly priced collateral, we assume that reserve assets are exogenously priced.
First, we exclude endogenous/circularly priced collateral by assuming that reserve assets are exogenously priced.
Second, we assume that the system has an accurate oracle that provides the price of reserve assets in USD. The need for an oracle is inescapable in a stablecoin that pegs to outside assets, so this is not an unusual assumption. Since oracles can provide manipulation surfaces (see, e.g., \cite{werner2021sok}), it is important to incorporate other protective mechanisms; however, these are separate from the P-AMM mechanism itself.

%note: we first describe our mechanism for general states consisting of total reserve value, gyro dollar supply, and redemption level. implementation: we show that our constructions can in fact be normalized by assuming that $y_a = 1$, which simplifies the computation, which allows useful precomputations to reduce computational effort

\subsection{Primary Markets: Related Work}

Our work is most closely related to currency peg models in international economics (e.g., \cite{morris1998unique,guimaraes2007risk}) as well as models of pegged money market mutual funds (e.g., \cite{parlatore2016fragility}). Although these types of models have been discussed and adapted recently in the context of stablecoins (e.g., \cite{routledge2021currency,li2020managing}), there is no prior work building a cryptocurrency mechanism that can adapt the lessons of good currency peg policies. Our work constructs such a mechanism from first principles that functions in a novel way as an autonomous primary market for stablecoin issuance.
%This can be thought of as a passive, pre-programmed version of open market operations ---how central banks interact with markets to shape monetary policy--- embedded in a stablecoin protocol.
For further academic background on stablecoins, we refer to \cite{klages2020stablecoins,bullmann2019search,pernice2019monetary} and references therein.

For an overview of theory work on AMMs, we refer to \cite{angeris2020improved,angeris2020does,capponi2021adoption}. Further background and references are available in \cite{werner2021sok}. Current AMMs resemble Uniswap \cite{adams2021uniswap} and Curve \cite{egorov2019stableswap}, in which liquidity providers add pairs of assets to a pool that quotes trading prices algorithmically depending on the state of the pool. Such AMMs are \emph{secondary markets} in which assets that already exist are traded. In contrast, we develop a new type of AMM that plays the role of a \emph{primary market}, in which assets are minted and redeemed against the protocol itself. No existing work analyzes such constructions for stablecoins.

A well-designed primary market for a stablecoin can be interpreted as an autonomous version of open market operations, comparing with how central banks interact with markets to shape monetary policy.
See Appendix~\ref{apx:monetary-theory} for further discussion on this comparison.
%When new stablecoins are sold on the primary market, the balance sheet is expanded, and when stablecoins are redeemed, the balance sheet is contracted. The primary market design determines how much the balance sheet changes, supposing all proceeds of the market go onto the balance sheet.\footnote{
%	Notably, many algorithmic stablecoins divert a share of primary market cash flow to holders of ``equity'' tokens ---we consider such systems \emph{insolvent-by-design} as they give away part of the ``seigniorage'' income from purchases of newly minted stablecoins (typically via buybacks of ``equity'' tokens), unlike a bank that maintains full asset-backing of deposits.
%	This structure has contributed to many experienced crises for these coins.
%}
%Notice here that the primary market mechanism essentially solves the scaling issues that arise in leverage-based stablecoins, like Dai: the stablecoin is always able to meet excess demand by expanding the balance sheet (it does not need to match the demand of other agents in doing so).
Designing an algorithmic primary market presents a challenge akin to designing a rules-based monetary policy (i.e., a \emph{Taylor Rule}; see, e.g., \cite{orphanides2010taylor}).
However, rather than setting nominal interest rates in a quasi-algorithmic way, as in normal central bank monetary policy, an algorithmic primary market is setting prices in a programmatic way.
%As this is not reliant on an assessment of the output gap and GDP, among other things that can't be measured accurately, we might expect this to be easier and more robust (it is notoriously difficult to formulate traditional Taylor rules that are robust to wide arrays of settings).

%% MOVED DOWN INTO CASE STUDIES SECTION:
%In this paper, we consider the wider problem of designing good primary market mechanisms for minting and redeeming stablecoins. We first cover several case studies that illustrate the consequences of different primary market shapes.
%These help to explain the issues these stablecoins have faced and illustrate that existing primary markets ---which are often not directly recognized as primary markets--- are ad hoc in design, both in the shapes chosen and the ability to adapt shape to changing circumstances.
%One general consequence of this ad hoc structure is that existing systems are left to rely on protocol governance to make quick fixes to primary market shape in a crisis, which further introduces vectors for governance abuse (e.g., governance extractable value \cite{lee2021gov}). This can cause significant problems in decentralized and pseudonymous systems.

%%%%%%%%%%%%%%%%%%%%%%%%%%%%%%%%%

\section{Primary Markets: Case Studies}\label{sec:case-studies}
% todo To Steffen, it was not clear what exactly a “curve of an AMM” is or how AMMs relate to what we're doing. This might be my relatively little exposure to DeFi, though. Clarify if you think it makes sense. (Ari: partly addressed)
In contrasting the shapes of primary market mechanisms, it will be useful to interpret these as automated market makers (AMMs), which price the exchange of assets algorithmically along a curve as a function of reserves and possibly other state variables.\footnote{
	For recent background on AMMs, see \cite{angeris2020improved}, which focuses on constant function market makers (CFMMs). However, primary markets will not fit this category in general.
	% Note that a stablecoin's \emph{primary market AMM (P-AMM)} may not be intentionally planned and systems may rely on protocol governance to intervene in a crisis. This may introduce vectors for governance abuse (e.g., governance extractable value \cite{lee2021gov}).
}
Sometimes these are explicit AMM curves implemented by the protocol, while other times we must factor in the effects of several mechanisms to find an implicit AMM curve that describes the primary market.
This AMM structure will depend on the assets backing the system. In some cases, these assets are \emph{implicitly} backing the system, such as in seigniorage shares systems. Other times, they are a portfolio of assets more explicitly. We will refer to this asset backing as the \emph{reserve assets}.

\begin{figure*}
	\centering
	\begin{subfigure}{0.49\textwidth}
		\includegraphics[width=\textwidth]{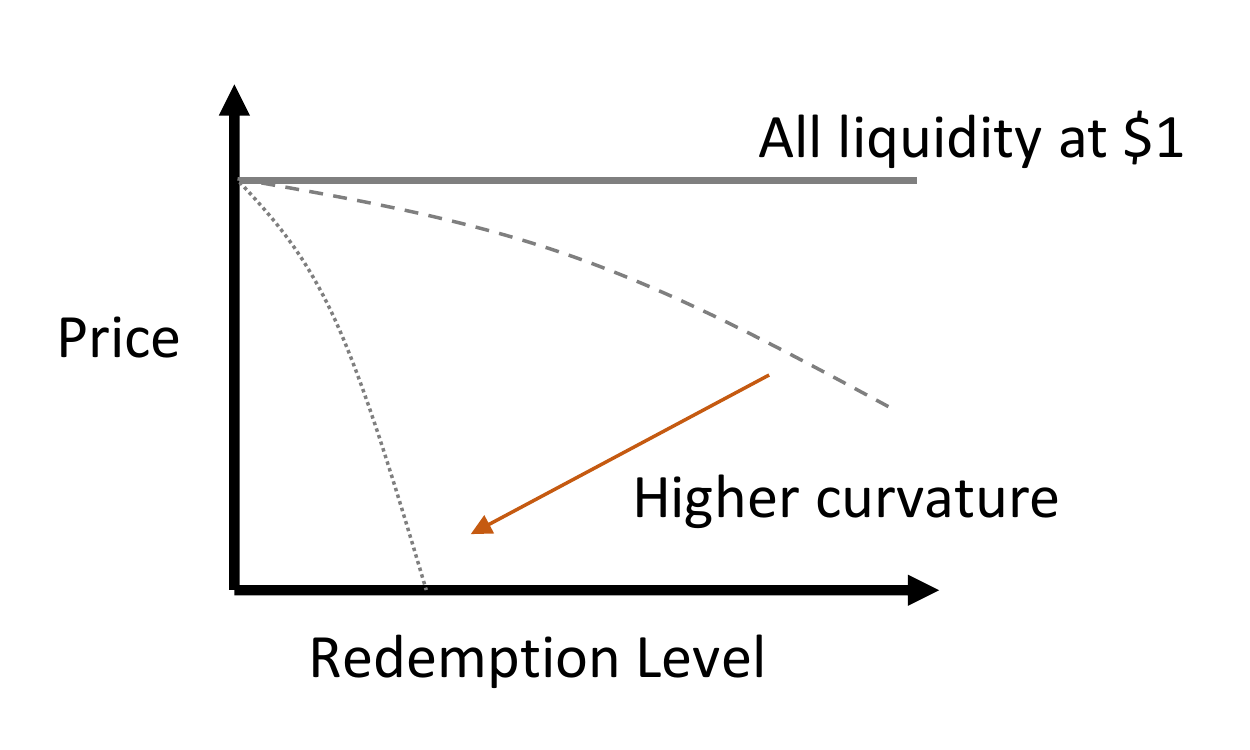}
		\caption{}\label{fig:stylized-curves1}
	\end{subfigure}
	\begin{subfigure}{0.49\textwidth}
		\includegraphics[width=\textwidth]{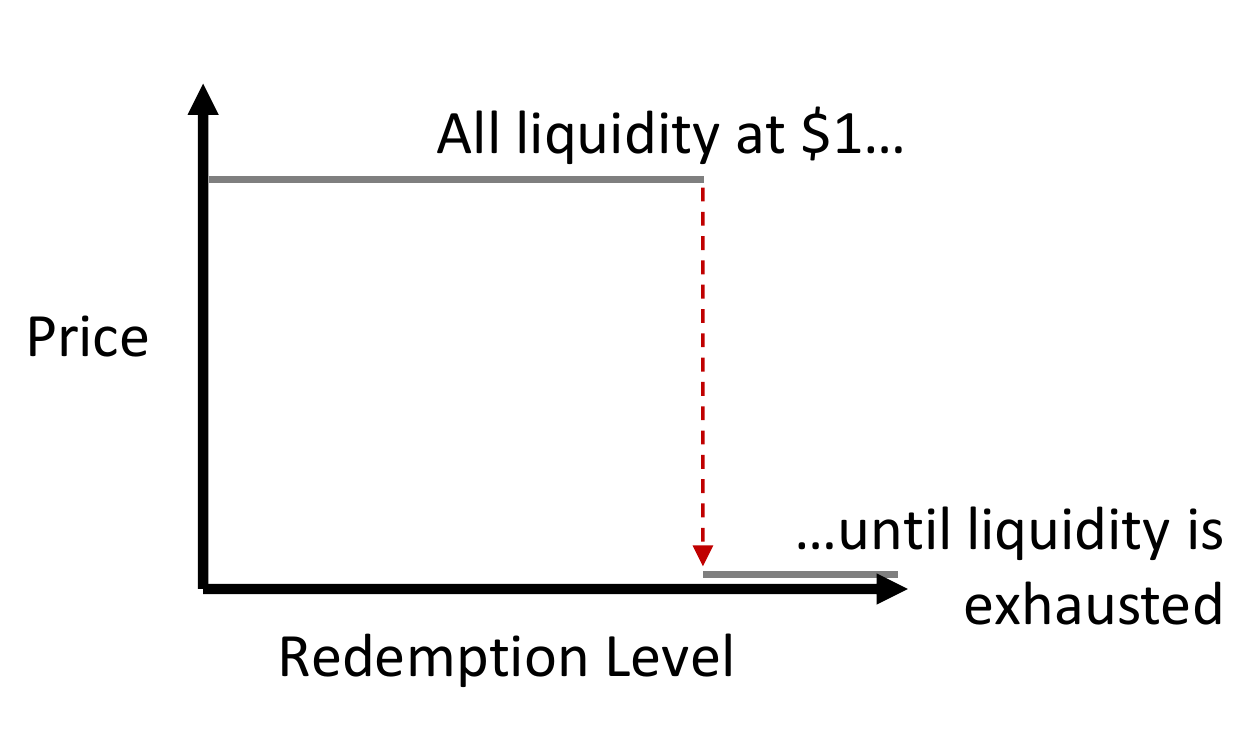}
		\caption{}\label{fig:stylized-curves2}
	\end{subfigure}
	\caption{Stylized primary market redemption curves.}\label{fig:stylized-curves}
\end{figure*}

A primary market can be separated into two curves (possibly coinciding): a minting curve and a redemption curve. These curves price the stablecoin in terms of underlying reserve assets as a function of system state (e.g., level of redemptions and reserve health). To illustrate, Figure~\ref{fig:stylized-curves1} shows some stylized possible redemption curves, plotted as a function of redemption level. An advanced redemption curve might shift the curvature of the 2-d curve as other variables in the state change (e.g., reserve health). Note that, should the reserve assets backing the system be exhausted, the redemption curve becomes flat at \$0, as depicted in Figure~\ref{fig:stylized-curves2}.

\paragraph{USDC/USDT}
Custodial stablecoins like USDC and USDT have flat redemption curves at $\sim\$1$. Note that this primary market has a large off-chain component, where dollars are actually exchanged for stablecoins.
Because of this off-chain component, users must trust the issuer to maintain the primary market. There are various reasons why this may not happen or may not be possible, including custodial and regulatory risks as well as potential loss on risky reserve assets. Note that Dai's PSM discussed above essentially borrows USDC's primary market by wrapping USDC, and so the PSM redemption curve is similar.

%In traditional finance, e-money, money market mutual funds with pegged redemptions, and exchange traded funds (ETFs) also bear resemblance to this type of primary market structure.

\paragraph{Basis}
In Basis-type designs, including Basis Cash and Empty Set Dollar, there is an implicit redemption curve for ``coupons'', which promise to be redeemable for a multiple of new stablecoins later. Often, these coupons also expire a certain time after creation.
However, since there is no asset backing of the system (all income from selling newly minted stablecoins is disbursed to various stakeholders), there is no redemption available for exogenous assets.
In the event that stablecoin demand and willingness to speculate on growth of the system deteriorates, the value of these coupons circularly goes to zero, and the redemption curve becomes flat at \$0.
As seen in Figure~\ref{fig:algo-montage}, these systems did not maintain a peg both because of this solely circular value structure and the lack of a supporting primary market mechanism.

\begin{figure*}
	\centering
	\begin{subfigure}[b]{0.49\textwidth}
		\includegraphics[width=\textwidth]{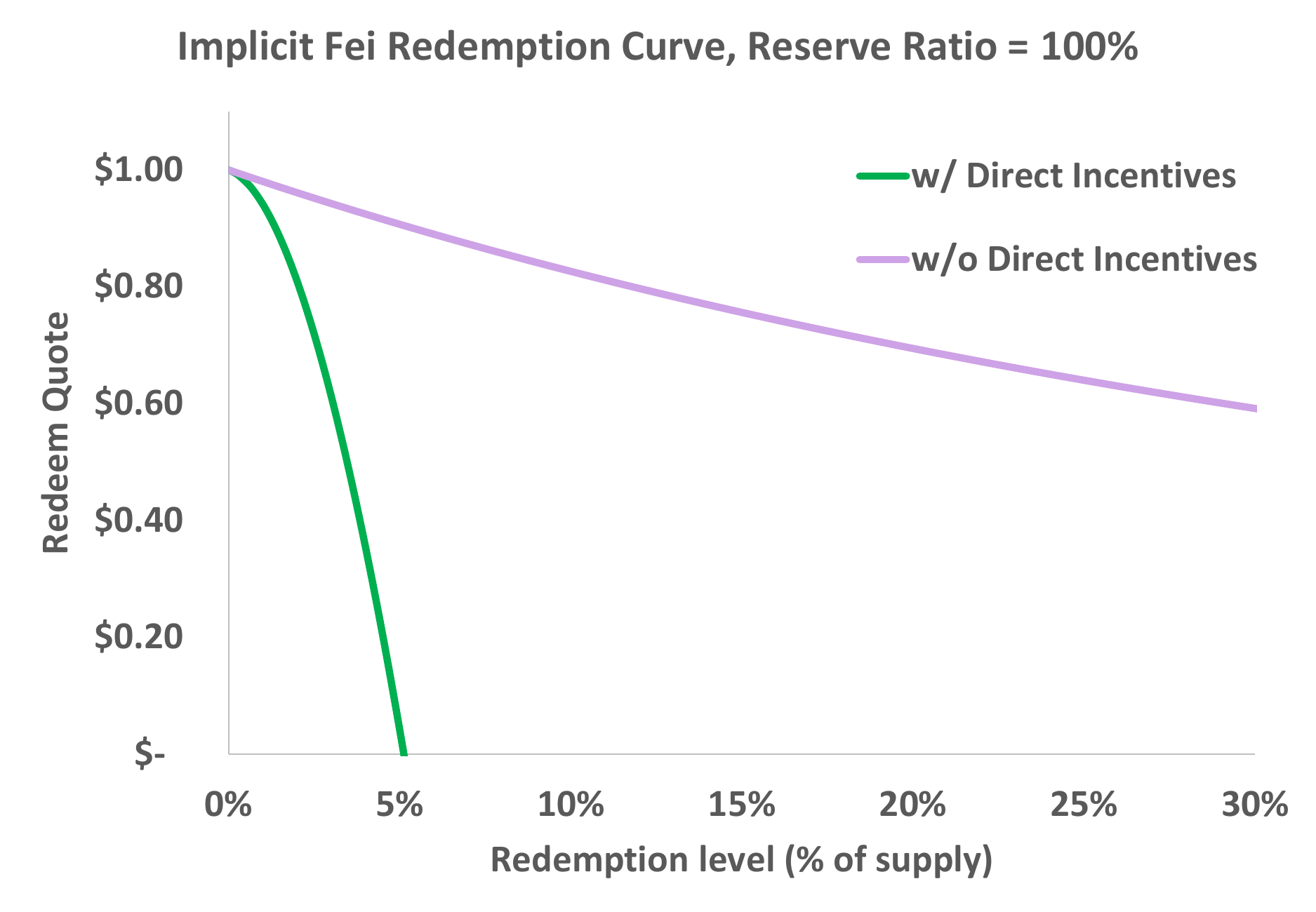}
		\caption{}\label{fig:fei-redeem-curve}
	\end{subfigure}
	\begin{subfigure}[b]{0.49\textwidth}
		\includegraphics[width=\textwidth]{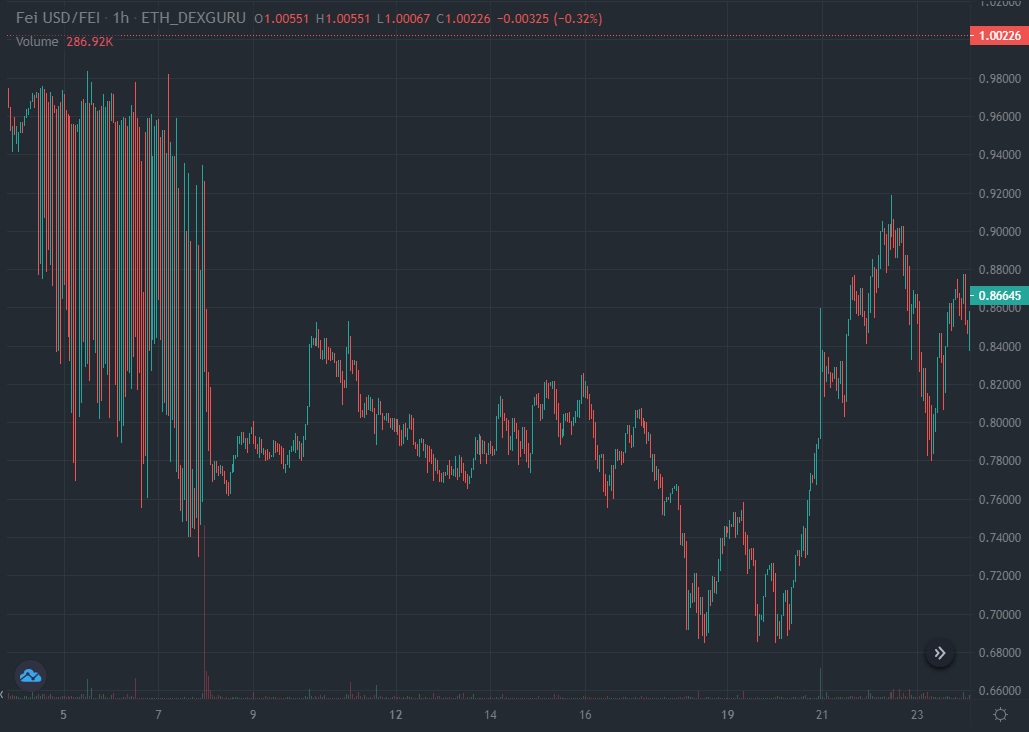}
		\caption{}\label{fig:fei-chart}
	\end{subfigure}
	\caption{Fei case study. (a) implicit redemption curve shape w/ and w/o direct incentives. (b) depegging following large redemptions at launch.}\label{fig:fei-study}
\end{figure*}

\paragraph{Fei}
The Fei stablecoin places reserve assets in a constant-product AMM pool. The action akin to redemption is to sell Fei for ETH in this pool. At launch, this pool was designed with a fee that grows quadratically in the amount of redemptions (``direct incentives''). This has the effect of distorting the implicit Fei redemption curve into a poor shape (see Figure~\ref{fig:fei-redeem-curve}) that essentially guarantees low primary market liquidity during a supply contraction, leaving Fei holders entrapped once secondary market liquidity dries up, even under good reserve health.
Following the Fei launch in Apricl~2021 and despite an effectively high reserve ratio, the implicit redemption curve was unable to handle the size of sought redemptions, leading the stablecoin to deviate erratically from the peg (see Figure~\ref{fig:fei-chart}).\footnote{
	The direct incentive mechanism was later removed, shifting the implicit redemption curve to the constant product curve visualized in Figure~\ref{fig:fei-redeem-curve}. Later, Fei governance chose to forgo this implicit redemption curve by offering direct redemptions at \$0.95.
}

%This is precisely what happened following the Fei launch in April 2021. Although initially intending to be under-reserved at launch, the system actually started over-reserved because of appreciation in the ETH reserve asset. 
%As the initial Fei supply was much higher than demand, many holders sought to exit the system. Despite the effectively high reserve ratio, the implicit redemption curve was unable to handle the size of sought redemptions, leading the stablecoin to deviate erratically from the peg, as seen in Figure~\ref{fig:fei-chart}.
%The direct incentive mechanism was later removed, shifting the implicit redemption curve to the constant product curve visualized in Figure~\ref{fig:fei-redeem-curve}. Later, Fei governance chose to forgo this implicit redemption curve by offering direct redemptions at \$0.95.

\paragraph{UST / Seigniorage Shares}

The TerraUSD (UST) stablecoin was intended to be backed by a seigniorage shares-style token, LUNA. UST was redeemable for newly minted LUNA, inflating the supply of the latter. The redemption curve was intended to be flat at \$1. However, since LUNA was an endogenous/circularly priced asset, there was no guarantee that speculators' demand would be enough to support its UST backing.
In Feb. 2022, a partial backing of Bitcoin was added. Despite this, in May 2022, a downwards spiral was triggered, in which UST saw mass redemptions, the Bitcoin reserve was exhausted, and the value of LUNA went to zero, removing essentially all support from UST and allowing it to crash to mere cents on the dollar (see Figure~\ref{fig:ust-study}). This is a variation on the stylized redemption curve in Figure~\ref{fig:stylized-curves2}.

\begin{figure*}
	\centering
	\begin{subfigure}[b]{.9\textwidth}
		\includegraphics[width=\textwidth]{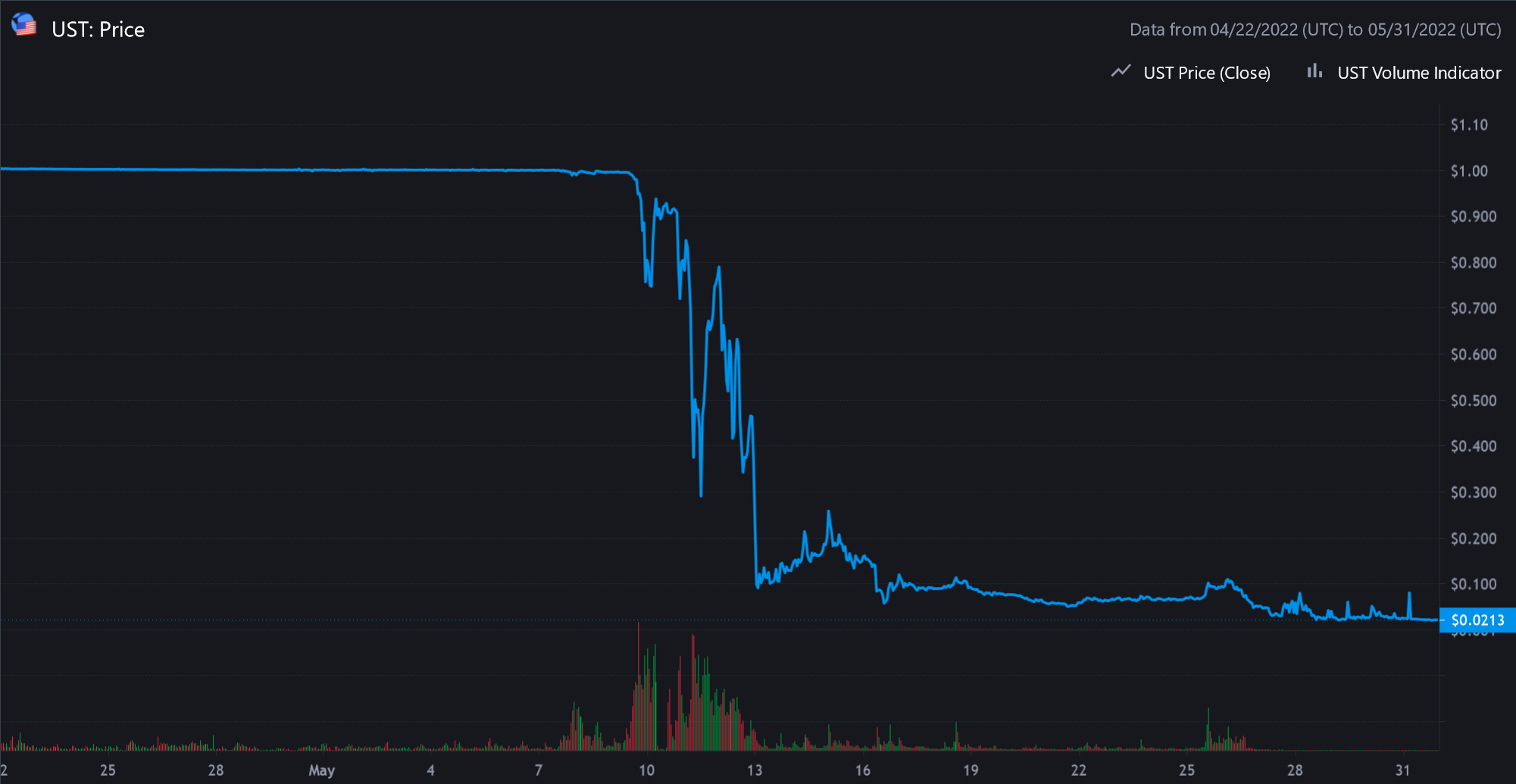}
		\caption{}\label{fig:UST}
	\end{subfigure}
	\begin{subfigure}[b]{.9\textwidth}
		\includegraphics[width=\textwidth]{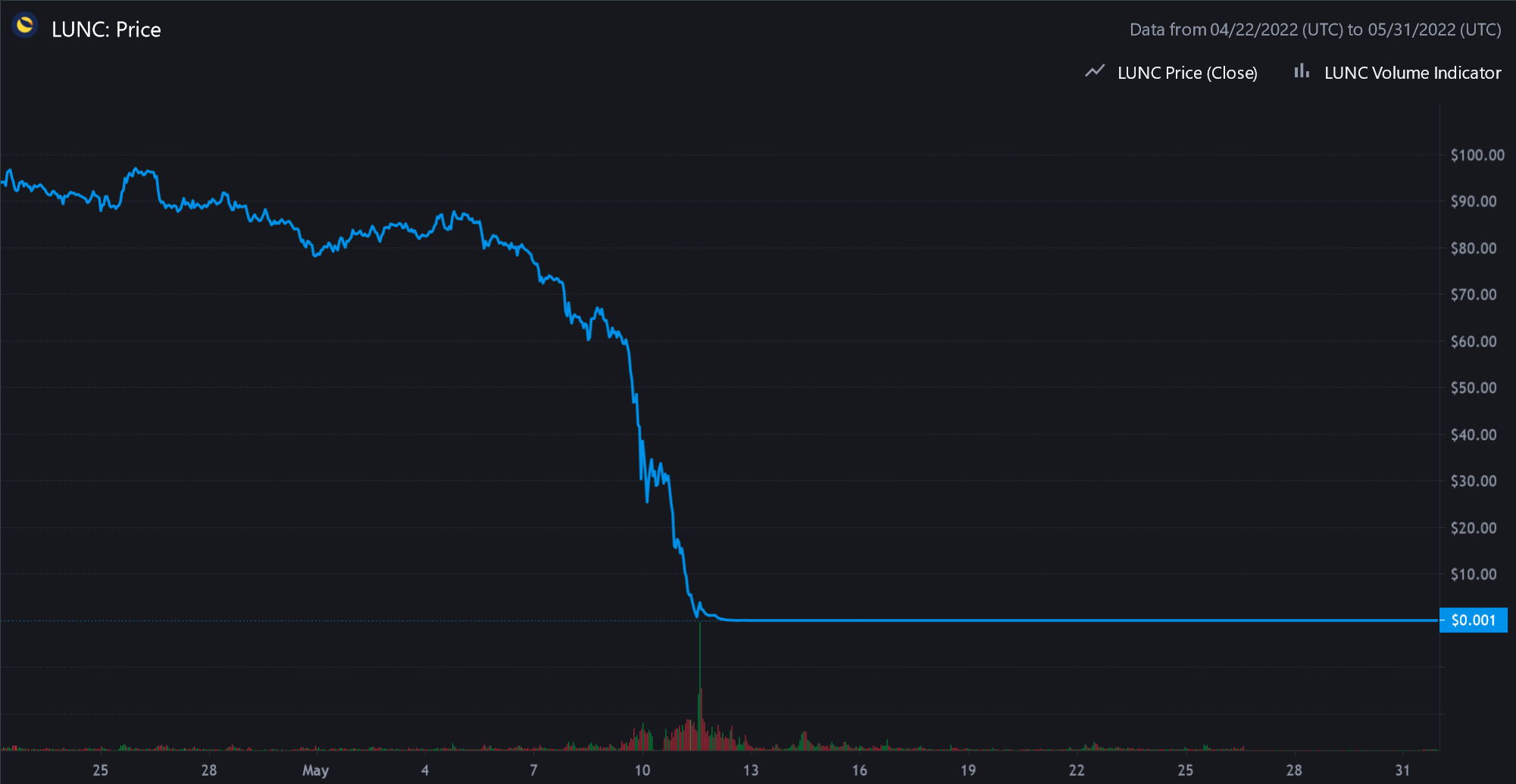}
		\caption{}\label{fig:LUNA}
	\end{subfigure}
	\caption{UST case study. (a) following a downwards spiral, UST price near \$0. (b) LUNA price going to zero at the same time.}\label{fig:ust-study}
\end{figure*}

%\paragraph{Iron/Seigniorage Shares}
%The Iron stablecoin was intended to be $\sim74\%$ backed by USDC and the remaining portion backed by the seigniorage shares-style token, Titan. It was redeemable proportionally for a basket of both, inflating the Titan supply to fulfill the Titan portion of redemptions. The redemption curve was intended to be flat at \$1. However, since Titan was a volatile and endogenous/circularly priced asset, there was no guarantee that its market cap would be enough to support its share of the Iron backing. In fact, a downwards spiral was triggered, in which Iron would see mass redemptions, and the value of Titan would go to zero, leaving the redemption price of Iron at the $\sim\$0.74$ worth of USDC. This is a variation on the stylized redemption curve in Figure~\ref{fig:stylized-curves2} and is precisely what happened in June 2021 as shown in Figure~\ref{fig:iron-study}.

%%%%%%%%%%%%%%%%%%%%%%%%%%%%%%%%%%%%%%%%%%%%%%%%%%%%%%%%%

\section{Desiderata for P-AMM Design}\label{sec:desiderata}

A major missing piece in the current space is the rigorous design of stablecoin primary markets developed from first principles.
As we've seen in the previous case studies, the primary market design plays a large role in the stability of these systems.
Up until this paper, it is not well-specified what the desirable properties are in designing these primary markets.
We first tackle this issue before continuing on the design of our
%the Gyroscope
primary market automated market maker (P-AMM) redemption curves in the remainder of this paper.

We strive for several desiderata in the design of a primary market mechanism that has good properties from stability and usability perspectives. We separate these into properties of the P-AMM within a block and more general properties.

Within a block, the following properties are desirable.
\begin{enumerate}
	\item The relative collateralization (i.e., reserve ratio) of the protocol is guaranteed to stay above a lower bound (unless this is impossible because of an exogenous shock to reserve assets).
	
	\item The P-AMM normally maintains a region of open market operations in which the stablecoin price is $\sim \$1$.
	
	\item The worst-case P-AMM redemption rate is lower bounded.
	
	\item The P-AMM redemption curve is continuous and not too steep, unless this would violate the other desiderata.
	
	\item There is no incentive for redeemers to strategically subdivide redemptions.
\end{enumerate}

These desiderata are informed by results about currency peg models, in which the shape of optimal monetary policies involve peg support up to a threshold and reserve preservation past that.
The first property means that the loss for the protocol is bounded, as the reserve is never exhausted in the operation of the P-AMM unless all reserve asset prices exogenously go to zero.
The second property means that the stablecoin can support a possible equilibrium price at the dollar target.
The third property means that the loss for stablecoin holders who redeem is bounded.

The fourth property reduces risk for traders and incentives for potential speculative attacks. Supposing the peg is not maintainable, the system can't know what the equilibrium price of the stablecoin will be. It only knows upper and lower bounds on it, considering the target and the level of reserves. A continuous P-AMM enables the new equilibrium price to be found along the curve and helps smooth flows around it.

The fourth and fifth properties are also motivated by usability: it is simpler for traders to use when the pricing is continuous and predictable and strategy in optimal order reporting is simple and minimized.
As we will see, the fifth property is related to notions of path independence and path deficiency.

More generally, and across blocks, we desire the following properties.
\begin{enumerate}\setcounter{enumi}{5}
	\item Over many blocks, the reserve can only be exhausted over a long time period.
	
	\item Over many blocks, a de-pegged stablecoin has a path toward regaining peg.
	
	\item The P-AMM mechanism can be efficiently implemented and computed on-chain.
\end{enumerate}

The sixth property ensures that the stablecoin's asset backing persists well into the future (e.g., speculative attacks on the system cannot exhaust the reserve).
The seventh property means that there are credible reasons why speculators could decide to bet on re-pegging of the stablecoin during a crisis.
The eighth property ensures that the mechanism could be used under the severe computational restrictions of real world smart contract systems (e.g., is not prohibitive in terms of gas fees on Ethereum).

%%%%%%%%%%%%%%%%%%%%%%%%%%%%%%%%%%%%%%%%%%%%%%%%%%%%%%%%%

\section{Redemption Curve Design}\label{sec:design}

We now discuss our P-AMM design. We first introduce the dynamical system framework within which our design will take place.
We then describe a greatly simplified redemption curve design that satisfies all desiderata but one: its redemption curve is not continuous but has a steep discontinuity. Finally, we introduce our actual, more sophisticated continuous redemption curve design.

\subsection{Dynamical System, Anchor Point}

For the purpose of designing the P-AMM, we model the stablecoin system along three dimensions: an outstanding stablecoin (SC) supply $y$, a total reserve value backing the stablecoin $b$, and a level of stablecoin redemptions from the reserve $x$. These state variables are summarized in the following table.
\begin{center}
	\begin{tabular}{c|l}
		\textbf{State Variable}	&	\textbf{Definition} \\
		\hline
		$b$			&	total reserve value (in USD) \\
		$y$			&	outstanding SC supply \\
		$x$			&	level of SC redemptions
	\end{tabular}
\end{center}
We model the system as a dynamical system, in which $x$ is the independent variable that drives the system. Put another way, $x$ will represent the ``current point along the trading path'' of  the P-AMM. We will also be interested in the \emph{reserve ratio}
$r(x) := b(x)/y(x)$, which describes the reserve value per outstanding stablecoin.\footnote{
	We will write $b,y,x$ and $r$ referring to the ``current' state of these variables. In contrast, we will write $b(x)$ etc for the value at some point of the driving variable $x$ and based on other system parameters.
}
Observe that $y(x) = y_a - x$ and $b(x) = b_a - \int_0^x p(x') \mathrm{d}x'$, where $p(x')$ is the marginal redemption price offered by the P-AMM at redemption level $x'$.

The dynamical system models P-AMM trades that occur within a single block. At the beginning of the block, we will have initial conditions $(x_0, b_0, y_0)$. Here, $x_0$ represents a measure of redemption history in previous blocks. Net redemptions within the modeled block will increase $x$ from $x_0$. The final P-AMM will evolve over many blocks using this same intra-block model; however, $x_0$ at the start of each block will be computed as an exponentially time-discounted sum over all past SC redemptions in previous blocks. For our analysis in this paper, we restrict ourselves to the context of a single block, in which $x_0$ is a fixed initial condition. The initial conditions are summarized in the following table.
\begin{center}
	\begin{tabular}{c|l}
		\textbf{Init. Condition}	&	\textbf{Definition} \\
		\hline
		$x_0$			&	level of SC redemptions at block start \\
		$b_0$			&	reserve value at block start ($b_0=b(x_0)$) \\
		$y_0$			&	SC supply at block start ($y_0=y(x_0)$)
	\end{tabular}
\end{center}

Note that we will generlaly not have $x_0 = 0$ in practice. However, it will be useful to reference a fictitious initial condition that would describe a starting point of 0. We call this the \emph{anchor point}, which is formally the triplet
$$(0, b_a, y_a),$$
where $b_a = b(0)$ and $y_a = y(0)$.
Many times, we will be interested in the reserve ratio at the anchor point $r_a = b_a/y_a$.
Figure~\ref{fig:rr_curves} visualizes what the reserve ratio curves will look like as a function of $x$ for various values of $r_a$. As we will see, each curve will have a unique anchor point $r_a$, which corresponds to the starting point of the curve.

\begin{figure}
	\centering
	\includegraphics[width=9cm]{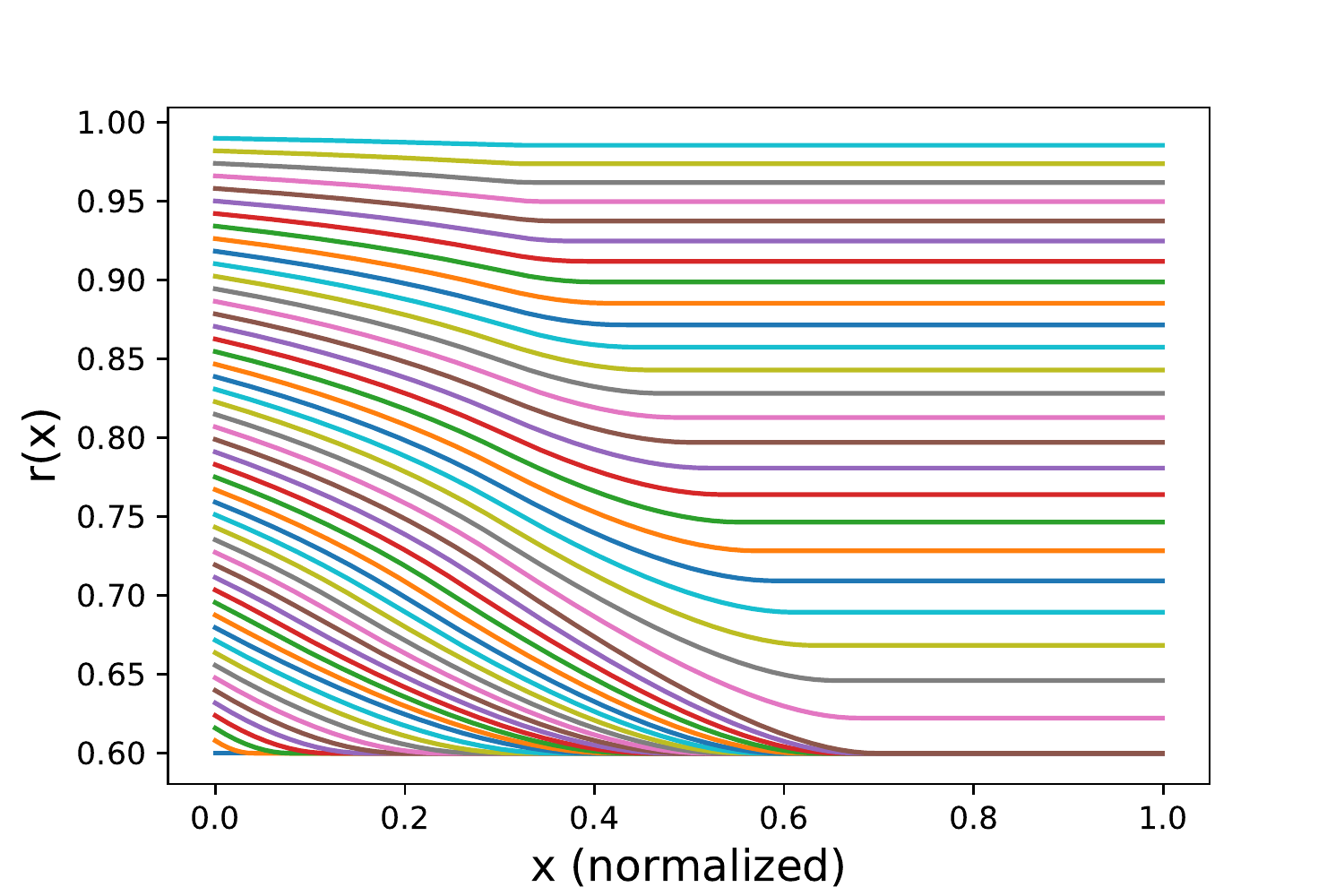}
	\caption{Reserve ratio curves as a function of $x$ for different values of $r_a$ (starting points) in the “normalized” case where $y_a=1$.}\label{fig:rr_curves}
\end{figure}

The pricing curve $p$ is, in general, a function of $b_a$ and $y_a$.
We will see later that the anchor point can be expressed in terms of the current state alone. With this in mind, it will be analytically useful to define the evolution of the dynamical system in terms of the current state $(x,b,y)$ directly.
Toward this, we will construct an abstract pricing function $\rho(x,b,y)$ that we will show is equivalent to the function $p$ in case $b/y < 1$; in the (trivial) case where $b/y \ge 1$, we set $\rho(x, b, y)=1$.
The dynamical system is then described by the following system of ordinary differential equations:
\begin{equation}\label{eq:odes}
	\begin{aligned}
		\frac{\de b(x)}{\de x} &= -\rho\big(x, b(x), y(x)\big) \\
		\frac{\de y(x)}{\de x} &= -1.
	\end{aligned}
\end{equation}

Our P-AMM, conceptually speaking, solves the initial value problem defined by $(x_0, b_0, y_0)$ and the system \eqref{eq:odes}. We then transition to the new state $(x_0 + X, b(x_0 + X), y(x_0 + X)) = (x_0 + X, b(x_0 + X), y_0 - X)$ and the redemption amount (which is paid out to the redeemer) is $b(x_0 + X) - b_0$.

\subsection{Simplified, Discontinuous Redemption Curve Design}
\label{sec:discrete-decay}

We now discus a simplified P-AMM redemption curve as a pedagogical starting point. This simplified curve has discrete price decay (i.e., the curve is discontinuous: a portion of the curve is at \$1 and another portion maintains a sustainable reserve ratio) and is very simple to reason about. This will fullfil many desirable properties except for continuity. We provide a semi-formal treatment here; see Appendix~\ref{apx:discrete-decay} for the formal details.

To define the redemption curve, we assume that the anchor point $(b_a, y_a)$ is fixed and given. We will discuss in Sections \ref{sec:reconstruction-uniqueness} and \ref{sec:implementation} how the anchor point can be chosen based on the current state of the system. Assume WLOG $b_a < y_a$ (i.e., $r_a < 1$) since otherwise, we always use the trivial redemption curve $p(x)=1$. For simplicity of analysis, we will for now disregard any trading fees that may be added to the P-AMM.

\begin{figure}
	\centering
	\includegraphics[width=0.6\textwidth]{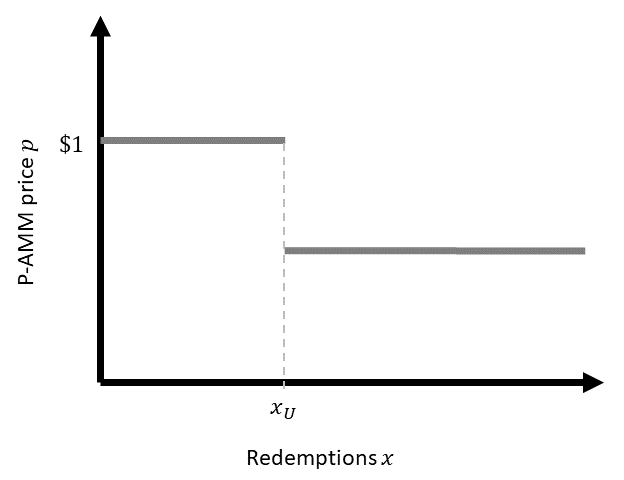}
	\caption{Simplified redemption curve with discrete price decay.}
	\label{fig:pamm_discrete_decay}
\end{figure}

Our simplified redemption curve is illustrated in Figure~\ref{fig:pamm_discrete_decay} and defined by the following function:
\[
	p(x) :=
	\begin{cases}
		1 &\text{if } x \le x_U
		\\
		r_U := r(x_U) &\text{if } x > x_U,
	\end{cases}
\]
Observe that $r(x_U)$ can be computed from only $b_a$, $y_a$, and $x_U$, so that $p(x)$ is a well-defined function.

This redemption curve maintains a redemption price of 1 up until some redemption pressure value $x_U \in [0, y_a]$ and then drops the redemption price to the reserve ratio below that value. From this point onwards, the reserve ratio will stay constant (this is easy to see) and in particular, we will be able to provide redemptions at price $r_U$, potentially until the whole SC supply has been redeemed.

The parameter $x_U$ is a \emph{dynamic parameter}. Concretely, this means that, given a specific value of this parameter, the shape of the redemption curve is fixed. However, the dynamic parameters itself depends on the state of the system. More in detail, it is a function of the anchor point $(b_a, y_a)$.

The choice of the parameter $x_U$ introduces a trade-off: when it is high, we maintain a price of 1 for a long time (in terms of redemption pressure). This allows the stablecoin to maintain the peg for longer, but the eventual redemption price $r_U$ will be low. Lower $x_U$ values enable a higher eventual redemption price but weaken the peg more quickly. To weigh this trade-off, we assume that two \emph{static parameters} are set externally (e.g., by the protocol's governance system):
\begin{center}
	\begin{tabular}{c|l}
		\textbf{Parameter}	&	\textbf{Definition} \\
		\hline
		$\bar x_U$			&	$\in [0,\infty]$ upper bound on $x_U$ ($x_U \leq \bar x_U$) \\
		$\bar\theta$		&	$\in [0,1]$ target reserve ratio floor
	\end{tabular}
\end{center}

$x_U$ is then chosen such that the following conditions hold:
\begin{enumerate}
	\item At any point $x$ of the redemption curve, $r(x) \ge \bar\theta$, if this is possible, and $r(x)$ maximal otherwise. Equivalently, we want $r_U \ge \bar\theta$ if possible and $r_U$ maximal otherwise.
	\item $x_U \le \bar x_U$.
	\item Among the  $x_U$ values that satisfy these conditions, $x_U$ is chosen maximal.
\end{enumerate}
One can show that this is the case iff
\[
	x_U = \min\left(\bar x_U,\,\max\left(0,\, y_a \frac {r_a - \btheta} {1-\btheta}\right)\right)
\]

The static parameter $\bar x_U$ is optional and can be set to $\infty$ to turn off this feature. In this case, the choice of $x_U$ will always create a situation where $r_U=\min(r_a,\bar\theta)$, i.e., we always choose the lowest acceptable eventual reserve ratio to maximize $x_U$. Setting $\bar x_U < \infty$ makes the trade-off less extreme and allows $r_U > \bar\theta$ when $b_a$ is large.

It is easy to see that the redemption curve design from this section satisfies our desiderata 1.--3. from Section~\ref{sec:desiderata}, but it obviously violates desideratum~4 (continuity). This is undesirable for traders and it creates a potential target for speculative attacks on the stablecoin price. This is why in the following, we discuss a continuous redemption curve design.

\subsection{Continuous Redemption Curve}

\begin{figure}
	\centering
	\includegraphics[width=0.6\textwidth]{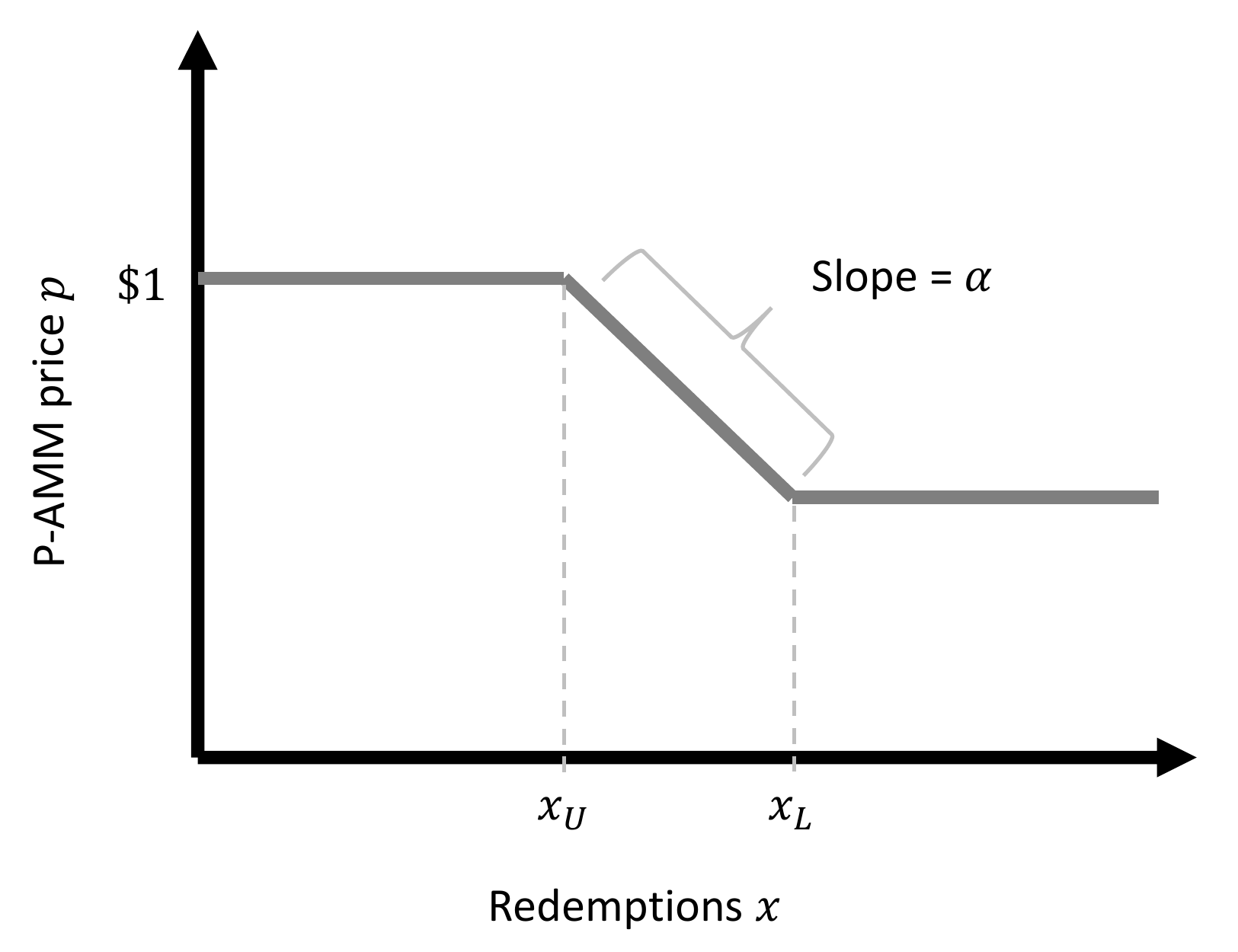}
	\caption{Our redemption curve design with piecewise-linear price decay.}\label{fig:pamm_linear_decay}
\end{figure}

To maintain continuity, we alter the design from the previous sub-section to receive a three-piece-wise curve design. This requires some more sophisticated machinery. Again, we assume in this section that an anchor point $(b_a, y_a)$ is given.

We now parameterize the curve in terms of \emph{three} dynamic parameters.
Conceptually, the three dynamic parameters describe three regions of the P-AMM pricing curve $p$ as a function of $x$, as visualized in Figure~\ref{fig:pamm_linear_decay}.
For a small redemption level $x$, the P-AMM provides redemptions at \$1 up until a redemption level $x_U$ is reached. In a second region, the P-AMM pricing decays linearly with slope $\alpha$ as more redemptions occur up until a redemption level $x_L$. Then, in a third region, the P-AMM provides redemptions at the new reserve ratio (reserve value per outstanding stablecoin), which is fully sustainable for the entire remaining stablecoin supply. The dynamic parameters are summarized in the following table.
\begin{center}
	\begin{tabular}{c|l}
		\textbf{Dynamic Params}	&	\textbf{Definition} \\
		\hline
		$\alpha$		&	decay slope of redemption curve \\
		$x_U$			&	point at which redemption deviates from \$1 \\
		$x_L$			&	point at which redemption stops decaying \\
		& at new reserve ratio
	\end{tabular}
\end{center}
The resulting P-AMM pricing curve, as a function of $x$ and parameterized by the anchor point, is, in the case that $b_a < y_a$,
\begin{equation}\label{eq:p_curve}
p(x; b_a, y_a) =
\begin{cases}
	1, & x\leq x_U \\
	1 - \alpha(x-x_U), & x_U \leq x \leq x_L \\
	r_L, & x\geq x_L
\end{cases},
\end{equation}
where $r_L = r(x_L)$. In the other case that $b_a \geq y_a$, we will simply set $p(x) = 1$. Notice that the dynamic parameters $\alpha, x_U, x_L$ are, in fact, functions of the anchor point $(b_a, y_a)$. We discuss the rules by which the dynamic parameters are chosen in Section~\ref{sec:calc-params}.

We define three static parameters that constrain the shape of the curve and inform the choice of the dynamic parameters. These are the only parameters that are set externally. We define a lower bound $\bar\alpha$ to the linear decay slope $\alpha$, an upper bound $\bar x_U$ to $x_U$, and a target reserve ratio floor $\bar\theta$. The target reserve ratio floor is the minimum reserve ratio that the P-AMM curve can decay to, and it is the value of the reserve ratio in the third region. In case that the initial reserve ratio $b_0/y_0$ is smaller than $\btheta$, the P-AMM only offers redemptions at the initial reserve ratio (i.e., $x_L = 0$). These parameters are summarized in the following table.
\begin{center}
	\begin{tabular}{c|l}
		\textbf{Parameter}	&	\textbf{Definition} \\
		\hline
		$\bar\alpha$		&	$\in (0,\infty)$ lower bound on decay slope ($\alpha \geq \bar \alpha$) \\
		$\bar x_U$			&	$\in [0,\infty]$ upper bound on $x_U$ ($x_U \leq \bar x_U$) \\
		$\bar\theta$		&	$\in [0,1]$ target reserve ratio floor
	\end{tabular}
\end{center}

Note that an implicit fourth static parameter is the target for the stablecoin price, thus far assumed to be \$1. In general, this could take different values (and could be changed over time by governance) to adjust monetary policy. The underlying mechanics and our essential results would stay the same.
In this paper, we are agnostic to how the static parameters were set and consider them fixed, but arbitrary.
One particularly useful choice is to set them proportional to the anchored outstanding SC supply $y_a$, which we believe would minimize the need for governance interaction as the outstanding amount changes. We explore this option further in Section~\ref{sec:implementation}.

\section{Calculation of Dynamic Parameters for a given anchor point}
\label{sec:calc-params}

We first establish how to calculate the dynamic parameters of the redemption curve when the anchor point $(b_a, y_a)$ is given.
Recall that this anchor point is a mathematical modelling tool and is not known as a real quantity at any point in time. Instead, it will be “reconstructed” from the current market state (see Section~\ref{sec:reconstruction-uniqueness} below).
%%%
%We now define the rules by which the dynamic parameters are chosen based on an anchor point $(b_a, y_a)$.
Assume for non-triviality that $1 > b_a/y_a > \btheta$.
Then the dynamic parameters are chosen as follows.

\begin{itemize}
	\item $x_L$ is chosen such that $p(x_L) = r(x_L)$, i.e., the computed redemption price (in the linear segment) equals the reserve ratio at this point. If $x_L$ is chosen like this, all remaining stablecoin units could be redeemed at this price without running out of reserves. $x_L$ is a function of $(b_a, y_a)$, $x_U$, and $\alpha$. Note that $r(x_L)$ is the lowest reserve ratio on the curve.\footnote{%
		One can show that $r(x_L)$ increases when $\alpha$ is increased (i.e., the reserve is protected when redemption prices decay more steeply) and when $x_U$ is decreased (i.e., the reserve is protected when the redemption price starts to decay earlier).
	}
	\item $\alpha$ and $x_U$ are chosen to guarantee that $r(x_L) \ge \btheta$ while minimizing price decay.
	\item Among the $(\alpha, x_U)$ pairs satisfying the previous condition, we prioritize making $\alpha$ as small as possible (but at least $\balpha$), i.e., making price decay as mild as possible in the linear part. If $\alpha=\balpha$, we then choose $x_U$ as large as possible (but at most $\bzp$) while always ensuring $r(x_L) \ge \btheta$.
\end{itemize}

This leads to the following equations (see Appendix~\ref{apx:calc-params} for details; let $\Delta_a := y_a - b_a$):
\begin{align*}
	\alpha &= \max(\balpha, \halpha)\ \text{where}\\
	&\phantom{={}}\halpha =
\begin{cases}
	\halpha_{H}:=2\frac{1-r_{a}}{y_{a}}, & r_{a}\ge\frac{1+\btheta}{2}\\
	\halpha_{L}:=\frac{1}{2}\frac{\itheta^{2}}{b_{a}-\btheta y_{a}}, & r_{a}\le\frac{1+\btheta}{2}
\end{cases}\\[5pt]
	x_U &= \min(\bzp, \hzp)\ \text{where}\\
	&\phantom{={}}\hzp =
\begin{cases}
	\hzph:=y_{a}-\sqrt{2\frac{\Delta_{a}}{\alpha}}, & \alpha\Delta_{a}\le\frac{1}{2}\itheta^{2}\\
	\hzpl:=y_{a}-\frac{\Delta_{a}}{\itheta}-\frac{1}{2\alpha}\itheta & \alpha\Delta_{a}\ge\frac{1}{2}\itheta^{2}.
\end{cases}\\[5pt]
	x_{L}	&= y_{a}-\sqrt{(y_{a}-x_{U})^{2}-\frac{2}{\alpha}\Delta_a} \\
	r_{L}	&= 1-\alpha(x_{L}-x_{U})
\end{align*}

%%%%%%%%%%%%%%%%%%%%%%%%%%%%%%%%%%%%%%%%%%%%%%%%%%%%%%%%%
\section{Uniqueness of Reconstruction}\label{sec:reconstruction-uniqueness}

We now move on to show that we can construct the anchor point $(b_a, y_a)$ uniquely from the current state $(x, b, y)$, which proves that our dynamical system \eqref{eq:odes} is in fact well-defined. Recall that, in that dynamical system, $\rho$, which is a function solely of the current state, `reconstructs' $p$, which is also a function of the anchor point $(b_a, y_a)$.

Our proof in this section is indirect: we show that that each state $(x, b, y)$ can only have arisen from one specific anchor point (see Section~\ref{sec:implementation} for an explicit reconstruction technique).
Uniqueness of $y_a$ is trivial because $y = y_a - x$, so $y_a = y + x$; it remains to show that $b_a$ is unique.
We show that this is the case because, for fixed $x$, the reserve value
\[
b(x; b_a, y_a) = \int_0^x p(x'; b_a, y_a) \;\de x'
\]
is a strictly monotonic function of $b_a$, whenever that state is non-trivial.\footnote{%
	Equivalently, the reserve ratio $r(x)$ is strictly monotonic in $b_a$, since it is a strictly monotonic transformation of $b(x)$ for fixed $x$.
	In this section, we consider the static parameters $\btheta$, $\balpha$, $\bzp$ fixed while the dynamic parameters $\alpha$ and $x_U$ take on values dependent on the anchor point as discussed in Section~\ref{sec:calc-params}.
}

\begin{theorem}
	\label{thm:mon-b0} Fix values $y_{a}$, $\btheta$, $\balpha$,
	$\bzp$ and fix some $x\in[0,y_{a})$. Assume that $x_{U}$ and $\alpha$
	are chosen dependent on $b_{a}$ according to the rule described above. Let $b_{a}$ and $b_{a}'$ be such that $b_a < b_a'$ and
	$1 > r(x;b_{a}),\,r(x;b_{a}')>\btheta$ (and in particular $1 > b_{a}/y_{a},\,b_{a}'/y_{a}>\btheta)$.
	Then $b(x;b_{a})<b(x;b_{a}')$.
\end{theorem}

The result is less immediate than it may seem at first. While it is easy to see that a lower $b_a$ leads to a lower value of $b(x)$ \emph{when leaving the parameters $\alpha,x_U$ fixed}, the situation we consider here is different. In particular, as we reduce $b_a$, the parameters $\alpha$ and $x_U$ will adjust according to propositions \ref{prop:alpha} and \ref{prop:zp} to ensure our desiderata.
A priori, it might be the case that, through this adjustment, a lower value of $b_a$ leads to a lower value of $b(x)$ at some point $x$, but to a \emph{higher} or the same value of $b(x')$ at some other point $x'$.
Theorem~\ref{thm:mon-b0} proves that this is not the case: the whole reserve value curve $b(\,\cdot\,)$ is strictly monotonic in the anchor reserve value $b_a$ at all non-trivial points.

We segment the space of $(x, b_a)$ pairs along all case distinctions we have made so far. Specifically, we distinguish cases along the following
dimensions. Note that these cases define closed regions and overlap at their boundaries and that within the intersection of any selection of these regions, the function $b(\,\cdot\,)$ is smooth.
Recall that $r_a = y_a/b_a$ is the reserve ratio at the anchor point.
\begin{itemize}
	\item \textbf{Case I--III} depending on the values of $\alpha$ and $x_{U}$.
	Specifically, let $\alpha=\max(\halpha,\balpha)$, write short $\hzp:=\hzp(\alpha)$,
	and let $x_{U}=\min(\hzp,\bzp)$, and define:
	\begin{description}[labelwidth=\widthof{\bfseries III}]
%		\settowidth\labelwidth{\makelabel{III}}
		\item[I] $\halpha\ge\balpha$ and $\hzp\ge\bzp$, so that $\alpha=\balpha$
		and $x_{U}=\bzp$.
		\item[II] $\halpha\ge\balpha$ and $\hzp\le\bzp$, so that $\alpha=\balpha$
		and $x_{U}<\bzp$ except in the equality case.
		\item[III] $\halpha\le\balpha$, so that $\alpha<\balpha$ except in
		the equality case and (in any case) $x_{U}=0$.
	\end{description}
	\item Within case II, \textbf{case II h}, where $\alpha\Delta_{a}\le\frac{1}{2}\itheta^{2}$ and thus $\hzp=\hzph$ and \textbf{case II l}, where $\alpha\Delta_{a}\ge\frac{1}{2}\itheta^{2}$ and thus $\hzp=\hzpl$.
	\item Within case III, \textbf{case III H}, where $r_{a}\ge\frac{1+\btheta}{2}$ and thus $\halpha=\halpha_{H}$
	and \textbf{case III L}, where $r_{a}\le\frac{1+\btheta}{2}$ and thus
	$\halpha=\halpha_{L}$.
	\item \textbf{Case i--iii} depending on the value of $x$ relative to $x_{U}$
	and $x_{L}$. Specifically, define:
	\begin{description}[labelwidth=\widthof{\bfseries iii}]
		\item [{i}] $x\le x_{U}$
		\item [{ii}] $x_{U}\le x\le x_{L}$
		\item [{iii}] $x_{L}\le x$
	\end{description}
\end{itemize}

Note that only the distinction between case i--iii depends on $x$. In the following, we will address cases by a sequence of letters, such as case II h ii. Note that not all 36 potential combinations of these cases need to be addressed one-by-one. Many of these cases are irrelevant or easy to handle.
For instance, in case I, the values of the parameters are known to be the constants $(\balpha,\bzp)$ and thus we do not have to distinguish the H/L or h/l cases.
The following proposition helps us when distinguishing the H/L and h/l cases.
This will be useful for the proof of Theorem~\ref{thm:mon-b0}
Recall that all dynamic parameters we have defined so far, like $x_{L}$, are functions of $b_{a}$.

%The following proposition provides conditions under which the ultimate constant segment is degenerate because it either does not exist or it lies at the reserve ratio floor. This will be useful for the proof of Theorem~\ref{thm:mon-b0}.

\begin{proposition}
	\label{prop:HLrm}In case II h and III H, $x_{L}=y_{a}$. In case
	II l and III L, $r_{L}=\btheta$.
\end{proposition}

\begin{center} \hyperlink{pf:HLrm}{\texttt{[Link to Proof]}} \end{center}

The result allows us to exclude certain parts of the state space from the analysis because there, the recovery rate is at most our defined floor and thus the mechanism defines that redemption must happen at the reserve ratio.

The result immediately implies that we do not need to consider cases II iii or III iii because the recovery rate is at most our defined floor and thus the mechanism defines that redemption must happen at the reserve ratio.

\begin{corollary}
	\label{cor:iii}In case II iii and III iii, $r(x)\le\btheta$.
\end{corollary}

\begin{center} \hyperlink{pf:cor:iii}{\texttt{[Link to Proof]}} \end{center}

\begin{figure*}
	% todo these figures are mildly ugly. Review color scheme, what else we can do, and how to get the background color to fill the whole thing.
	\centering
	\begin{subfigure}[b]{.51\textwidth}
		\includegraphics[width=\linewidth]{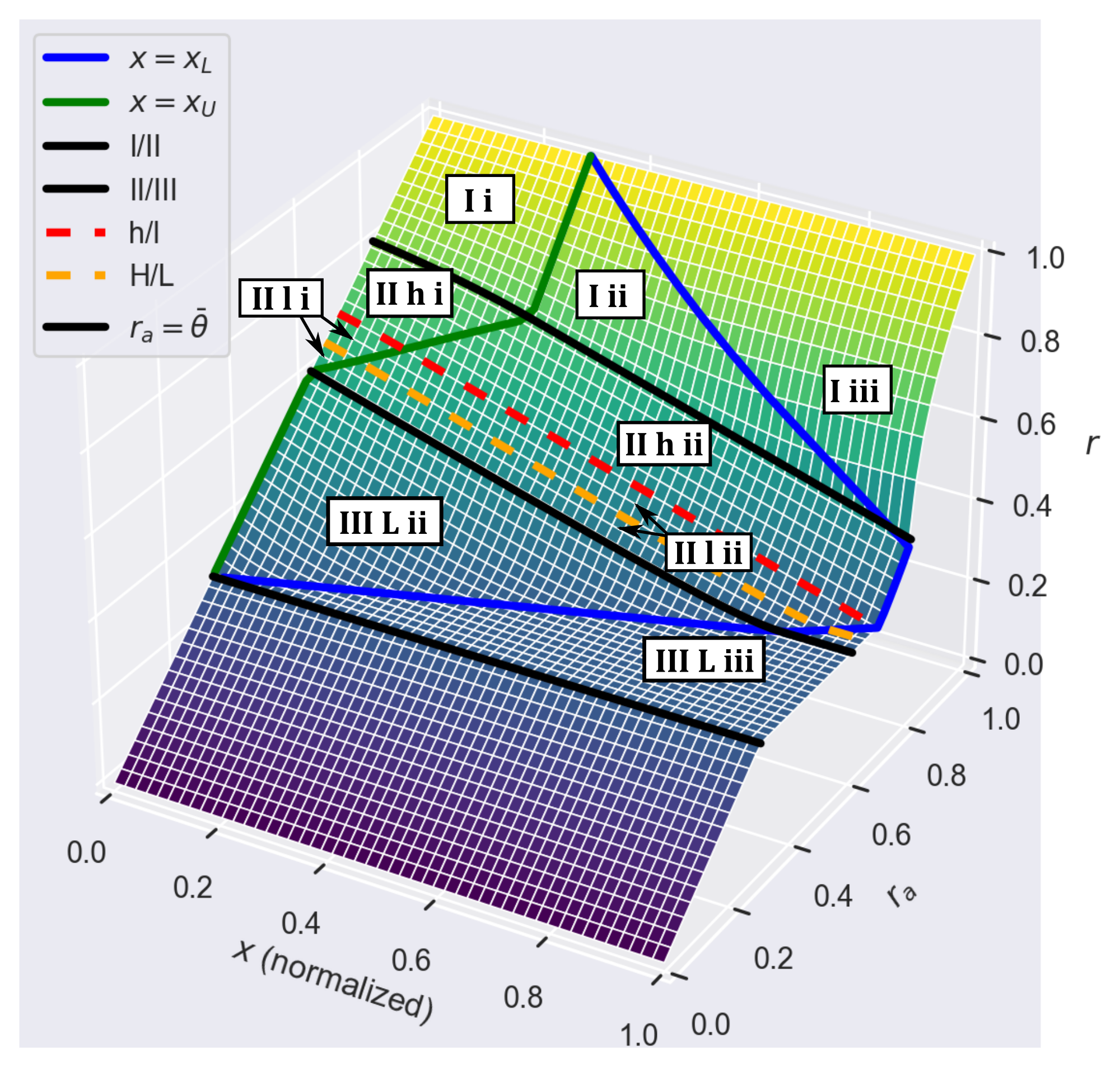}
		\caption{3d curve.}
		\label{fig:regions-1-3d}
	\end{subfigure}\hfill
	\begin{subfigure}[b]{.47\textwidth}
		\includegraphics[width=\linewidth]{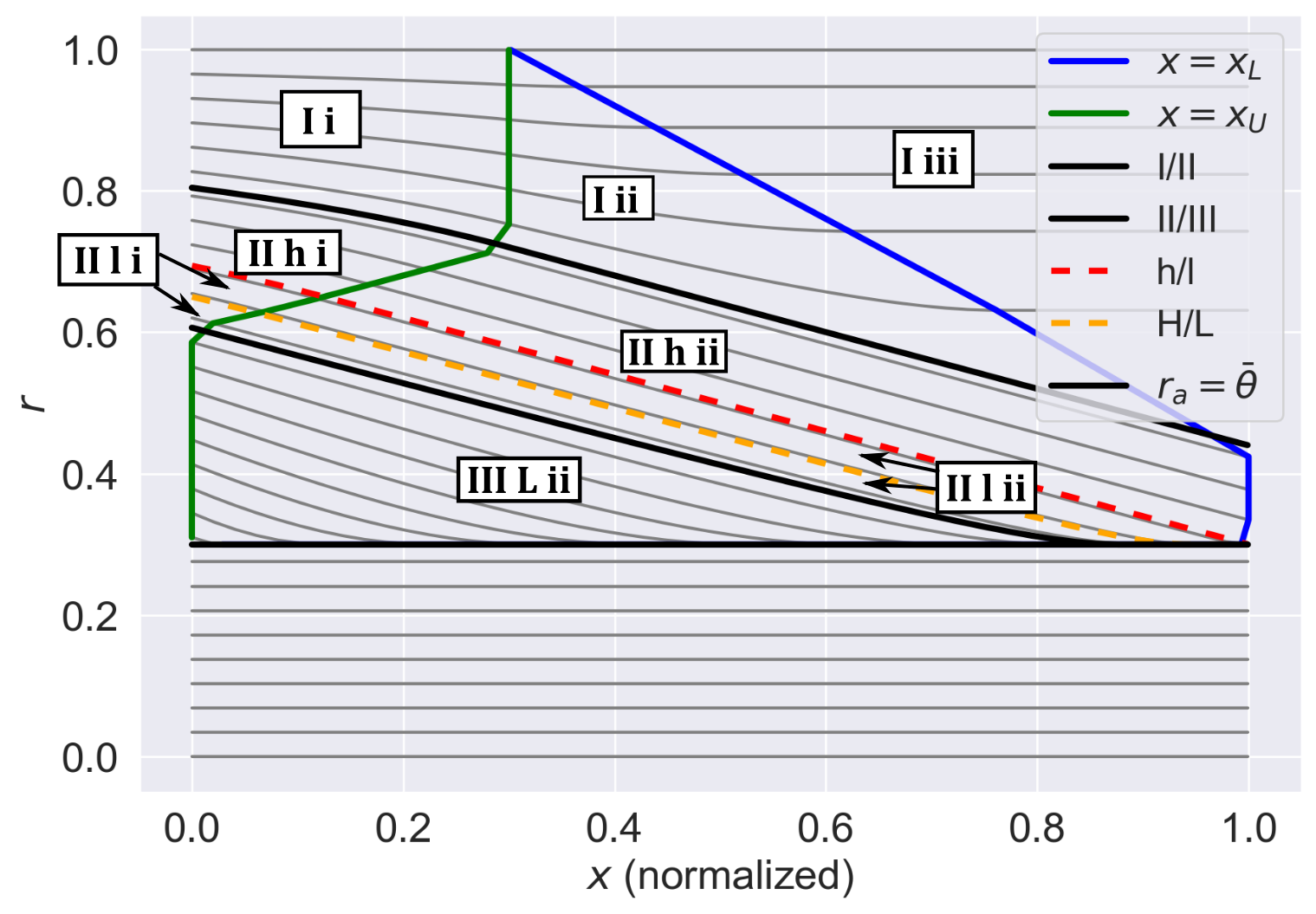}
		\caption{2d projection to the $(x, r)$ plane.}
		\label{fig:regions-1-2d}
	\end{subfigure}
	\caption{Reserve ratio $r$ as a function of the current redemption amount $x$ and the initial reserve ratio $r_a$ for the normalized case $y_a=1$. Due to normalization, we have $r_a=b_a$. The parameters $\btheta = 0.3$, $\balpha=0.8$, and $\bzp=0.3$ were used. The figures also depict $x_L$ and $x_U$ on the X axis as functions of $r_a$}
	\label{fig:regions-1}
\end{figure*}

  Conceptually, we can visualize the function $b(x;b_a)$, with which Theorem~\ref{thm:mon-b0} is concerned, as a 3d surface.
Figure~\ref{fig:regions-1-3d} shows the surface of the reserve ratio as a function of $x$ and $b_a$, where the different regions partition the $(x, b_a)$ space.\footnote{
      Note that $b(x;b_a) = r(x;b_a) \cdot (y_a-x)$, so that there is a simple 1:1 relationship between the reserve value and the reserve ratio. We present the reserve ratio because we find it more illustrative of the different phenomena.}
    We have chosen $y_a=1$ so that $b_a=r_a$ (but in general $b \neq r$). Notice that, as mentioned above, several of the regions can be essentially ignored.
This is for two reasons:
first, some of the lines separating different cases fall into a part of the $(x, r_a)$ space where the respective case is not relevant. For instance, in this example, the line separating the H and L cases itself falls into region~II, so that only case III L is present. In contrast, the line separating cases h and l falls into region II, so that both cases II h and II l are present.
The second reason why distinctions between regions can be ignored is that they fall along a certain flat region of the surface that has zero area in a projection of interest. This projection of interest is the 2-d $(x,r)$ space. This is visualized in Figure~\ref{fig:regions-1-2d}, in which many $r$ curves are plotted (in gray) for differing values of $r_a$, which are the starting points of these curves. In this space, the region boundaries are given by particular cases of $r$ curves, as shown. Notice that the flat section of the 3-d surface mentioned above
disappears in the 2d-projection.
% Steffen re-wrote this paragraph to better match the two figures.

\begin{figure*}
	\centering
	\begin{subfigure}{.49\textwidth}
		\includegraphics[width=\linewidth+5pt]{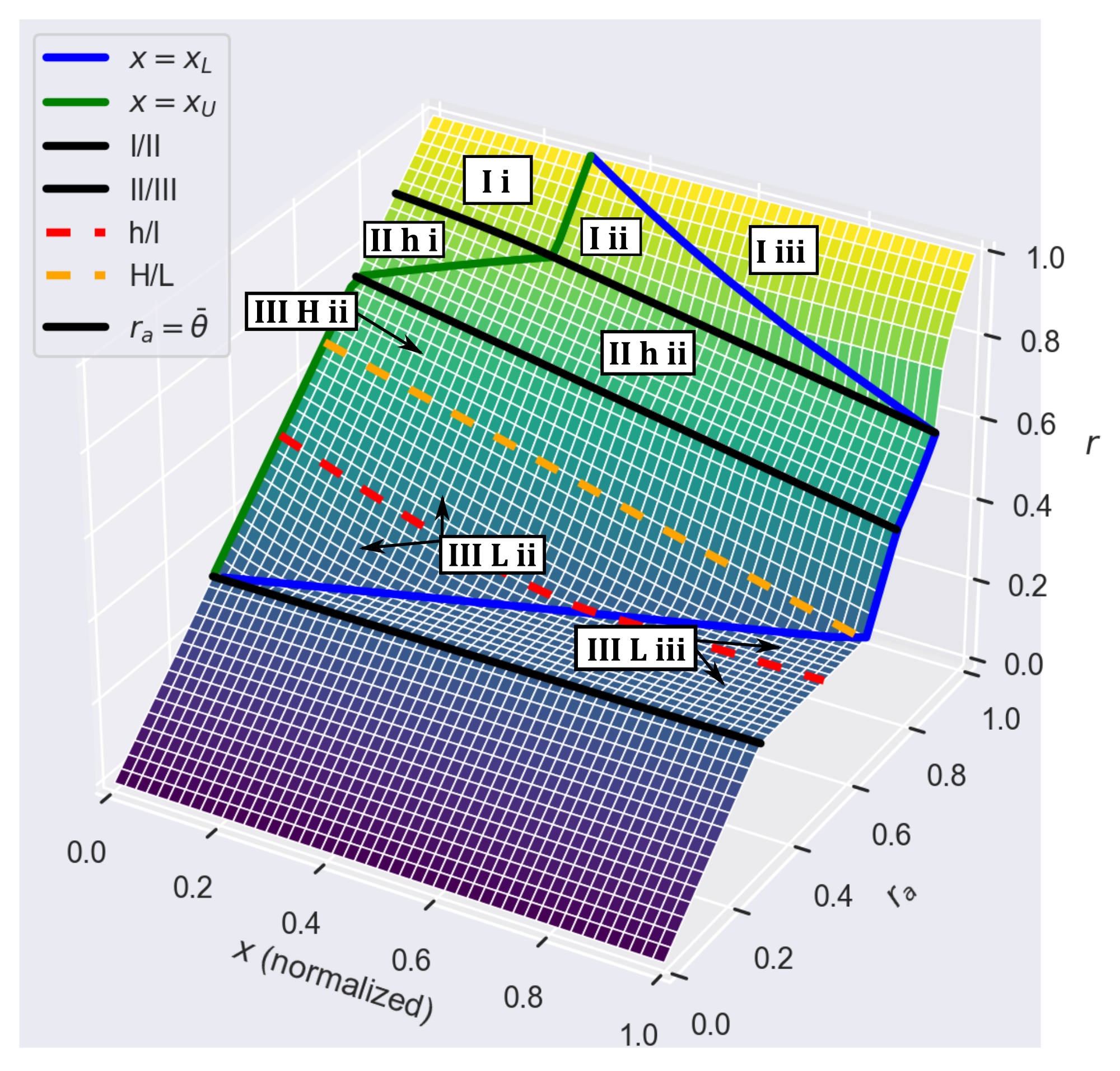}% +5pt = evil hack b/c the pic is a bit broken.
		\caption{$\btheta=0.3$, $\balpha=0.5$, $\bzp=0.3$. There is a stretch where $x_L=1$.}
		\label{fig:regions-2-05}
	\end{subfigure}\hfill
	\begin{subfigure}{.49\textwidth}
		\includegraphics[width=\linewidth+5pt]{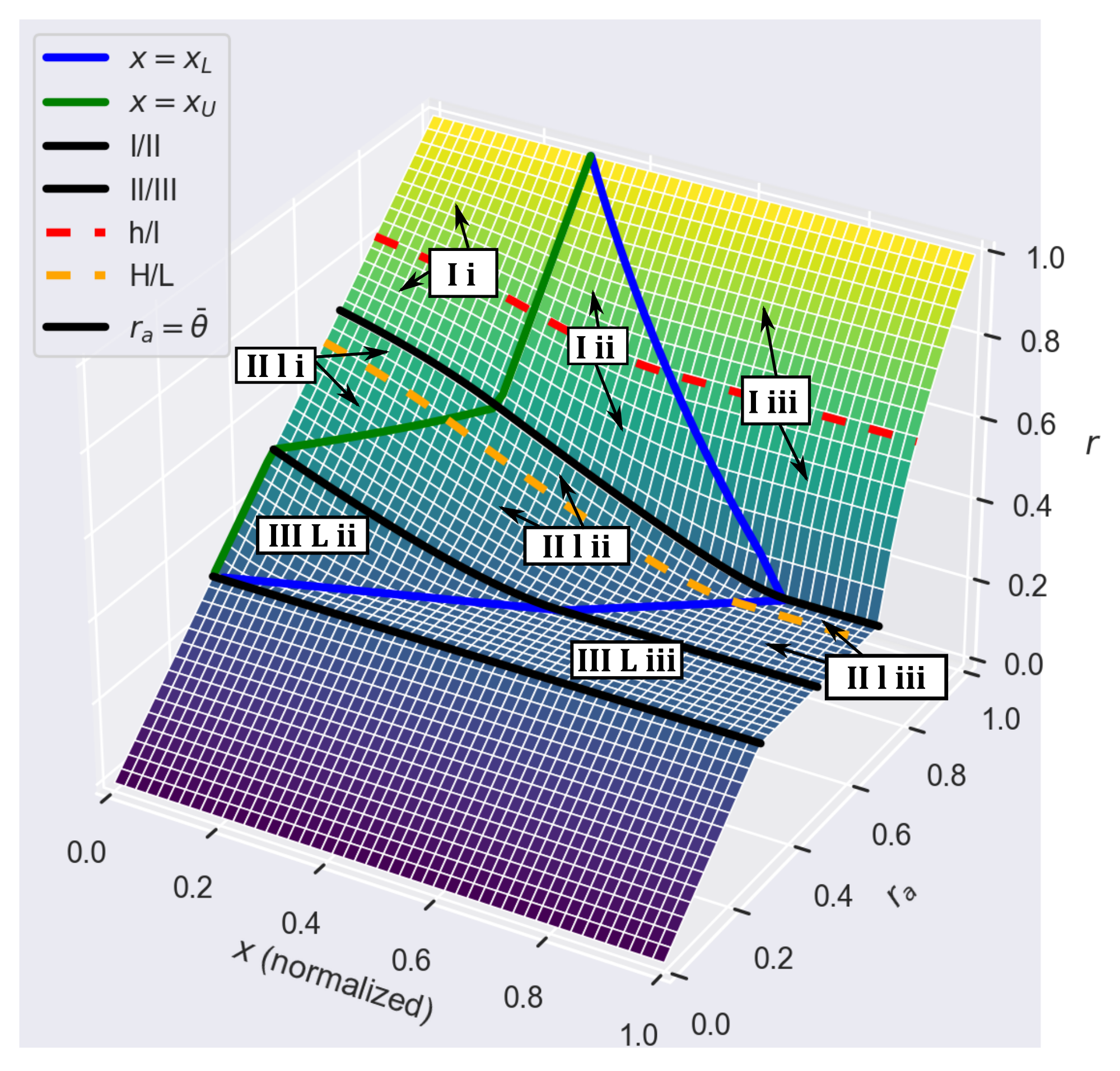}
		\caption{$\btheta=0.3$, $\balpha=1.3$, $\bzp=0.3$. There is no point where $x_L=1$.}
		\label{fig:regions-2-13}
	\end{subfigure}
	\caption{Reserve ratio $r$ as a function of the current redemption amount $x$ and the initial reserve ratio $r_a$ for $y_a=1$ and two choices of parameters.}
	\label{fig:regions-23}
\end{figure*}

In Figure~\ref{fig:regions-1}, region III H does not exist and there is a stretch of $r_a$ values where $x_L=1$. This, however, is not universally the case and depends on the parameters. We illustrate two 3d curves $r(x; r_a)$ where this is and is not the case, respectively, in Figure~\ref{fig:regions-23}.

Using the region partition, we can prove Theorem~\ref{thm:mon-b0}.
Our proof is by case distinction over the different regions. Observe that the function $b(x, b_a)$ is smooth across any given region and we can compute the partial derivative of $b(x, b_a)$ with respect to $b_a$ explicitly.
The main challenge is to handle the adjustment in the dynamic parameters that takes place as $b_a$ is varied.
\begin{center} \hyperlink{pf:mon-b0}{\texttt{[Link to Proof of Theorem~\ref{thm:mon-b0}]}} \end{center}

%%%%%%%%%%%%%%%%%%%%%%%%%%%%%%%%%%%%%%%%%%%%%%%%%%%%%%%%%
\section{Path Properties}\label{sec:path-properties}

Now that we've established that the P-AMM design is well-defined and robust in shape, we move on to show that it obeys many useful trading properties, including in settings involving trading fees and a separate minting curve. We characterize these in terms of path independence and path deficiency, which will lead us to two main useful properties:
\begin{itemize}
	\item There is no incentive for redeemers to strategically subdivide redemptions, including in some settings with trading fees.
	
	\item In a wide array of settings involving a P-AMM with minting, redeeming, and trading fees, the protocol itself is only better off in terms of the reserve ratio curve no matter which trading path is realized.
\end{itemize}

%%%%%%%%%%%%%%%%%%%%%%%%%%%%%%%%%%%%%%%%%%%%%%%%%%%%%%%%%
\subsection{Path Independence}

We first show that the P-AMM redemption curve, as developed thus far without trading fees, is path independent within a block. This means that the end result of any path of redemptions is the same for any given starting and ending point. On the ground, this is useful for traders as they do not need to worry about how exactly they use the P-AMM within a block: there is no incentive to split up redemptions into smaller chunks or to merge many small redemptions into bigger units.

\begin{theorem} \label{thm:path_ind}
	(Path independence) Let $S_0 := (x_0, b_0, y_0)$ be a state and let $X,Y>0$ such that $X+Y\leq y_0$. Let the state $S_X$ result from redeeming $X$ at $S_0$, $S_{X,Y}$ from redeeming $Y$ at $S_X$, and $S_{X+Y}$ from redeeming $X+Y$ at $S_0$. Let the paid-out amounts be $P_X, P_{X,Y}$, and $P_{X+Y}$ respectively. Then the following hold.
	\begin{enumerate}
		\item $P_X + P_{X,Y} = P_{X+Y}$
		\item $S_{X,Y} = S_{X+Y}$
	\end{enumerate}
\end{theorem}

\begin{center} \hyperlink{pf:path_ind}{\texttt{[Link to Proof]}} \end{center}

The theorem holds because the right-hand sides of (\ref{eq:odes}) only depend on the current state of the system and not, for instance, on the initial conditions or the specific trading path that led to the current state. In the dynamical systems literature, this property is also known as invariance under horizontal translations.

Note that the theorem only holds for redemptions immediately following each other within the same block ---the context of our model in this paper--- as otherwise we need to consider the time decay of the initial condition $x_0$ across blocks as well as exogenous changes to $b$ and minting operations happening between two redemptions. Additionally, if there are many traders using the P-AMM to redeem, there will at times be an incentive to be earlier in the redemption queue.
%s into bigger units.

\subsection{Path Deficiency Properties with Minting and Fees}

We next discuss two properties analogous to path deficiency in CFMMs \citep[see][]{angeris2020improved} that holds in an extended setting with fees and the possibility of minting tokens  in addition to redemption (where minting occurs at the same or a higher price than redemption).
In this section, we provide an informal discussion of these results. For a formal treatment, see Appendix~\ref{apx:path-properties}.

Our first result is that, for any non-negative fee, trading along any \emph{trading path} (i.e., any sequence of mint and redeem operations) can only improve system health compared to the net of the path. This property is important because it implies that the system is robust even under (and, in fact, benefits from) combinations of mint and redeem operations, which may result from market volatility. We capture “system health” by the reserve ratio $r(x)$ at a given redemption level $x$. This is a function of the anchor point $(b_a, y_a)$ and the static parameters $\bzp$, $\balpha$, $\btheta$. We assume in the following that these values are held fixed.

\begin{theorem}[Path Deficiency Vs.\ Netting, Informal]\label{thm:path-deficiency1-informal}
	Consider any trading path ending in redemption level $x$. Let $r_1$ be the reserve ratio of the system after following the trading path through its mint and redeem operations and let $r_0 := r(x)$ be the reserve ratio at the net minting or redemption amount. Then $r_1 \ge r_0$.
\end{theorem}

Our second result states that, in settings with a proportional fee, there is no incentive for a trader to strategically subdivide a net redemption trade into a sequence of different trades. Note that, clearly, this may not be true when fees are non-linear in the size of a trade.

\begin{theorem}[Path Deficiency Vs.\ Subdivision, Informal]\label{thm:path-deficiency2-informal}
	Assume that a non-negative fee is taken proportional to the redemption amount. Then the following hold:
	\begin{enumerate}
		\item The system is path independent with respect to combinations of only redemption (but not minting) operations even in the presence of fees.
		\item An individual trader has no incentive to subdivide a net redemption within a block into a combination of minting and redemption operations.
	\end{enumerate}
\end{theorem}

For traders, the previous theorem means that interacting with the mechanism is straightforward. For the mechanism, it means that the mechanism cannot be exploited by forming complex trading paths.

Note that there may still be strategic interaction between many traders, but these considerations are limited and fairly simple (and can in fact be avoided with a batch settlement of trades in the block, which is fundamentally possible if more difficult to implement).
There can be an incentive for a given trader to get a redemption trade in earlier than other redemption trades.
There can also be an incentive for a given trader to get a redemption trade in after any minting trades are settled. Note that in many circumstances where this would matter, we are not likely to see mint and redemption transactions in the same block, however, as a trader would likely get a better price for one or the other on a secondary market.
A further concern is whether arbitrage trades, in which the net redemption is zero, are profitable (e.g., a sandwich attack around other trades). This is not endogenously profitable from the P-AMM structure alone since redemption prices are $\leq 1$ while minting prices are $\geq 1$.

\section{Efficient Implementation}
\label{sec:implementation}

We now discuss how to implement our mechanism on-chain. This section provides a high-level overview. The formal details are presented in Appendix~\ref{apx:implementation}. Appendix~\ref{apx:discrete-decay:implementation} discusses the implementation of the simplified redemption curve with discrete price decay.

The P-AMM can be operationalized as follows: we are given the current state of the system $(x, b, y)$, a redemption amount $X$, and we need to compute a redemption amount $B$ that the trader should receive in exchange for $X$ units of the stablecoin.
To do this, the system constructs the anchor point $(b_a, y_a)$ consistent with the current state, constructs the redemption curve from this, and then integrates over the redemption curve to compute the redemption amount.

Most of the implementation operates in a space normalized to $y_a=1$. This is convenient because it allows the system to be configured in a way that is invariant under scaling of the stablecoin supply. We therefore assume that the static parameters refer to the \emph{normalized} space. Specifically, we assume that the following static parameters are given:
\begin{center}
	\begin{tabular}{c|l}
		\textbf{Parameter}	&	\textbf{Definition} \\
		\hline
		$\balphao$		&	$\in (0,\infty)$ lower bound on decay slope in normalized space \\
		$\bzpo$			&	$\in [0,1]$ upper bound on $x_U$ in normalized space \\
		$\bar\theta$		&	$\in [0,1]$ target reserve ratio floor
	\end{tabular}
\end{center}
Note that under the normalization, we assume that the \emph{anchor} stablecoin supply $y_a$ is 1, not the \emph{current} supply $y$. Note also that the target reserve ratio floor is not affected by normalization. In Appendix~\ref{apx:implementation}, we discuss the conversion between normalized and non-normalized space in detail.

The implementation of the algorithm now proceeds as follows (see Algorithm~\ref{alg:redemption} in Appendix~\ref{apx:implementation}).
Assume WLOG that $b/y\in (\btheta, 1)$ since otherwise, the behavior of the P-AMM is trivial.

First, we can easily reconstruct $y_a = y + x$. The algorithm then normalizes all values to $y_a=1$; specifically, we consider the \emph{normalized state} $(x/y_a, b/y_a, y/y_a)$ and at the end of the algorithm, we will undo the normalization by scaling the computed normalized redemption amount by $y_a$. To simplify the exposition, assume in this section WLOG that $y_a = 1$ so that the original state and the normalized state are the same and we don't need to introduce any additional variables.

The algorithm now reconstructs the anchor reserve value $b_a$ consistent with the current state. This is the most complex step of the algorithm.
Recall that $(b_a, y_a)$ together with the static parameters determine the dynamic parameters and thus the full redemption curve. We are looking for a $b_a$ such that
\[
b = b(x; b_a, y_a=1)
.
\]
It follows from continuity that such a $b_a$ must exist and Theorem~\ref{thm:mon-b0} implies that it is unique, but the theorem provides no way of computing it efficiently. To do this, we proceed in two steps:
\begin{enumerate}
	\item First, we determine the \emph{region} of $(x, b, y)$ among the regions introduced in Section~\ref{sec:reconstruction-uniqueness}. Recall that these regions depend on $b_a$ because they depend on the dynamic parameters $x_U$ and $\alpha$. The main insight that enables this step is that the region of a state can still be determined efficiently even if $b_a$ is not known. To enable this, the algorithm (pre-)computes certain $b_a$ values that separate different regions. We also show that $(b_a, y_a)$ can be replaced by $(b, y)$ when computing certain specific thresholds.
	\item Second, given the region, we determine $b_a$. This is now relatively straightforward by replacing the definitions of $\alpha$ and $x_U$ from propositions \ref{prop:alpha} and \ref{prop:zp} as functions of $b_a$ into the definition. Note that the region determines all case distinctions in the definitions, so that the equation $b = b(x; b_a, y_a)$ is smooth. In fact, it turns out that this can be written as a polynomial equation of degree at most 2 in $b_a$ and thus solved easily for each individual region.
\end{enumerate}

With $b_a$ determined, we can compute $\alpha$ and $x_U$ and from this we compute
\[
B = b - b(x + X; b_a, y_a)
.
\]
Note that the involved integral can be computed easily because $p$ is piecewise-linear in $x$.
We finally undo the normalization by returning $y_a \cdot B$.

Regarding computational cost, the algorithm can be implemented using only a constant number of arithmetic operations (no loops) and at most two square roots. While techniques like pre-computation and caching could be used to further speed up the region detection step, their utility would have to be carefully traded off against the cost of storage; our current implementation does not use these techniques.

\section{Concluding Remarks}

We have designed a desirable P-AMM redemption curve based on an anchored state that codifies how a stablecoin can sustainably adapt monetary policy to respond to crisis events without external input. This can work to mitigate currency runs if the stablecoin becomes under-reserved or, with some modification to consider the liquid reserve ratio, to mitigate bank run-like risks if the reserve contains illiquid assets.
We showed how this design satisfies desiderata 1--5 and 8 introduced in Section~\ref{sec:desiderata}.

Desiderata 6 and 7 can be reasoned about considering how the anchored state changes over time according to an exponentially time-discounted sum and using desiderata 1--5.
It is possible to show that desiderata 6 and 7 are satisfied formally. We leave this as the starting point for a wider study of time evolution. In particular, future work should also study how the P-AMM behaves under wider market reactions and other stabilization mechanisms, including a model of exogenous random shocks to reserve value.

The desiderata properties that the P-AMM achieves were informed by results about optimal monetary policies in currency peg models and redemption policies in bank run models. A natural follow-up question is how to optimally parameterize the P-AMM. A first step in this direction would be to adapt currency peg models to a stablecoin setting with the general P-AMM shape defining the monetary policy and optimize the outcomes of the resulting game.

A remaining question is how one should optimally calibrate the P-AMM. This would be best carried out using formal game-theoretic models of speculative attacks on the P-AMM. An example of this type of attack involves an attacker who can take an outside short position and try to profit by triggering a de-peg of the stablecoin. Models of these attacks would integrate the P-AMM, a secondary market, and a shorting market, and is outside the scope of this paper. Models of speculative attacks on fiat currencies, which informed the P-AMM desiderata in the first place, may serve as a good starting point; we discuss these in Appendix~\ref{apx:monetary-theory}.

Another avenue for future work is to explore alternative P-AMM designs that may satisfy the desiderata. For instance, a P-AMM with a sigmoidal shape would have further smoothness. However, such designs would likely present computational issues on-chain, including both the raw number of computational steps required (i.e., gas requirements for on-chain execution) and amplification of rounding errors arising from fixed point arithmetic.

\bibliographystyle{abbrvnat}
%\bibliography{references}

\addcontentsline{toc}{section}{References}  % Has to go below for some reason

\appendix
\section*{Appendix}
\addcontentsline{toc}{section}{Appendix}  % Has to go below for some reason

%%%%%%%%%%%%%%%%%%%%%%%%%%%%%%%%%%%%%%%%%%%%%%%%%%%%%%%%%%%%%%%%%%%%%%%%%%%%%%
\section{Parallels with the monetary economics literature.}\label{apx:monetary-theory}
A well-designed primary market for a stablecoin can be interpreted as an autonomous version of open market operations, comparing with how central banks interact with markets to shape monetary policy.
When new stablecoins are sold on the primary market, the balance sheet is expanded, and when stablecoins are redeemed, the balance sheet is contracted. The primary market design determines how much the balance sheet changes, supposing all proceeds of the market go onto the balance sheet.\footnote{
	Notably, many algorithmic stablecoins divert a share of primary market cash flow to holders of ``equity'' tokens ---we consider such systems \emph{insolvent-by-design} as they give away part of the ``seigniorage'' income from purchases of newly minted stablecoins (typically via buybacks of ``equity'' tokens), unlike a bank that maintains full asset-backing of deposits.
	This structure has contributed to many experienced crises for these coins.
}
Notice here that the primary market mechanism essentially solves the scaling issues that arise in leverage-based stablecoins, like Dai: the stablecoin is always able to meet excess demand by expanding the balance sheet (it does not need to match the demand of other agents in doing so).

The P-AMM can be compared to a variant on a crawling or managed float system for a currency peg. The monetary economics literature on these topics provides a starting point to understand this design.

International monetary economics is concerned with balance of payments crises (i.e., sudden changes in capital flows). With a stablecoin as opposed to a national currency, we're less concerned with money flows in/out of a country's economy. The analog for a stablecoin is with flows in/out of the reserve and the level of economic demand for use of the stablecoin as opposed to demand to speculate on the stablecoin. Further, stablecoin monetary policy is simplified to targeting stability relative to the target as opposed to further targeting growth of a national economy.

Speculative attacks on currency pegs are characterized in the global games literature (e.g., \cite{morris1998unique}). In these models, speculators can coordinate to attack the currency while profiting from bets on currency devaluation. High levels of coordination can force the government to abandon the peg. There is a unique equilibrium in such games, shown in \cite{morris1998unique}, given uncertainty in common knowledge of fundamentals (e.g., faith in government policy, economic demand, and health of reserves), which can lead to speculative attacks even when fundamentals are strong.

The curvature of the P-AMM serves to deter speculative attacks by increasing their cost in several ways. In large outflow settings, the curvature of $f$ can allow short-term (though not necessarily long-term) depreciation from the peg. This can be interpreted as raising interest rates for new stablecoin holders. Akin to zero coupon bonds bought at a discount, buyers expect to redeem for a higher price later. This is supported by the fundamental value of the reserve, which, when healthy, tends to shift the coordination equilibrium to \$1 as outflows equilibrate.\footnote{While certain uncollateralized (or ``implicit collateral'', see \cite{klages2020stablecoins}) stablecoins also propose similar narratives as here, they do so without a fundamental force, such as from a reserve, pushing coordination toward the \$1 equilibrium. Accordingly, one may question whether the stable equilibrium may really be \$0 price in such cases.}
Compared to a typical currency peg, the curvature of the P-AMM forces an attacking speculator to redeem at deteriorating prices throughout the attack, after which the redemption rate can bounce back. As a consequence, the crisis has to be stretched over long periods of time, during which speculators incur the spread loss, to have a permanent effect on the peg and reserve health.
Additionally, the funding rate for a short bet on the stablecoin--a prime profit source for a speculative attack--ought to take into account the transparent shape of the P-AMM and state of the reserve. In settings that are otherwise prime for speculative peg attacks (e.g., when reserve value per stablecoin is much less than \$1), the short funding rate ought to be very high to account for the ease of causing short-term devaluations via the P-AMM shape, which serves to further raise the costs of attack.

Lastly, we contrast with the bank run model in \cite{diamond1983bank}. In that model, the bank serves as insurance for two types of agents: one type who will need to withdraw early and another type who will not, but without knowing which is which ahead-of-time. Given the setup of the model, the bank is often prone to bank runs that depletes the bank's liquid assets. Speculative attacks on a stablecoin can often be viewed in a similar way to a bank run.
In this context, the stablecoin design effectively alters the assumptions of the Diamond-Dybvig model to deter the undesirable bank run equilibrium (see \cite{selgin2020modeling} for further discussion of the following points).
First, the curvature of the P-AMM reduces the redemption rate of large withdrawals. One cause of fragility in the Diamond-Dybvig model is requiring absolute liquidity out of bank deposits. Altering this structure can increase robustness at relatively small costs in terms of stablecoin liquidity.
Second, since a liquid stablecoin is tradeable on secondary markets there is often no need to directly redeem it in the primary market.\footnote{Similarly, digital commercial bank money does not need to be redeemed from the bank to use but can be used as a means of payment directly. A stablecoin on a public interoperable blockchain could be even more flexibly used without requiring redemption.}

Some qualitative results from the literature on speculative attacks on currencies can be applied directly to the P-AMM setting.
In particular, \cite{routledge2021currency} analyzes a game theoretic model that is implicitly similar to the P-AMM setting,
in which a government or currency issuer tries to maintain an optimal exchange rate peg while under-reserved. In this model, the optimal strategy involves the central bank demonstrating commitment to devalue the currency if too many traders demand redemption from reserves, which eliminates the speculative attack equilibrium and stabilizes the exchange rate.

Although \cite{routledge2021currency} is written in the context of cryptocurrency stabilization, the model is most descriptive of traditional currency models, in which the exchange rate is pegged to another currency that is held by the central bank in reserves. The cryptocurrency setting is a little different: (1) reserve assets may not be the currency target, and (2) exchange rate policy needs to be encoded on-chain in smart contracts as opposed to determined on-the-fly by central bank governors. However, after translating reserve asset values into the appropriate numeraire for the peg and adjusting for the fact that these values follow a stochastic process, the P-AMM setting can be adapted to the underlying model in \cite{routledge2021currency}. In this case, the P-AMM shape when under-reserved provides the means of committing to stablecoin devaluation in the event of speculative attacks, which can eliminate the speculative attack equilibrium.

%%%%%%%%%%%%%%%%%%%%%%%%%%%%%%%%%%%%%%%%%%%%%%%%%%%%%%%%%
\section{Simplified, Discrete Redemption Curve}\label{apx:discrete-decay}

In this section, we provide the formal details for the simplified price decay function from Section~\ref{sec:discrete-decay} where prices decay discretely from 1 to the reserve ratio. This discrete price decay form leads to much simpler calculations than the linear one outlined in the main text of this paper. However, it has the disadvantage that price decay happens abruptly, which creates risk for arbitrageurs and may also create an opportunity for speculative attacks.

We first describe the general design and the computation of dynamic parameters (similar to Section~\ref{sec:calc-params}) and we then discuss the implementation of the mechanism (similar to Section~\ref{sec:implementation}).

\subsection{Design and Dynamic Parameters}\label{apx:discrete-decay:calc-params}

Recall that here, we use
\[
	p(x; b_a, y_a) :=
	\begin{cases}
		1 &\text{if } x \ge x_U
		\\
		r_U := p(x_U; b_a, y_a) &\text{if } x < x_U,
	\end{cases}
\]
where $x_U$ is a dynamic parameter and $r_U$ is chosen such that $r_U=r(x_U)$ based on the other parameters. The parameters fulfill the same role as in our linear price decay mechanism where, however, the linear segment (and its parameter $\alpha$) as well as the lower cut-off point $x_L$ are missing.

\begin{lemma}
	\label{lem:discrete-b}
	Let $b_a < y_a$. Then the following hold:
	\begin{enumerate}
		\item We have
		\[
			b(x) =
			\begin{cases}
				b_a - x &\text{if } x \ge x_U
				\\
				b_a - x_U - r_U \cdot (x - x_U) &\text{if } x \le x_U
			\end{cases}
		\]
		\item If $x_U < b_a$, then $r_U = \frac {b_a - x_U} {y_a - x_U}$ and the reserve is never exhausted. Otherwise, the reserve is exhausted at a point $x < y_a$.
		\item In this case, we have $r_U \ge \btheta$ iff $b_a / y_a \ge \btheta$ and $x_U \le \frac {b_a - \btheta y_a} \theta$, where $\theta = 1-\btheta$.
	\end{enumerate}
\end{lemma}
\begin{proof}
	\leavevmode
	\begin{enumerate}
		\item This follows immediately by integration over $p$.
		\item First assume $x_U < b_a$. By part~1 we have $b(x_U) = b_a - x_U$ and further $y(x_U) = y_a - x_U$, and $r_U = r(x_U) = b(x_U) / y(x_U)$. By assumption, $r_U \in (0, 1)$ and, since we redeem at the reserve ratio after the point $x_U$, the reserve is never exhausted.
		Now assume $x_U \ge b_a$. Then, by part~1, we have $b(b_a)=0$ and by assumption, $b_a < y_a$.
		\item By part~2 and simple algebraic transformation, $r_U \ge \btheta$ iff $x_U \le \frac {b_a - \btheta y_a} \theta$. This is only possible if $b_a / y_a \ge \btheta$ since otherwise, the right-hand side is negative.
		\qedhere
	\end{enumerate}
\end{proof}

Let $\bzp$ be the (optional) ceiling value for $x_U$. Then by the previous lemma the maximal $x_U \le \bzp$ that ensures $r_U \ge \btheta$ (if possible, and $r_U$ maximal otherwise) is $x_U = \min(\bzp, \hzp)$, where
\begin{equation}
	\hzp := \max\left(0,\,y_a \frac {r_a - \btheta} {1-\btheta}\right)
	\label{eq:discrete-xu}
\end{equation}
where we recall that $r_a := b_a/y_a$.
Note that we choose $x_U=0$ if $r_a < \btheta$, so that in this case, we redeem at the reserve ratio from the very beginning.

\subsection{Implementation}\label{apx:discrete-decay:implementation}

To implement our mechanism (i.e., to define $\rho(x, b, y)$), we can perform a distinction into different regions like in sections \ref{sec:reconstruction-uniqueness} and \ref{sec:implementation}, though the operation is much simpler here. We only consider two dimensions with two cases each.
\begin{itemize}
	\item Case I, where $\hzp \ge \bzp$, and case~II, where $\hzp \le \bzp$.
	\item Case i, where $x \le x_U$ and case~iii, where $x \ge x_U$. There is no case~ii.
\end{itemize}

For the implementation, we proceed analogously to Section~\ref{sec:implementation}.
Reconstruction of $b_a$
is greatly simplified in the discrete case because the range of possible values is much restricted and, importantly, all relevant equations are linear and do not involve either a square root or a square. First note that, as always, we can easily obtain $y_a = y + x$, so it only remains to reconstruct $b_a$. Also recall that we only need to do this when $1 > r=b/y > \btheta$ since we otherwise always redeem at 1 or $r$, respectively.
It is easy to see that, like in Theorem~\ref{thm:mon-b0}, $b(x)$ (and equivalently $r(x)$) for a fixed $x$ is strictly monotonic in $b_a$ whenever $r > \btheta$ and thus, any such point has a unique $b_a$ associated.
Reconstruction of $b_a$ can now be done using a “trial and error” approach as follows:
\begin{itemize}
	\item (Case i) Let $b_a' = b+x$ and test if $x \le x_U(b_a', y_a)$. In this case, $b_a = b_a'$.
	\item (Case I iii) Otherwise, $b_a = b + \bzp + r\cdot(x-\bzp)$.\footnote{This value is unimportant to calculate redemption amounts, though, and would not have to be computed. This is because, since we're in case iii, the redemption price will always be $r=b/y$ for all $x' \ge x$.}
\end{itemize}
Correctness of the above reconstruction follows from uniqueness of $b_a$ and the fact that the respective conditions hold in the respective cases.
Correctness in case I iii follows by solving the equation for $b(x)$ from Lemma~\ref{lem:discrete-b} for $b_a$.
Note that case II iii does not need to be checked because, by choice of $x_U$, we must have $r_U = \btheta$ in this case.
To compute a redemption amount, in case~i, we have to now compute the dynamic parameter $x_U$ using Equation~\ref{eq:discrete-xu} and then integrate over the redemption curve, just like in the case with linear price decay. In case I iii, this is trivial and we can simply compute $X \cdot r$, where $X$ is the amount of redeemed stablecoins.
It is obvious that the mechanism can be implemented using a constant number of basic arithmetic operations.

%%%%%%%%%%%%%%%%%%%%%%%%%%%%%%%%%%%%%%%%%%%%%%%%%%%%%%%%%
\section{Detailed Calculation of Dynamic Parameters}
\label{apx:calc-params}

We first establish how to calculate the dynamic parameters of the redemption curve from the anchor point $(b_a, y_a)$. We do this as a series of technical lemmas that will be used in our later results. We start by showing how to calculate $b(x)$ in the simplified context in which the dynamic parameters are known/fixed. Recall though that, in general, the dynamic parameters are functions of the current state (more precisely, of the anchor point) and the static parameters. We then move on to derive results about how to calculate the dynamic parameters in their general form.
We prove additional technical guarantees, which are not required for the purpose of this exposition, but may be useful when implementing our methods, in Appendix~\ref{apx:sanity-lemmas}.

We now consider how the current state is connected to the anchor point. For $y(x)$, we simply have $y(x) = y_a -x$. When $\alpha, x_U$, and $x_L$ are known and fixed, it is simple to calculate the function $b(x)$. The next proposition specifies how to do this.
To state the proposition, observe that the reserve ratio at $x$ is $r(x) = b(x)/y(x) =b(x)/(y_{a}-x)$. Observe that $r(0)=r_a=b_a/y_a$. Note that $r(y_{a})$ would be ill-defined in this sense since the denominator would be 0. We extend the definition continuously by $r(y_{a}):=\lim_{x\to y_{a}}r(x)$. This will be well-defined for all relevant cases below. 

\begin{proposition}\label{prop:b}
	For fixed $\alpha, x_U, x_L$ with $x_U \le x_L$ we have
	$$b(x)=b_{a}-\int_{0}^{x} p(x')\de x'=\begin{cases}
		b_{a}-x, & x\le x_{U}\\
		b_{a}-x+\frac{\alpha}{2}(x-x_{U})^{2}, & x_{U}\le x\le x_{L}\\
		r_{L}(y_{a}-x), & x_{L}\le x,
	\end{cases}$$
	where $r_L = r(x_L)$.
	$r_L$ can be computed using the second case in the case distinction alone.
\end{proposition}

\begin{center} \hyperlink{pf:prop:b}{\texttt{[Link to Proof]}} \end{center}

% todo maybe kill this par; seems repetitive.
Recall that we have three static parameters
$\balpha\in(0,\infty), \bzp\in[0,\infty]$, and $\btheta\in[0,1]$. $\btheta$ defines a floor on the reserve ratio, if achievable, and $\balpha$ and $\bzp$ are bounds on the respective parameter: we always have $\alpha\ge\balpha$ and $x_{U}\le\bzp$.
%%%
Depending on the anchor point $(b_a, y_a)$ and constrained by these parameters, we choose values for the dynamic parameters that determine the curve shape. 
Define an auxiliary function
\begin{equation}
	\rhou(x):=\begin{cases}
	1, & x\le x_{U}\\
	1-\alpha(x-x_{U}), & x\ge x_{U}.
\end{cases}
\label{eq:rhou}
\end{equation}
Then the dynamic parameters $x_U$, $\alpha$, and $x_L$ are chosen according to the following rule.
\begin{itemize}
	\item For given $x_U$ and $\alpha$, $x_L$ is chosen such that $x_L \in (x_U,y_a]$ and $p(x_L; b_a, y_a) = r(x_L ; b_a, y_a)$, where both sides of this equation can be computed based on $p^U$. The values of the remaining parameters $x_U$ and $\alpha$ will be chosen such that such a point exists (see Proposition~\ref{prop:xm} below).
	Note that the case $x_L = y_a$ is not pathological, but, as we will see, it occurs regularly.
	It is easy to see that, by choice of $x_L$, we have $r(x) = p(x) = r_L$ for all $x \ge x_L$, i.e., we redeem at the reserve ratio beyond $x_L$, and this is sustainable.
	
	\item $x_U$ and $\alpha$ are chosen such that $0 \leq x_U \leq \bzp$, $\alpha \geq \balpha$, $x_L$ exists, and $r_L \geq \btheta$, if possible. It is easy to see that this is possible iff $r_a > \btheta$ (one possible choice is $x_U = 0$ and $\alpha \rightarrow \infty$ as $r_a \searrow \btheta$). In the trivial case where this does not hold, we set the marginal redemption price to constant $r_a$.
	
	\item Among the admissible combinations of $x_U$ and $\alpha$, the parameter values are chosen such that first, $\alpha$ is minimized among all possible $\alpha$ values and, second, $x_U$ is maximized given this $\alpha$. This implies that, if there are admissible solutions and $\alpha < \balpha$, then we must have $x_U = 0$. This follows from the fact that $r_L$ increases if we reduce $x_U$.
\end{itemize}

The rule by which $x_u$ and $\alpha$ are chosen in the third step encodes that, when confronted with a trade-off between a not-too-steep price decay and a prolonged support of the exact peg of \$1, our mechanism prioritizes the former over the latter.
We argue that this is the appropriate trade-off in the interest of market stability for the reasons outlined in Section~\ref{sec:desiderata}.
% TODO CONFERENCE ^ review & perhaps bolster that argument. (comment by Lewis)
Note that the mechanism only applies trade-off only applies within the limits set by the $\balpha$ and $\bzp$ static parameters.

Going forward, we will focus on the non-trivial case where $b_a < y_a$ (otherwise $b > y$ at any point and the marginal price is constant at 1), $r_a >\btheta$ (otherwise the marginal price is constant at $r_a = r(x)\;\forall x$), and, where $x_U$ occurs as a parameter, we assume $y_a \ge x_U$ (otherwise the system would be configured to redeem at price 1 for all $x$ and in particular will run out of reserves at some point since $b_a < y_a$).

The following proposition shows how to calculate $x_L$ based on $x_U$ and $\alpha$ and in what settings such a point exists. This result also shows that a key tenet of the primary market design will be choosing parameters such that $x_L$ exists as otherwise the reserve can be exhausted.
%%%
To make our formulas more compact, define the following shorthands:
let $\Delta_{a}=y_{a}-b_{a}$ and $y_{U}=y_{a}-x_{U}$; let $b_{L}=b(x_{L})$.
% todo check to which extent y_U and b_L are actually used.

\begin{proposition}\label{prop:xm}
	For given fixed parameters $x_U, \alpha$, the following hold.
	\begin{enumerate}
		\item There exists a point $x_{L}\in[0,y_{a}]$ where $p(x_{L})=r(x_{L})$ iff
		\begin{equation}\label{eq:xm-exists}
			\alpha\ge2\frac{y_{a}-b_{a}}{\left(y_{a}-x_{U}\right)^{2}}.
		\end{equation}
		
		\item If (\ref{eq:xm-exists}) does not hold and $x_{L}:=\infty$ in the definition of $p$,  then the reserve is exhausted before all tokens have been redeemed. Formally, in this case, there is $x\in[0,y_{a})$ such that $b(x)=0$.
		
		\item \label{itm:xm:fmls} If (\ref{eq:xm-exists}) holds, then $x_{L}$ is unique and
		\begin{align*}
			x_{L}	&= y_{a}-\sqrt{(y_{a}-x_{U})^{2}-\frac{2}{\alpha}(y_{a}-b_{a})} \\
			r_{L}	&= 1-\alpha(x_{L}-x_{U}) \\
			b_{L}	&= (1-\alpha(x_{L}-x_{U}))(y_{a}-x_{L})
		\end{align*}
	\end{enumerate}
\end{proposition}

\begin{center} \hyperlink{pf:prop:xm}{\texttt{[Link to Proof]}} \end{center}

% todo should this remark be moved somewhere else? Or killed?
\begin{remark}
	From the previous proposition, we easily receive an analytical expression for $r(x)=b(x)/y(x)=b(x)/(y_{a}-x)$. Observe that $b(x)$ and $r(x)$ are continuous functions of $x$. Observe in particular that, for $x\ge x_{L}$, we have $r(x)=r_{L}(y_{a}-x)/(y_{a}-x)=r_{L}= p(x)$. Thus, after $x_{L}$, the reserve ratio remains constant and the marginal redemption price is the reserve ratio.
\end{remark}

% todo the following part seems unnecessarily complicated and can probably be re-written more clearly somehow.

We now formalize the rule by which $\alpha$ and $x_U$ are chosen.
Fix $y_{a}$, $b_{a}$, and $\btheta$.
We call a pair $(\alpha,x_{U})$ \emph{admissible} if $x_{L}$ exists and $r_{L}\ge\btheta$.
We call $x_{U}$ \emph{admissible for} $\alpha$ if $(\alpha,x_{U})$ is admissible and we call $\alpha$ \emph{admissible }if there exists some $x_{U}$ such that $(\alpha,x_{U})$ is admissible.
Note that the set of $x_{U}$ admissible for $\alpha$ always form a closed interval $[0,\hzp(\alpha)]$ and the set of admissible $\alpha$ is also a closed interval $[0, \halpha]$.
This follows from monotonicity in (\ref{eq:xm-exists}), the fact that $r_{L}$ is monotonically decreasing in $x_{U}$ (as can be seen from the formulas in Proposition~\ref{prop:xm}), and closedness.
Define $\halpha$ and $\hzp(\alpha)$ as the interval bounds indicated above. In words, $\halpha$ is the minimum $\alpha$, disregarding $\balpha$, such that $(\alpha, x_U=0)$ is admissible and $\hzp(\alpha)$ is the maximum $x_U$, disregarding $\bzp$, such that $(\alpha, x_U)$ is admissible.
By Proposition~\ref{prop:xm} we can see that $\halpha$ always exists (if $b_{a}/y_{a}>\btheta$) and is positive and for $\alpha \ge \halpha$, $\hzp(\alpha)$ always exists.

We now choose the dynamic parameters $\alpha$ and $x_U$ as follows:
\begin{enumerate}
	\item First let $\alpha=\max(\halpha,\balpha)$.
	\item Then let $x_{U}=\min(\hzp(\alpha),\bzp)$.
\end{enumerate}

Note that the upper bound $\bzp$ is essentially optional for our construction; we can choose $\bzp=\infty$ (or $\bzp=y_{a}$) to deactivate it.
In this case, we always have $x_{L}=y_{a}$ if the $\btheta$ bound on the reserve ratio is not binding.\footnote{This will be captured formally as cases II h and III H below.}
The lower bound $\balpha$ is also optional and can be deactivated by setting $\balpha=0$. Note, however, that for $\balpha=0$, we will \emph{always} receive $x_{U}=0$. This is because then $\alpha=\halpha$ and it is easy to see that $\hzp(\halpha)=0$ (otherwise, we could have chosen $\halpha$ smaller by strict monotonicity; see Proposition~\ref{prop:xm}).
The following lemma will help us in our construction.
Let $\itheta = 1-\btheta$.

\begin{lemma}
	\label{lem:thetafloor-prep}If (\ref{eq:xm-exists})
	holds, then $r_{L}\ge\btheta$ iff
	\begin{alignat}{4}
		&  & \alpha(y_{a}-x_{U}) & \,\le\itheta\tag{TH}\label{eq:thetafloor-high}\\
		\text{or} &\qquad  & \alpha(b_{a}-\btheta y_{a})-\alpha\itheta x_{U}-\frac{1}{2}\itheta^{2} & \,\ge0.\tag{TL}\label{eq:thetafloor-low}
	\end{alignat}
\end{lemma}

\begin{center} \hyperlink{pf:lem:thetafloor-prep}{\texttt{[Link to Proof]}} \end{center}

We can interpret the distinction between the conditions \eqref{eq:thetafloor-high} and \eqref{eq:thetafloor-low} in terms of whether or not the reserve ratio floor $\btheta$ is binding.
Observe first that (\ref{eq:thetafloor-high}) is equivalent to $1-\alpha(y_{a}-x_{U})\ge\btheta$, and this implies that
$p(x)\ge\btheta$ for all $x$ and independently of $x_L$, including for the case $x_L=y_a$, where the redemption curve $p$ has no final constant segment. In this case we say that $\btheta$ is not binding. If \eqref{eq:thetafloor-high} does not hold, then we need to choose $x_L < y_a$ since otherwise the reserve ratio would fall short of the floor $\btheta$; we say that $\btheta$ is binding.
% todo maybe move the following end of the paragraph to earlier, to *motivate* why we do this.
Conceptually, if at least one of \eqref{eq:thetafloor-high} or \eqref{eq:thetafloor-low} holds, the redemption price will always be at least $\bar\theta$, conditional on the assumptions at the beginning of this section (in particular, the system starts with enough reserve capitalization). This is desirable as anyone can understand this bounding (as well as other PAMM mechanics) ahead-of-time.

Armed with Lemma~\ref{lem:thetafloor-prep}, we can now construct the values $\halpha$ and $\hzp(\alpha)$ for any $\alpha$. We begin with $\halpha$. Recall that $\itheta=1-\btheta$.
\begin{proposition}
	\label{prop:alpha} We have
	\[
	\halpha=\begin{cases}
		\halpha_{H}:=2\frac{1-r_{a}}{y_{a}}, & r_{a}\ge\frac{1+\btheta}{2}\\
		\halpha_{L}:=\frac{1}{2}\frac{\itheta^{2}}{b_{a}-\btheta y_{a}}, & r_{a}\le\frac{1+\btheta}{2}
	\end{cases}
	\]
	and $\halpha$ is continuous in the other parameters.
\end{proposition}

\begin{center} \hyperlink{pf:prop:alpha}{\texttt{[Link to Proof]}} \end{center}

We continue with an explicit formula for $\hzp(\alpha)$. Note that, due to the way in which we  choose our parameters, we only need to consider $\hzp(\balpha)$. However, no additional effort is required to obtain a formula for general $\alpha$.

\begin{proposition}
	\label{prop:zp}We have
	\[
	\hzp(\alpha)=\begin{cases}
		\hzph:=y_{a}-\sqrt{2\frac{\Delta_{a}}{\alpha}}, & \alpha\Delta_{a}\le\frac{1}{2}\itheta^{2}\\
		\hzpl:=y_{a}-\frac{\Delta_{a}}{\itheta}-\frac{1}{2\alpha}\itheta & \alpha\Delta_{a}\ge\frac{1}{2}\itheta^{2}.
	\end{cases}
	\]
	and $\hzp(\alpha)$ is continuous in $\alpha$ and in the other parameters.
\end{proposition}

\begin{center} \hyperlink{pf:prop:zp}{\texttt{[Link to Proof]}} \end{center}

The technical results in this section show how the dynamic parameters $\alpha$, $x_U$, and $x_L$ can be calculated from the current state using the anchor point. In the following sections, we will proceed with our analysis with these rules for choosing dynamic parameters as given.

\section{A Formal Treatment of Path Deficiency under Minting and Fees}
\label{apx:path-properties}

%%%%%%%%%%%%%%%%%%%%%%%%%%%%%%%%%%%%%%%%%%%%%%%%%%%%%%%%%
\subsection{Extension to Fees and Minting}\label{sec:extend-fees-minting}

The P-AMM in reality will take an extended form of the setup developed thus far. In this extended form, the redemption curve will incorporate a trading fee and there will be a separate minting curve for $x$ that moves in the reverse direction. We will now show how the desired properties --but not path independence directly-- can be retained in this extended form.

This extended form can no longer be modeled by a single dynamical system. Instead, different differential equations describe the effects of increasing $x$ (redemptions) as opposed to decreasing $x$ (minting).
Let $\gamma(x,b,y)\geq 0$ be the trading fee that is imposed on redemptions. And let $\varphi(x,b,y)\geq 1$ be a function describing the marginal price of minting a new stablecoin. Notice that $\varphi(x,b,y) = 1 + \eps$ is such a function for any $\eps>0$.

In the extended form, redemption actions are described by the following slightly altered form of (\ref{eq:odes}):
\begin{equation}\label{eq:redeem-odes-fee}
	\begin{aligned}
		\frac{\de b(x)}{\de x} &= -\rho\big(x, b(x), y(x)\big) + \gamma\big(x, b(x), y(x)\big) \\
		\frac{\de y(x)}{\de x} &= -1.
	\end{aligned}
\end{equation}
And minting actions are described by the different set of differential equations:
\begin{equation}\label{eq:mint-odes}
	\begin{aligned}
		\frac{\de b(x)}{\de x} &= -\varphi\big(x,b(x),y(x)\big) \\
		\frac{\de y(x)}{\de x} &= -1.
	\end{aligned}
\end{equation}

As the extended system evolves differently for different directions of change in $x$, we no longer have path independence. To see this, simply consider a closed path in $x$, which returns to the same starting point in $x$-space, but will often not return to the same starting point in $(x,b,y)$-space.
However, there a generalization of path independence, called \emph{path deficiency}, which retains many of the useful properties we desire.

%%%%%%%%%%%%%%%%%%%%%%%%%%%%%%%%%%%%%%%%%%%%%%%%%%%%%%%%%
\subsection{Path Deficiency Properties}

We next show two properties analogous to path deficiency in CFMMs (see \cite{angeris2020improved}).
As the P-AMM is not a CFMM, we approach this slightly differently. For our purposes, we will characterize path deficiency-like results in terms of \emph{reserve ratio curves} that can be encountered along a trading path. These reserve ratio curves are defined in previous sections and visualized in Figure~\ref{fig:rr_curves}.
In particular, these reserve ratio curves are functions arising from our original system (\ref{eq:odes}) that map $x$ to a reserve ratio $r(x)$ parameterized by an anchor point $r_a=b_a/y_a$, where we assume $\bzp$, $\balpha$, and $\btheta$ are fixed. Without loss of generality, we will take $y_a = 1$ so that $r_a = b_a$.
%(we show how this can be normalized in Section~\ref{sec:implementation} on implementation).

Recall that each current state is associated with a single such reserve ratio curve. While in the case of redemptions without fees, we always remain on this reserve ratio curve, when we add in fees and a separate minting curve, we instead may shift reserve ratio curves as we move along a path in $x$.

We start with the following definitions:
\begin{itemize}
	\item $\mathcal{R}$ is the set of all reserve ratio curves,
	
	\item $\mathfrak{r} \in \mathcal R$ is some initial reserve ratio curve,
	
	\item $\mathcal{C}$ is the set of paths in $[0,1]$, i.e., $\mathcal C = \{ f: [0,1] \rightarrow [0,1] \hspace{0.2cm} \vert \hspace{0.2cm} f \text{ is continuous} \}$,
	
	\item $r_{f,\mathfrak r}: [0,1] \rightarrow [0,\infty)$ is the function returning the reserve ratio at points along the path $f \in\mathcal C$ starting at the initial point $\big(f(0), \mathfrak r(f(0))\big)$ in $(x,r)$-space.
\end{itemize}
Notice that $r_{f,\mathfrak r}$ sweeps away the details of (\ref{eq:redeem-odes-fee}) and (\ref{eq:mint-odes}) but is easy to see is well-defined for a given $f \in \mathcal C$.

Paths in the set $\mathcal C$ are interpreted as paths for the variable $x$. Note that we consider paths for $x$ within $[0,1]$. The upper bound comes from the maximum amount that can be redeemed. It is inherently possible for more supply to be minted than $y_a$, and so the lower bound could conceptually be passed in reality.
% todo review how exactly this would be done; Steffen wants to know.
It is possible to extend the results by renormalizing the system to $y_a=1$.
%We establish that this normalization from the current state works as intended in the following section on implementation.

The following lemma establishes that the anchor point $r_a$ is weakly increasing along any trading path.

\begin{lemma}\label{lem:path-deficiency1}
	Let $\mathfrak r \in \mathcal R$ such that $\mathfrak r \leq 1$ and $f\in \mathcal C$. Then $r_a\big( f(t), r_{f,\mathfrak r}(t) \big)$ ---the anchor point for each state $(x,r)$ along the path--- is non-decreasing in $t$.
\end{lemma}

\begin{center} \hyperlink{pf:lem-path-deficiency1}{\texttt{[Link to Proof]}} \end{center}

This enables our first path deficiency-like result in the next theorem.
Consider that we start on an initial reserve ratio curve. Moving along this curve describes the behavior of the original path independent system along a trading path, which we proved various desirable properties about in the previous sections.
The reserve ratio curve that we are on --independent of where we are on it-- is one good measure for the health of the system as being on a higher curve is point-wise weakly better than being on a lower curve.
The following theorem establishes that the protocol health is weakly increasing in this way along any trading path.\footnote{
	These results parallel those of path deficiency in CFMMs (see \cite{angeris2020improved}). To draw the parallel further, these and potentially further path deficiency results for P-AMMs may be expressed in terms of \emph{reachable sets of reserve ratio curves}
	$S(\mathfrak{r}) = \Big\{\ \psi \in \mathcal{R} \hspace{0.2cm} \vert \hspace{0.2cm} 
		\big(f(t), r_{f,\mathfrak r}(t)\big) \in \psi \text{ for some } t \in [0,1] \text{ and for some } f \in \mathcal{C} \Big\}$
	that weakly contract along a path.
	To illustrate, Lemma~\ref{lem:path-deficiency1} would express that functions in this reachable set are point-wise lower bounded by $\mathfrak r$, and that reachable sets do not expand along a path.
	This may be useful in expanded contexts, such as involving discrete trades, in which it is not obvious that two valid paths can be concatenated into a single valid path, or settings in which reserve ratio curves are not nicely represented by anchor points.
}

% Hack to get the labels right. This carries on to references, which makes everything magically consistent.
\bgroup
\renewcommand\thetheorem{\ref{thm:path-deficiency1-informal} (formal)}
\begin{theorem}\label{thm:path-deficiency1}
	Let $\mathfrak r \in \mathcal R$ such that $\mathfrak r \leq 1$. Then for all $f \in \mathcal C$ and for all $t \in [0,1]$, we have
	$\mathfrak r\big( f(t) \big) \leq r_{f,\mathfrak r}(t).$
\end{theorem}
\egroup

\begin{center} \hyperlink{pf:path-deficiency1}{\texttt{[Link to Proof]}} \end{center}

We now turn to our second path deficiency-like result. The following theorem shows that, in settings with a proportional fee, there is no incentive for a trader to strategically subdivide a net redemption trade into a sequence of different trades.

\bgroup
\renewcommand\thetheorem{\ref{thm:path-deficiency2-informal} (formal)}
\begin{theorem}
	\label{thm:path-deficiency2} Let $\gamma(\cdot) = \eps\rho(\cdot)$ for some $0 \leq \eps < 1$ and let $\mathfrak r \leq 1$ be the initial reserve ratio curve. Then:
	\begin{enumerate}
		\item The redemption system described in (\ref{eq:redeem-odes-fee}) is path independent within a block.
		\item An individual trader in the extended system described in Section~\ref{sec:extend-fees-minting} has no incentive to subdivide a net redemption within a block.
	\end{enumerate}
\end{theorem}
\egroup

\begin{center} \hyperlink{pf:path-deficiency2}{\texttt{[Link to Proof]}} \end{center}

\section{Details on Efficient Implementation}\label{apx:implementation}\label{apx:implementation}

In this section, we discuss how to implement our mechanism on-chain.
To execute a redemption using our mechanism, we are given the current state $(x, b(x), y(x))$ and we need to find $b_a$ and $y_a$ such that $b(x; b_a, y_a) = b(x)$. We can then simply integrate the curve of marginal redemption rates, starting at $x$, to execute a redemption; note that integration can easily be done analytically using Proposition~\ref{prop:b}.
Recall from Section~\ref{sec:reconstruction-uniqueness} that finding $y_a$ is trivial because $y(x) = y_a - x$, so $y_a = y(x) + x$. The main challenge is therefore to find the appropriate $b_a$. By strict monotonicity (Theorem~\ref{thm:mon-b0}), $b_a$ is unique.
The proofs for this section can be found in Appendix~\ref{apx:proofs-implementation}.
We prove additional technical guarantees, which are not required for the purpose of this exposition, but may be useful when implementing our methods, in Appendix~\ref{apx:sanity-lemmas}.

In this section, we assume that the parameters $\balpha$ and $\bzp$ scale in the anchored amount of outstanding stablecoins $y_a$ (see also Section~\ref{sec:design}).
More in detail, we assume that $\bzp = \bzpo \cdot y_a$ and $\balpha = \balphao / y_a$, where $\bzpo$ and $\balphao$ are parameters that can be chosen arbitrarily and that correspond to the value of $\bzp$ and $\balpha$ for $y_a=1$.
Choosing the parameters like this is intuitive and economically meaningful because $\bzpo$ is now the \emph{share} of “initial” (anchored) outstanding stablecoins that can be redeemed at dollar value and $\balphao$ is minimum decay in reserve ratio relative to the \emph{share} of “initial” outstanding stablecoins that have been redeemed.
Note that the minimum reserve ratio $\btheta$ remains an absolute number and is \emph{not} relative to $y_a$.
As these parameters $\bzpo$ and $\balphao$ are economically meaningful, we expect that they only need to be adjusted rarely, so that governance interventions are minimized.

If we were working under the mild computational constraints of a regular machine, our mechanism could be implemented in a straightforward way using bipartition. Using this method, we can find a $b_a \in [0, y_a]$ such that $b(x; b_a, y_a) = b(x) \pm \eps$, where $\eps$ is a chosen precision. This would require $\log(y_a / \eps)$ evaluations of $b(x; b_a, y_a)$ in the worst case. 
While this is computationally unproblematic for a regular machine, this method is far too computationally expensive to be implemented on-chain: each evaluation of $b(x; b_a, y_a)$ (for a changing value of $b_a$) requires the re-computation of the dynamic parameters $x_U$ and $\alpha$, and it requires computing up to two square roots (for $x_L$ and potentially for $x_U$). Square roots can of course only be computed using an iterative process, which is also computationally relatively expensive.
The whole computation then has to be repeated several times until the interval of possible $b_a$ values is small enough.

Fortunately, we can implement our mechanism in a much more computationally efficient way that is amenable to implementation on-chain.
Our method makes use of precomputed data that only depends on the normalized parameters.
More in detail, we will show the following theorem. We can limit our view to states where the reserve ratio $r$ lies in the interval $(\btheta, 1)$ as our mechanism otherwise simply yields a redemption rate of $r$ or $1$, respectively.

\begin{theorem}
	\label{thm:reconstruction-all}
	Given is a state $(x, b, y)$ such that $1 > b/y > \btheta$. Let $y_a = y + x$.
	We can find the unique $b_a$ such that $b = b(x; b_a, y_a)$ using a constant number of basic arithmetic operations and at most one square root, together with precomputed data that only depends on the parameters $\bzpo$, $\balphao$, and $\btheta$ and is independent of the state $(x, b, y)$.
	%%%
	Our precomputation can be performed using a constant number of basic arithmetic operations and one square root, and it can be verified using a constant number of arithmetic operations and without taking any square root.
\end{theorem}

Our method proceeds in three steps: in the precomputation step, we compute the curves of form $x \mapsto b(x; b_a, y_a=1)$ that mark the threshold between the different cases defined in Section~\ref{sec:reconstruction-uniqueness}: the threshold between case I and II, case II and III, and the thresholds between cases h and l and H and L, respectively.
Note that these computations in principle only need to do once as long as the parameters  $\bzpo$, $\balphao$, $\btheta$ are kept the same.

In the second step and given a state $(x, b, y)$, we use this precomputed information to detect the region (such as II h i) into which this state falls. Since most of the work has already been done in the precomputation step, this step only requires a constant number of arithmetic operations and no square roots.
Finally, once we know the region, we solve an equation of degree 1 or 2 to compute the value of $b_a$ that leads to $b$. This third step may involve a square root, depending on the region.

\subsection{Normalization}

Recall that the precomputed information only applies to the $y_a=1$ case. To use it in the region detection step, we normalize all values to $y_a=1$. We then run all further computations on this normalized version and scale the resulting $b_a$ back to the original value of $y_a$. Formally, we exploit the following scaling property.
Throughout this section, we denote the dynamic parameters as functions of the values on which they depend. For example, we write $\alpha(b_a, y_a)$ for the value that the dynamic parameter $\alpha$ would take according to Proposition~\ref{prop:alpha} at a certain $(b_a, y_a)$ combination.

\begin{lemma}
	\label{lem:scaling}
	Assume that $\bzp$ and $\balpha$ are chosen relative to $y_a$ as described above.
	Let $\zeta > 0$. Then
	\[
		b(\zeta x; \zeta b_a, \zeta y_a) = \zeta b(x; b_a, y_a)
		.
	\]
\end{lemma}

\begin{center} \hyperlink{pf:lem:scaling}{\texttt{[Link to Proof]}} \end{center}

The lemma immediately implies that it is enough to be able to determine $b_a$ in the \emph{normalized} case where $y_a=1$.
\begin{theorem}
	\label{thm:normalization}
	Assume that $\bzp$ and $\balpha$ are chosen relative to $y_a$ as described above. For some given $x$, $b$, and $y_a$, let $b_{a,1}$ be such that $b(x/y_a; b_{a,1}, 1) = b/y_a$. Let $b_a = b_{a,1} y_a$. Then $b(x; b_a, y_a) = b$.
\end{theorem}
\begin{proof}
	This follows immediately from Lemma~\ref{lem:scaling} for $\zeta = 1/y_a$.
\end{proof}

By Theorem~\ref{thm:normalization}, it is enough to consider the case $y_a=1$. Otherwise, given some state $(x, b, y)$ we can consider the scaled state $(x/y_a, b/y_a, y/y_a)$ (where $y_a:=y+x$), consider the resulting $b_a$, and return $y_a \cdot b_a$.
Despite our normalization, in the following, we will usually use an explicit variable for $y_a$ in the interest of clarity of the exposition.

\subsection{Precomputation Step}

We now discuss our precomputation step. We use the following lemma to distinguish between cases I, II, and III.

\begin{lemma}
	\label{lem:ba_for_xu}
	Let $z \le y_a$. Let $y_z := y_a - z$. Then $\hzp = z$ iff
	\[
		b_a =
		\begin{cases}
			y_a - \frac \alpha 2 \cdot y_z^2 & \text{if } 1 - \alpha y_z \ge \btheta
			\\
			y_a - \theta y_z + \theta^2 \cdot \frac 1 {2\alpha} &\text{otherwise}.
		\end{cases}
	\]
\end{lemma}

\begin{center} \hyperlink{pf:lem:ba_for_xu}{\texttt{[Link to Proof]}} \end{center}

\begin{algorithm}
	\caption{Precomputation step}
	\label{alg:precompute}
	\begin{algorithmic}[1]
		\Require{$\balpha, \bzp, \btheta$}
		\Ensure{$P := (b_a^{I/II},\, x_L^{I/II},\, b_a^{II/III},\, b_a^{h/l},\, x_U^{h/l},\, b_a^{H/L},\, \alpha^{H/L})$}
		\State Calculate $b_a^{I/II}$ via Lemma~\ref{lem:ba_for_xu} such that $\hzp = \bzp$ for $y_a=1$ and $\alpha=\balpha$.
		\State Calculate $x_L^{I/II}$ via Proposition~\ref{prop:xm} for $b_a^{I/II}$, $y_a=1$, $\alpha=\balpha$, and $x_U=\bzp$.
		\State Calculate $b_a^{II/III}$ via Lemma~\ref{lem:ba_for_xu} such that $\hzp=0$ for $y_a=1$ and $\alpha=\balpha$.
		\State Let $b_a^{h/l} := y_a - \theta^2 \cdot \frac 1 {2 \cdot \balpha}$.
		\State Calculate $x_U^{h/l}$ via Proposition~\ref{prop:zp} for $b_a^{h/l}$, $y_a=1$, and $\alpha=\balpha$.
		\State Let $b_a^{H/L} := y_a \cdot \frac {1+\btheta} 2$
		\State Calculate $\alpha^{H/L}$ via Proposition~\ref{prop:alpha} for $b_a^{H/L}$ and $y_a=1$.
	\end{algorithmic}
\end{algorithm}

Our precomputation step is depicted in Algorithm~\ref{alg:precompute}.
It exploits Lemma~\ref{lem:ba_for_xu} as well as other relationships to calculate all relevant thresholds between cases I / II / III as well as those between cases h / l and and H / L. Each of these thresholds is determined by a specific $b_a$ value. We also compute the dynamic parameters for each of these $b_a$ values as far as it is necessary, so that any point along the respective threshold curve for $b(x)$ can be quickly evaluated. The following lemma tells us that these threshold values can be used to determine in which of the different regions we are.
As discussed above, we only need to consider the normalized case $y_a=1$.

\begin{lemma}
	\label{lem:precompute}
	Fix $y_a=1$ and let $b_a \in [0, y_a)$ be arbitrary. Let $x \in [0, y_a]$ and assume that $1 > b(x; b_a)/y(x) > \btheta$.
	Consider the precomputed values chosen like in Algorithm~\ref{alg:precompute}.
	Then the following hold:
	\begin{enumerate}
		\item
			(a) We have $b_a \ge b_a^{I/II}$ iff we are in case I.
			(b) For $b_a^{I/II}$ we have $\alpha=\balpha$, $x_U=\bzp$, and $x_L=x_L^{I/II}$.
			(c) We have $b(x; b_a) \ge b(x; b_a^{I/II})$ iff we are in case I.
		\item
			(a) We have $b_a \le b_a^{II/III}$ iff we are in case III.
			(b) For $(b_a^{II/III})$ we have $\alpha=\balpha$ and $x_U=0$.
			(c) We have $b(x; b_a) \le b(x; b_a^{II/III}, \alpha=\balpha, x_U=0, x_L=1)$ iff we are in case III.
		\item
			Assume that we are in case II.
			(a) We have $b_a \ge b_a^{h/l}$ iff we are in case h.
			(b) For $b_a^{h/l}$ we have $\alpha=\balpha$, $x_U=x_U^{h/l}$, and $x_L=y_a$.
			(c) We have $b(x; b_a) \ge b(x; b_a^{h/l})$ iff we are in case h.
		\item 
			Assume that we are in case III.
			(a) We have $b_a \ge b_a^{H/L}$ iff we are in case H.
			(b) For $(b_a^{H/L})$ we have $\alpha=\alpha^{H/L}$, $x_U=0$, and $x_L=y_a$.
			(c) We have $b(x; b_a) \ge b(x; b_a^{H/L})$ iff we are in case H.
	\end{enumerate}
\end{lemma}

\begin{center} \hyperlink{pf:lem:precompute}{\texttt{[Link to Proof]}} \end{center}

Regarding the runtime properties of the precomputation step, we receive:

\begin{proposition}
	\label{prop:precompute-runtime}
	Algorithm~\ref{alg:precompute} only requires a constant number of basic arithmetic operations and the computation of at most one square root.
	Furthermore, correctness of the values computed by algorithm~\ref{alg:precompute} can be verified using a constant number of basic arithmetic operations and no square root.
\end{proposition}

\begin{center} \hyperlink{pf:prop:precompute-runtime}{\texttt{[Link to Proof]}} \end{center}

%Recall that Algorithm~\ref{alg:precompute} only has to be run when these parameters change, which will likely only occur very rarely.
%When this happens, there are two options for adjusting the precomputed values: either the precomputation could happen on-chain, or the precomputation might even happen off-chain. In the latter case, we can verify on-chain that the computed values are correct with further reduced computational effort.

% Moved to have this placed early.
\begin{algorithm}
	\caption{Region detection. We use “emit” as a shorthand to denote that the algorithm outputs that we are in a particular case without stopping execution of the function.}
	\label{alg:region-detection}
	\begin{algorithmic}[1]
		\Require{$\balpha, \bzp, \btheta, (x, b, y)$ such that $\btheta < b/y < 1$ and $y_a := y+x = 1$; $P$ precomputed using Algorithm~\ref{alg:precompute}}
		\Ensure{The region in which $(x, b)$ lies.}
		\If{$b \ge b(x; b_a^{I/II}, y_a=1, \alpha=\balpha, x_U=\bzp, x_L=x_L^{I/II})$}
			\State Emit case I
			%\If{$x \le \bzp$}
%			\State Let $b_a^{\text{ii}} := b + x - \frac \balpha 2 (x-\bzp)^2$
			\If{$x \le \bzp$}
				\State Emit case i
			\ElsIf{$b/y \le 1-\balpha(x-\bzp)$}\label{ln:region-detection:Iii}
				\State Emit case ii
			\Else
				\State Emit case iii
			\EndIf
			% todo writing: below, whenever we set x_L=1, do we really want b^U instead? Is this defined?
			% todo also writing: This fixed-parameter notation is new. Introduce or clear?
		\ElsIf{$b \ge b(x; b_a^{II/III}, y_a=1, \alpha=\balpha, x_U=0, x_L=1)$}
			\State Emit case II
			\If{$b \ge b(x; b_a^{h/l}, y_a=1, \alpha=\balpha, x_U=x_U^{h/l}, x_L=1)$}
				\State Emit case h
				\If{$y-b \le \frac \balpha 2 \cdot y^2$}\label{ln:region-detection:IIhi}
					\State Emit case i
				\Else
					\State Emit case ii
				\EndIf
			\Else
				\State Emit case l
				\If{$b - \btheta y \ge \frac {\theta^2} {2\balpha}$} \label{ln:region-detection:IIli}
					\State Emit case i
				\Else
					\State Emit case ii
				\EndIf
			\EndIf
		\Else
			\State Emit case ii
			\If{$b \ge b(x; b_a^{H/L}, y_a=1, \alpha=\alpha^{H/L}, x_U=0, x_L=1)$}
				\State Emit case H
			\Else
				\State Emit case L
			\EndIf
		\EndIf
	\end{algorithmic}
\end{algorithm}

\subsection{Region Detection}

Lemma~\ref{lem:precompute} provides a way to use the precomputed values to detect in which region we are along the I-III and h/l and H/L dimension. This can be done based on only the current state $(x, b, y)$ and without any knowledge of $b_a$. For example, to tell case I and II apart, by part (1c) of the lemma, we only need to compare the value $b(x)$ (which we know) to the value $b(x; b_a^{I/II})$. This latter value can be easily computed because we also know the values of $x_U$, $\alpha$, and $x_L$ corresponding to $b_a^{I/II}$ by part (1b) of the lemma.
The same reasoning applies to the other items.
Algorithm~\ref{alg:region-detection} formalizes this idea and adds an additional method that allows us to also tell case i--iii apart. In effect, Algorithm~\ref{alg:region-detection} allows us to fully reconstruct the region of a given point $(x, b, y)$ without knowledge of $b_a$ and only with a constant number of basic arithmetic operations (and in particular without computing a square root).

%The main challenge is, of course, solving the IVP. We can obviously not do anything even remotely resembling numerical integration on-chain, for various reasons (gas, complexity, precision, etc.). Thus, the IVP needs to be solved analytically for a given choice of p in closed form. This is a constraint on p, of course.
%
%In the full system, the outflow history z will of course decay over time. While one could model this explicitly, this is not necessary here because our transactions do not take actual wall time and are in particular not spread across several blocks.
%
%We may also ask for some additional fees, either to disincentivize redemptions or for other reasons. These are not contained in this analysis, but can probably be added “on top” with relative ease. Care must be taken not to unintentionally distort the redemption process. For example, if we choose p=b/y, then the reserve ratio b/y and thus the price remains constant. However, if we add a fee on top of that (i.e., p=b/y-\Delta, where \Delta>0), the reserve ratio, and thus the price, increases over time. Depending on how they are implemented, fees break path independence. Under reasonable fee structures, we however still have path deficiency cite that CFMM paper, which is likely enough for our use case.
%
%Throughout this document, we assume that b is liquid and can be easily accessed for redemptions. Any liquidity risks related to the value of b either need to be represented when calculating its value, or this model would have to be extended to account for them.

\begin{theorem}
	\label{thm:region-detection}
	Algorithm~\ref{alg:region-detection} is correct and only requires basic arithmetic operations, and only a constant number of them.
\end{theorem}

\begin{center} \hyperlink{pf:thm:region-detection}{\texttt{[Link to Proof]}} \end{center}

\subsection{Reconstructing the value of $b_a$}

Given sufficient knowledge about the region, it is now conceptually straightforward (though mathematically somewhat inconvenient) to reconstruct the precise value of $b_a$ by solving a linear or quadratic equation. This is captured in Algorithm~\ref{alg:ba-reconstruction}.

\begin{algorithm}
	\caption{Reconstruction of $b_a$}
	\label{alg:ba-reconstruction}
	\begin{algorithmic}[1]
		% todo this doesn't even need normalization. Maybe reformulate.
		\Require{$\balpha, \bzp, \btheta, (x, b, y)$ such that $\btheta < b/y < 1$ and $y_a := y+x = 1$; the region of $(x, b)$}
		\Ensure{$b_a$ such that $b(x, b_a, y_a=1) = b$}
		
		\State Let $r = b/y$.
		
		\If{case i applies}
			\State $b_a = b+x$
		\ElsIf{case I applies}
			\If{case ii applies}
				\State $b_a = b + x - \frac \balpha 2 (x-\bzp)^2$
			\Else
				\Comment{Now case iii applies.}
				% todocheck if I haven't precomputed this, or easily could.
				\State $b_a = y_a - (y_a-\bzp)(1-r) + \frac 1 {2\balpha} (1-r)^2$
			\EndIf
		\ElsIf{case II applies}
			\Comment{Now case ii applies.}
			\If{case h applies}
				\State Let $\Delta_a := \frac \balpha 2 \cdot \left(\frac 1 \alpha \cdot (1-r) + \frac 1 2 \cdot y\right)^2$
				\State $b_a = y_a - \Delta_a$
			\Else
				\Comment{Now case l applies}
				\State Let $p' := \theta \cdot (\frac \theta {2\alpha} + y)$
				\State Let $d := \theta^2 \cdot \frac 2 \alpha \cdot (b - \btheta y)$
				\State Let $\Delta_a := p' - \sqrt d$
				\State $b_a = y_a - \Delta_a$
			\EndIf
		\Else
			\Comment{Now case III ii applies.}
			\If{case H applies}
				\State Let $\Delta_a := \frac {y-b} {1-\frac {x^2} {y_a^2}}$.
				\State $b_a = y_a - \Delta_a$
			\Else
				\Comment{Now case L applies.}
				\State Let $p' := \frac 1 2 (y - b + \theta y_a)$
				\State Let $q := (y-b) \theta y_a + \frac 1 4 \theta^2 x^2$
				\State Let $\Delta_a := p' - \sqrt{p^{\prime 2} - q}$
				\State $b_a = y_a - \Delta_a$
			\EndIf
		\EndIf
	\end{algorithmic}
\end{algorithm}

\begin{theorem}
	\label{thm:ba-reconstruction}
	Algorithm~\ref{alg:ba-reconstruction} is correct and only requires a constant number of basic arithmetic operations together with at most one square root.
\end{theorem}

\begin{center} \hyperlink{pf:thm:ba-reconstruction}{\texttt{[Link to Proof]}} \end{center}

%\begin{remark}
%	To further reduce computational effort, one may opt to perform computation of the square root in Algorithm~\ref{alg:ba-reconstruction} off-chain and only verify that the result is correct on-chain.
%	The trade-off between potential gas costs and increased complexity of the mechanism should be considered carefully.
%\end{remark}

\subsection{Full Implementation}

\newcommand\deltax{X}

\begin{algorithm}
	\caption{Full implementation of the redemption mechanism}
	\label{alg:redemption}
	\begin{algorithmic}[1]
		\Require{$\balphao, \bzpo, \btheta$, the current state $(x, b, y)$, and a desired amount of redemptions $\deltax \le y$; the collection of precomputed values $P$ calculated by Algorithm~\ref{alg:precompute} for the given parameters.}
		\Ensure{A new state $(x', b', y')$ after the amount of $\deltax$ has been redeemed. The difference $b'-b$ is paid out as the redemption amount.}
		\State $x' := x + \deltax$; $y' := y - \deltax$
		\If{$\deltax = 0$}
			\State $b' := b$
		\ElsIf {$b/y \ge 1$}
			\State $b' := b - \deltax$
		\ElsIf {$b/y \le \btheta$}
			\State $b' := b - b/y \cdot \deltax$
		\Else
			\State Let $y_a := y+x$. Let $x_N := x/y_a$, $b_N := b/y_a$, $y_N := y/y_a$.
			\State Apply Algorithm~\ref{alg:region-detection} to $(x_N, b_N, y_N)$ and $P$ to determine the region of this point.
			\State Apply Algorithm~\ref{alg:ba-reconstruction} to $(x_N, b_N, y_N)$ and the region information computed in the previous step to determine $b_{a,N}$ such that $b(x_N; b_{a,N}, y_a=1) = b_N$.
			\State Let $b_a := b_{a,N} \cdot y_a$.
			\State $b' := b(x + \deltax; b_a, y_a)$
		\EndIf
	\end{algorithmic}
\end{algorithm}

Algorithm~\ref{alg:redemption} shows the full implementation of our redemption mechanism and illustrates how our reconstruction algorithm is applied.
We receive the following consistency result, which states that going to some state and then performing redemption over some amount using Algorithm~\ref{alg:redemption} is the same as going to the state corresponding to the overall redemption amount directly.
This is because Algorithm~\ref{alg:redemption} correctly reconstructs the original anchor point $(b_a, y_a)$ and then simply continues on the corresponding curve of marginal redemption prices.
This is another form of path independence (see also Section~\ref{sec:path-properties}).
The proof is immediate by the preceding theorems and is omitted.

\begin{corollary}
	Fix values for $\balphao$, $\bzpo$, $\btheta$.
	Let $b_a, y_a$ be such that $1 > b_a/y_a > \btheta$.
	Let $x \in [0, y_a]$ and let $\deltax \in [0, y_a-x]$.
	Let $(x', b', y')$ be the new state as of Algorithm~\ref{alg:redemption} when applied to the parameters, the state $(x,\, b(x; b_a, y_a),\, y(x; b_a, y_a))$, and the collection of precomputed values as of Algorithm~\ref{alg:precompute} for the parameters.
	Then $x' = x + \deltax$, $b' = b(x + \deltax; b_a, y_a)$, and $y' = y(x + \deltax; b_a, y_a)$.
\end{corollary}

%%%%%%%%%%%%%%%%%%%%%%%%%%%%%%%%%%%%%%%%%%%%%%%%%%%%%%%%%
\section{Proofs}\label{apx:proofs}

%%%%%%%%%%%%%%%%%%%%%%%%%%%%%%%%%%%%%%%%%%%%%%%%%%%%%%%%%
\subsection{Technical Lemmas}\label{apx:proofs-technical}

\paragraph{Prop.~\ref{prop:b}} \hypertarget{pf:prop:b}{}
\begin{proof}
	The first equality is just the definition. For the second equality, first assume that $x\le x_{U}$. Then $p(x')=1\,\forall x'\le x$ and thus the equality holds. For the second case, we have
	\begin{align*}
		b_{a}-\int_{0}^{x} p(x')\de x' &=b_{a}-x_{U}-\int_{x_{U}}^{x}1-\alpha(x'-x_{U})\de x' \\
		&= b_{a}-x_{U}-\int_{0}^{x-x_{U}}1-\alpha x'\,\de x' \\
		&= b_{a}-x_{U}-(x-x_{U})+\frac{\alpha}{2}(x-x_{U})^{2} \\
		&= b_{a}-x+\frac{\alpha}{2}(x-x_{U})^{2}
	\end{align*}

	Finally, if $x_{L}\le x$, we have
	\[b(x)	=b(x_{L})-\int_{0}^{x-x_{L}}r_{L}\,\de x'=r_{L}(y_{a}-x_{L})-r_{L}(x-x_{L})=r_{L}(y_{a}-x).\quad\qedhere\]
\end{proof}

\noindent\rule{0.49\textwidth}{1pt}
%%%%%%%%%%%%%%%%%%%%%%%%%%%%%%%%%%%%%%%%%%%%%%%%%%%
\paragraph{Prop.~\ref{prop:xm}} \hypertarget{pf:prop:xm}{}
\begin{proof}
	Consider the definition of $\rhou(x)$, and the implied value for $b(x)$, for the case $x_{U}\le x\le y_{a}$. Let $x'=x-x_{U}$ and assume first that $x < y_a$. We have
	\begin{align*}
		&  & \rhou(x) & =r(x)\\
		\Iff &  & 1-\alpha x' & =\frac{b_{a}-(x_{U}+x')+\alpha/2x^{\prime2}}{y_{U}-x'}\\
		\Iff &  & (y_{U}-x')(1-\alpha x') & =b_{a}-x_{U}-x'+\frac{\alpha}{2}x'^{2}\\
		\Iff &  & \frac{\alpha}{2}x'^{2}-\alpha y_{U}x'+y_{a}-b_{a} & =0\\
		\Iff &  & x' & =y_{U}\pm\sqrt{y_{U}^{2}-2/\alpha(y_{a}-b_{a})}\\
		\Iff &  & x & =y_{a}\pm\sqrt{y_{U}^{2}-2/\alpha(y_{a}-b_{a})}.
	\end{align*}
	Note that, if the discriminant is positive, then the ``+'' solution is $>y_{a}$ and thus not acceptable for $x_L$, so we only need to consider the ``-'' solution. Obviously, this exists iff (\ref{eq:xm-exists}) holds. Assuming that \eqref{eq:xm-exists} \emph{does} hold, we obviously have $x\le y_{a}$ and furthermore 		
	\begin{align*}
		&  & x=y_{a}-\sqrt{y_{U}^{2}-2/\alpha\Delta_{a}} & >x_{U}\\
		\Iff &  & y_{U}-\sqrt{y_{U}^{2}-2/\alpha\Delta_{a}} & >0,
	\end{align*}
	which is true because the radicand is $<y_{U}^{2}$ because $\Delta_{a}>0$ by our basic assumptions. Thus, this $x$ serves the role of $x_{L} := x$, and it is unique by the above.
	
	If for the above choice of $x_L$ we have $x_L < y_a$, then the identity $p(x_L)=r(x_L)$ follows by construction.
	Consider now the case where $x_L = y_a$. This is the case 
	iff the above discriminant is $0$, i.e., iff $\alpha=2\frac{y_{a}-b_{a}}{\left(y_{a}-x_{U}\right)^{2}}$. In this case, it is easy to see that $b(x_{L})=y(x_{L})=0$ and we can use L'Hospital's rule to compute
	$$r_{L}=\lim_{x\to y_{a}}\frac{b(x)}{y(x)}=\lim_{x\to y_{a}}\frac{b_{a}-x+\alpha/2(x-x_{U})^{2}}{y_{a}-x}=\frac{-1+\alpha(x-x_{U})}{-1}=1-\alpha(y_{a}-x_{U})=p(x_{L}),$$
	so the identity $r_{L}=p(x_{L})$ still holds.
	
	The formulas for $r_{L}$ and $b_{L}$ now simply follow from the fact that $r_{L}=p(x_{L})$ and $b_{L}=r_{L}\cdot y(x_{L})$.
	
	Finally, consider the case where (\ref{eq:xm-exists}) does not hold and where we choose $x_{L}:=\infty$ to define $p$. Then by applying Prop.~\ref{prop:b} to $x:=y_{a}$ we receive
	$$b(y_{a})=b_{a}-y_{a}+\frac{\alpha}{2}(y_{a}-x_{U})^{2}<0,$$
	where the inequality is by the assumption that (\ref{eq:xm-exists}) does not hold. By continuity of $b$, there exists $x<y_{a}$ such that $b(x)=0$.
\end{proof}

\noindent\rule{0.49\textwidth}{1pt}
%%%%%%%%%%%%%%%%%%%%%%%%%%%%%%%%%%%%%%%%%%%%%%%%%%%
\paragraph{Lemma~\ref{lem:thetafloor-prep}} \hypertarget{pf:lem:thetafloor-prep}{}
\begin{proof}
By Proposition~\ref{prop:xm} we have that $r_{L}\ge\btheta$ iff
\begin{align*}
	&  & 1-\alpha(x_{L}-x_{U}) & \ge\btheta\\
	\Iff &  & 1-\alpha\left(y_{U}-\sqrt{y_{U}^{2}-2/\alpha\Delta_{a}}\right) & \ge\btheta\\
	\Iff &  & \alpha\sqrt{y_{U}^{2}-2/\alpha\Delta_{a}} & \ge\alpha y_{U}-\itheta\\
	\Iff &  & \alpha y_{U}-\itheta & \le0\\
	\text{or} &  & \alpha^{2}\left(y_{U}^{2}-2/\alpha\Delta_{a}\right) & \ge\left(\alpha y_{U}-\itheta\right)^{2}.
\end{align*}
The conclusion now follows via another simple algebraic transformation.
\end{proof}

\noindent\rule{0.49\textwidth}{1pt}
%%%%%%%%%%%%%%%%%%%%%%%%%%%%%%%%%%%%%%%%%%%%%%%%%%%
\paragraph{Prop.~\ref{prop:alpha}} \hypertarget{pf:prop:alpha}{}
\begin{proof}
	It is easy to see that the transition is continuous and thus well-defined;
	more in detail, in the edge case $r_{a}=\frac{1+\btheta}{2}$, we
	have $\halpha_{H}=\halpha_{L}=\itheta/y_{a}.$
	
	By the discussion at the beginning of this section, we only need to
	consider $x_{U}=0$. $\halpha$ is the maximal $\alpha$ such that
	\eqref{eq:xm-exists} holds and one of (TL) or (TH) hold (with $x_{U}=0$). Note that
	\eqref{eq:xm-exists} alone puts a bound on $\alpha$ and the right-hand side of \eqref{eq:xm-exists}
	is equal to $\halpha_{H}$ for $x_{U}=0$. Thus, whenever it is possible
	to choose $\alpha=\halpha_{H}$, this is minimal. We can choose $\alpha=\halpha_{H}$
	if (TH) holds for $\halpha$, i.e., if
	\begin{align*}
		&  & 2\frac{1-r_{a}}{y_{a}}\cdot y_{a} & \le\itheta\\
		\Iff &  & r_{a} & \ge\frac{1+\btheta}{2}.
	\end{align*}
	The equivalence immediately follows from the definition of $\itheta = 1-\btheta$.
	If this inequality does not hold, then there is no $\alpha$ that
	satisfies both \eqref{eq:xm-exists} and (TH). (observe that the two limit $\alpha$
	in different directions!)
	
	Assume next that $r_{a}\le\frac{1+\btheta}{2}$. Then $\alpha$ must
	be chosen minimal such as to satisfy \eqref{eq:xm-exists} and (TL). The minimal $\alpha$
	satisfying (TL) with $x_{U}=0$ is obviously $\halpha_{L}$. It remains
	to show that this also satisfies \eqref{eq:xm-exists}, i.e., that
	\begin{align*}
		&  & \frac{1}{2}\frac{\itheta^{2}}{b_{a}-\btheta y_{a}} & \ge2\frac{y_{a}-b_{a}}{y_{a}^{2}}\\
		\Iff &  & \frac{1}{2}\frac{\itheta^{2}}{r_{a}-\btheta} & \ge2(1-r_{a})\\
		\Iff &  & \frac{1}{4}\itheta^{2} & \ge(1-r_{a})(r_{a}-\btheta).
	\end{align*}
	The equivalences are by definition of $r_a=b_a/y_a$ and straightforward transformation.
	Let $\zeta=\frac{1+\btheta}{2}-r_{a}$. By assumption, $\zeta\ge0$
	and we have $1-r_{a}=\itheta/2+\zeta$ and $r_{a}-\btheta=\itheta/2-\zeta$.
	Thus, the above inequality is equivalent to
	\[
	\frac{1}{4}\itheta^{2}\ge(\itheta/2+\zeta)(\itheta/2-\zeta)=\frac{\itheta^{2}}{4}-\zeta^{2},
	\]
	which is obviously true.
\end{proof}

\noindent\rule{0.49\textwidth}{1pt}
%%%%%%%%%%%%%%%%%%%%%%%%%%%%%%%%%%%%%%%%%%%%%%%%%%%
\paragraph{Prop.~\ref{prop:zp}} \hypertarget{pf:prop:zp}{}
\begin{proof}
	We proceed similarly to Proposition~\ref{prop:alpha}. Again, it
	is easy to see that in the edge case $\alpha\Delta_{a}=\frac{1}{2}\itheta^{2}$
	we have $\hzph=\hzpl=y_{a}-\frac{\itheta}{\alpha}$. We need to choose
	$x_{U}$ such as to satisfy \eqref{eq:xm-exists} and one of (TH) or (TL), this time
	without assuming $x_{U}=0$ of course. \eqref{eq:xm-exists} is equivalent to
	\[
	x_{U}\le y_{a}-\sqrt{2\frac{\Delta_{a}}{\alpha}}.
	\]
	
	If we can choose $x_{U}$ equal to the right-hand side such that (TH)
	is satisfied for this choice, then this is optimal. This is the case
	if
	\begin{align*}
		&  & \alpha\cdot\sqrt{2\frac{\Delta_{a}}{\alpha}} & \le\itheta\\
		\Iff &  & \alpha\Delta_{a} & \le\frac{1}{2}\itheta^{2}.
	\end{align*}
	If this does not hold, no $x_{U}$ satisfies both \eqref{eq:xm-exists} and (TH).
	
	Assume now that $\alpha\Delta_{a}\ge\frac{1}{2}\itheta^{2}$. We need
	to satisfy (\eqref{eq:xm-exists} and) (TL). Clearly, (TL) holds iff
	\[
	x_{U}\ge\frac{b_{a}-\btheta y_{a}}{\itheta}-\frac{1}{2\alpha}\itheta^{2}=y_{a}-\frac{\Delta_{a}}{\itheta}-\frac{1}{2\alpha}\itheta=\hzpl.
	\]
	It remains to check that $x_{U}=\hzpl$ satisfies \eqref{eq:xm-exists}. This is the
	case iff
	\begin{align*}
		&  & y_{U}^{2} & \ge2\frac{\Delta_{a}}{\alpha}\\
		&  & \left(y_{a}-\left(y_{a}-\frac{\Delta_{a}}{\itheta}-\frac{1}{2\alpha}\itheta\right)\right)^{2} & \ge2\frac{\Delta_{a}}{\alpha}\\
		&  & \left(\frac{\Delta_{a}}{\itheta}+\frac{1}{2\alpha}\itheta\right)^{2}-2\frac{\Delta_{a}}{\alpha} & \ge0\\
		&  & \left(\frac{\Delta_{a}}{\itheta}-\frac{1}{2\alpha}\itheta\right)^{2} & \ge0,
	\end{align*}
	which is obviously true. The last line follows using the binomial
	formulae since $2\frac{\Delta_{a}}{\alpha}=2\cdot2\cdot\frac{\Delta_{a}}{\itheta}\cdot\frac{1}{2\alpha}\itheta$.
\end{proof}

\noindent\rule{0.49\textwidth}{1pt}

%%%%%%%%%%%%%%%%%%%%%%%%%%%%%%%%%%%%%%%%%%%%%%%%%%%%%%%%%
\subsection{Main Results}\label{apx:proofs-main}

\paragraph{Proposition~\ref{prop:HLrm}} \hypertarget{pf:HLrm}{}
\begin{proof}
        In case II h, we have $x_{U}=\hzp=\hzph$ and thus (by the proof of
	Proposition~\ref{prop:zp}) Inequality~\eqref{eq:xm-exists} holds with equality. It follows
	immediately from Proposition~\ref{prop:xm} that this implies $x_{L}=y_{a}$.
	Likewise in case III H.
		In case II l, $x_{U}=\hzp=\hzpl$
	and thus (by the proof of Proposition~\ref{prop:zp}) (TL) holds
	with equality. By the proof of Lemma~\ref{lem:thetafloor-prep},
	this immediately implies $r_{L}=\btheta$.
        Likewise for case III L.
\end{proof}

\noindent\rule{0.49\textwidth}{1pt}
%%%%%%%%%%%%%%%%%%%%%%%%%%%%%%%%%%%%%%%%%%%%%%%%%%%
\paragraph{Corollary~\ref{cor:iii}} \hypertarget{pf:cor:iii}{}
\begin{proof}
	By Proposition~\ref{prop:HLrm}, since we're in case II or III we
	have $x_{L}=y_{a}$ or $r_{L}=\btheta$. Since we're in case iii,
	$x_{L}\le x<y_{a}$ and thus $x_{L}=y_{a}$ is not possible. We must
	therefore have $r_{L}=\btheta$. And, again since $x\ge x_{L}$, $r(x)=r_{L}$.
\end{proof}

\noindent\rule{0.49\textwidth}{1pt}
%%%%%%%%%%%%%%%%%%%%%%%%%%%%%%%%%%%%%%%%%%%%%%%%%%%
\paragraph{Theorem~\ref{thm:mon-b0}} \hypertarget{pf:mon-b0}{}
\begin{proof}
	If $S \subseteq [0,y_a] \times [0,y_a]$ is a set of $(b_a, x)$ pairs, define the \emph{$b_a$-interior} of $S$ as the set
	$\{(b_a, x) \in S \mid \exists \eps > 0: (b_a + \eps, x), (b_a - \eps, x) \in S\}$.
        It suffices to show the following: for any point $(b_a, x)$ that lies within the $b_a$-interior of any of the sets of pairs $(b_a, x)$ defined by the following (topologically closed) conditions,
	if $r(x;b_{a})>\btheta$, then we have \[\frac{\de b(x;b_{a})}{\de b_{a}} > 0.\]
        We will perform case distinction in such a way that this derivative always exists.
	Assume that $r(x;b_{a})>\btheta$.
	The statement is easy to see in the following cases:
	
	\textbf{Case i:} This case is trivial because here, $b(x;b_{a})=b_{a}-x$ and so $\frac{\de b(x;b_{a})}{\de b_{a}} = 1 > 0$.

	\textbf{Case I ii and I iii:} Here
	the statement follows immediately from Prop.~\ref{prop:b} and Prop.~\ref{prop:xm} Part~\ref{itm:xm:fmls}. Note that whenever we are in case I, the parameters $\alpha=\balpha$ and $x_{U}=\bzp$ are constant.

	\textbf{Case II iii and III iii:} By Corollary~\ref{cor:iii}, we do not need to discuss these cases.

	It remains to show that the statement holds for case II ii and III ii, which requires some calculation. We distinguish four cases: II h ii, II l ii, III H ii, and III L ii. We combine Prop.~\ref{prop:b} with Propositions \ref{prop:zp} and \ref{prop:alpha}, respectively, to compute the partial derivatives.

	\textbf{Case II h ii:} Here we have $x_U = \hzph$ and
	\begin{align*}
		\frac{\de b(x;b_{a})}{\de b_{a}} & =\frac{\de}{\de b_{a}}\left[b_{a}-x+\frac{\balpha}{2}(x-\hzph(b_{a}))^{2}\right]\\
		& =1+\frac{\balpha}{2}\cdot2(x-\hzph(b_{a}))\cdot(-1)\cdot\frac{\de}{\de b_{a}}\hzph(b_{a})\\
		& =1-\balpha(x-\hzph(b_{a}))\cdot\frac{\de}{\de b_{a}}\left[y_{a}-\sqrt{2\frac{y_{a}-b_{a}}{\balpha}}\right]\\
		& =1-\balpha(x-\hzph(b_{a}))\cdot(-1)\cdot\frac{1}{2}\frac{1}{\sqrt{2\frac{y_{a}-b_{a}}{\balpha}}}\cdot(-\frac{2}{\balpha})\\
		& =1-\frac{x-\hzph(b_{a})}{y_{a}-\hzph(b_{a})}>0.
	\end{align*}
	The last equality is by definition of $\hzph(b_a)$ and
	the inequality is because $x<y_{a}$ and $x,y_{a}>\hzph(b_{a})$ by assumption.

	\textbf{Case II l ii:} Here we have $x_{U}=\hzpl$ and
	\[
	\frac{\de}{\de b_{a}}\hzpl(b_{a})=\frac{\de}{\de b_{a}}\left[y_{a}-\frac{\Delta_{a}}{\itheta}-\frac{1}{2\balpha}\itheta\right]=\frac{1}{\itheta}.
	\]
	Thus,
	\begin{align*}
		&  & \frac{\de b(x;b_{a})}{\de b_{a}}=1-\balpha(x-\hzph(b_{a}))\frac{1}{\itheta} & >0\\
		\Iff &  & p(x;b_{a})=1-\balpha(x-\hzph(b_{a})) & >\btheta
	\end{align*}
	This is true because $x\le x_{L}$ by assumption and thus
	$p(x)\ge r(x)$ and we have $r(x)>\btheta$ by assumption.
	%\todo{This case is interesting. Any further conceptual insights from this?}
	
	\textbf{Case III H ii:} In this case we have $x_{U}=0$ and $\alpha=\halpha_{H}(b_{a})$ and we have
	\begin{align*}
		\frac{\de b(x;b_{a})}{\de b_{a}} & =\frac{\de}{\de b_{a}}\left[b_{a}-x+\frac{\halpha_{H}(b_{a})}{2}x^{2}\right]\\
		& =1+\frac{1}{2}x^{2}\frac{\de}{\de b_{a}}\halpha_{H}(b_{a})\\
		& =1+\frac{1}{2}x^{2}\cdot2\cdot\left(-\frac{1}{y_{a}^{2}}\right)\\
		& =1-\left(\frac{x}{y_{a}}\right)^{2}>0
	\end{align*}
	since $x<y_{a}$.
	%\todo{Note the same pattern as above: Once we're in this case, it's just because $x<y_{a}$, which is trivial. Any implications?}

	\textbf{Case III L ii:} Here we have
	\[
	\frac{\de}{\de b_{a}}\halpha_{L}(b_{a})=\frac{1}{2}\cdot\itheta^{2}\cdot(-1)\cdot\frac{1}{\left(b_{a}-\btheta y_{a}\right)^{2}}=-\frac{1}{2}\frac{\itheta^{2}}{\left(b_{a}-\btheta y_{a}\right)^{2}}=-\halpha_L(b_a)\frac{1}{b_{a}-\btheta y_{a}},
	\]
	where the last equality is by definition of $\halpha_L(b_a)$.
	Therefore, we have (writing just $\alpha$ for $\halpha_{L}(b_{a})$ in the interest of brevity)
	\begin{align*}
		&  & \frac{\de b(x;b_{a})}{\de b_{a}} & =1+\frac{1}{2}x^{2}\cdot\left(-\alpha\frac{1}{b_{a}-\btheta y_{a}}\right)\\
		&  &  & =1-\alpha\cdot\frac{1}{2}x^{2}\frac{1}{b_{a}-\btheta y_{a}}
		.
	\end{align*}
	To see that this is positive, first note that, by Proposition~\ref{prop:HLrm}
	and Proposition~\ref{prop:xm}, $\btheta=r_{L}=1-\alpha x_{L}$ and
	thus $x_{L}=\itheta/\alpha$. Since we're in case ii and we have $r(x)>\btheta=r_{L}$, we must
	have $x<x_{L}=\itheta/\alpha$ and thus
	\begin{align*}
		1-\alpha\cdot\frac{1}{2}x^{2}\frac{1}{b_{a}-\btheta y_{a}} & >1-\alpha\cdot\frac{1}{2}\frac{\itheta^{2}}{\alpha^{2}}\frac{1}{b_{a}-\btheta y_{a}}\\
		& =1-\frac{1}{\alpha}\cdot\frac{1}{2}\frac{\itheta^{2}}{b_{a}-\btheta y_{a}}=1-\frac{1}{\alpha}\cdot\alpha=0
	\end{align*}
	as required.
	This concludes the proof.
\end{proof}

\noindent\rule{0.49\textwidth}{1pt}
%%%%%%%%%%%%%%%%%%%%%%%%%%%%%%%%%%%%%%%%%%%%%%%%%%%
\paragraph{Theorem~\ref{thm:path_ind}} \hypertarget{pf:path_ind}{}
\begin{proof}
	Let $(x^0, b^0, y^0)$ be a solution to the initial value problem at $S_0$. Note that $x^0,b^0,y^0$ are functions of $x$. Then the functions $b^{X}(x):=b^{0}(X+x)$ (and analogously for $y$) form the solution to the IVP at $S_{X}$. To see this, note that they satisfy the differential equations (because $(x^0,b^{0},y^{0})$ do and translation by $X$ does not affect the derivatives) and they satisfy the initial values by choice of $S_{X}$. We now have $b_{X,Y}=b^{X}(Y)=b^{0}(X+Y)=b_{X+Y}$, and likewise for the other two. This is easy to see for $y$ because it is simply $x$ plus a constant; however our argument does not depend on this fact. Overall, $S_{X,Y}=S_{X+Y}$. For the redemption amounts, we now have $P_{X}+P_{X,Y}=(b_{0}-b_{X})+(b_{X}-b_{X,Y})=b_{0}-b_{X,Y}=b_{0}-b_{X+Y}=P_{X+Y}$.
\end{proof}

\noindent\rule{0.49\textwidth}{1pt}
%%%%%%%%%%%%%%%%%%%%%%%%%%%%%%%%%%%%%%%%%%%%%%%%%%%
\paragraph{Lemma~\ref{lem:path-deficiency1}} \hypertarget{pf:lem-path-deficiency1}{}
\begin{proof}
	Separate into two cases: (i) when $\frac{\de f}{\de t} \geq 0$, and (ii) when $\frac{\de f}{\de t} < 0$. These correspond to the $x$ value increasing or decreasing respectively along the path $f$. Suppose we are at the current state $(f(t), r_{f,\mathfrak r}(t))$ and that this point is on the reserve ratio curve $\hat {\mathfrak r} \in \mathcal R$.
	In (i), $x$ is increasing (redemption operation), and so (\ref{eq:redeem-odes-fee}) applies. Taking derivative of $r(x)$,
	\begin{align*}
		\frac{\de r}{\de x} &= \frac{\frac{\de b}{\de x} y(x) + b(x)}{y^2(x)} \\
		&= \frac{r(x) - \rho(x, b(x), y(x)) + \gamma(x,b(x),y(x))}{y(x)}.
	\end{align*}
	The derivative is greater (in this case less negative) when $\gamma > 0$. When $\gamma = 0$, we have the system (\ref{eq:odes}), and the reserve ratio follows $\hat {\mathfrak r} (x)$.
	
	In (ii), $x$ is decreasing (minting operation), and so (\ref{eq:mint-odes}) applies. Taking derivative of $r(x)$,
	$$\frac{\de r}{\de x} = \frac{ r(x) - \varphi(x,b(x),y(x)) } {y(x)}.$$
	Since $\varphi \geq 1$, we have $\varphi \geq \rho$. And so the derivative is greater than the corresponding derivative in (\ref{eq:odes}), which would keep us on the reserve ratio curve $\hat {\mathfrak r}$.
	
	In both cases, the path does not bring us to a region below $\hat {\mathfrak r}$ in $(x,r)$ space. Since Theorem~\ref{thm:mon-b0} gives us that $r_a$ (via $b_a$) is monotonic in $r$, it must be non-decreasing along this path.
\end{proof}

\noindent\rule{0.49\textwidth}{1pt}
%%%%%%%%%%%%%%%%%%%%%%%%%%%%%%%%%%%%%%%%%%%%%%%%%%%
\paragraph{Theorem~\ref{thm:path-deficiency1}} \hypertarget{pf:path-deficiency1}{}
\begin{proof}
	From Lemma~\ref{lem:path-deficiency1}, we know that $r_a$ is weakly increasing along the path. Since reserve ratio curves are point-wise non-decreasing in their parameter $r_a$, the result follows immediately.
\end{proof}

\noindent\rule{0.49\textwidth}{1pt}
%%%%%%%%%%%%%%%%%%%%%%%%%%%%%%%%%%%%%%%%%%%%%%%%%%%
\paragraph{Theorem~\ref{thm:path-deficiency2}} \hypertarget{pf:path-deficiency2}{}
\begin{proof}
	This setup is equivalent to changing the RHS of (\ref{eq:odes}) to $(1-\eps)\rho(\cdot)$, which is a constant scaling. From linearity of integration, the only thing that changes structurally about the system is the hyperparameters (static parameters), which can be thought of as mapping in the following ways:
	\begin{itemize}
		\item a thus far implicit parameter specifying the $\$1$ target $\mapsto$ $(1-\eps)$ target,
		\item $\balpha \mapsto (1-\eps)\balpha$,
		\item $\btheta \mapsto (1-\eps)\btheta$.
	\end{itemize}
	A useful interpretation is that this scaling can be effectively `reversed' by normalizing the system (in this case the redemption system in isolation) back to a $\$1$ target. The underlying hyperparameters effectively shift slightly from the reverse scaling, but we are left with the same underlying structure and machinery.
	Since we proved the above results for all hyperparameter values --and it is simple to add in the thus far implicit hyperparameter specifying the target through this normalization argument-- the results still stand even though the hyperparameters shift.
	In particular, we have retained path independence for the redemption system in (\ref{eq:redeem-odes-fee}).
	
	To show the second result of the theorem, we need to add in the minting process as described in Section~\ref{sec:extend-fees-minting}. Consider a sole trader interacting with the system. If they desire a net redemption from the protocol, then by the path independence of the redemption curve, there is no incentive to subdivide the redemptions into several smaller redemptions. The remaining possibility is that the trader subdivides the net redemption into a sequence of redemptions and minting that nets to the desired redemption. It is simple to see that this is not profitable for two consecutive mint and redeem trades since the integrands $(1-\eps)\rho(\cdot) \leq \varphi(\cdot)$ in this region (reserve ratio $\leq 1$). In words, this is not profitable because the trader must pay a non-negative spread between minting and redeeming in backtracking in a path in $x$, and so it is more profitable not to backtrack (it is always better to cancel out a mint with a redemption). From a simple induction then, the best option for the trader is to choose the net redemption desired in aggregate.
\end{proof}

%%%%%%%%%%%%%%%%%%%%%%%%%%%%%%%%%%%%%%%%%%%%%%%%%%%%%%%%%%%%%%%%%%%%%%%
\subsection{Proofs regarding Implementation (Appendix~\ref{apx:implementation})}
\label{apx:proofs-implementation}
%%%%%%%%%%%%%%%%%%%%%%%%%%%%%%%%%%%%%%%%%%%%%%%%%%%%%%%%%%%%%%%%%%%%%%%

\paragraph{Lemma~\ref{lem:scaling}} \hypertarget{pf:lem:scaling}{}
\begin{proof}
	The lemma follows from the fact that all dynamic parameters and the input and output of the function $p$ scale in $\zeta$.
	More in detail, the statement follows from the following observations.	\begin{itemize}
		\item We have $\alpha(\zeta b_a, \zeta y_a) = 1/\zeta \cdot \alpha(b_a, y_a)$. To see this, note that $r_a(\zeta b_a, \zeta y_a) = r_a(b_a, y_a)$ and consider Proposition~\ref{prop:alpha} to see that $\halpha(\zeta b_a, \zeta y_a) = 1/\zeta \cdot \halpha(b_a, y_a)$. Note further that  $\balpha(\zeta y_a) = \balphao \cdot 1/(\zeta y_a) = 1/\zeta \cdot \balpha(y_a)$ by choice of $\balpha$.

		\item Similarly, we have $x_U(\zeta b_a, \zeta y_a) = \zeta x_U(b_a, y_a)$. To see this, note that $\Delta_a(\zeta b_a, \zeta y_a) = \zeta \Delta_a(b_a, y_a)$ and apply the preceding statement on $\alpha$ to Proposition~\ref{prop:zp} to receive that $\hzp(\zeta b_a, \zeta y_a) = \zeta \hzp(b_a, y_a)$. Also note that $\bzp(\zeta y_a) = \bzpo \zeta y_a = \zeta \bzp(y_a)$ by choice of $\bzp$.

		\item It is now easy to see that that $x_L(\zeta b_a, \zeta y_a) = \zeta x_L(b_a, y_a)$ and $r_L(\zeta b_a, \zeta y_a) = r_L(b_a, y_a)$ using Proposition~\ref{prop:xm} and then one can directly verify the statement of the lemma using Proposition~\ref{prop:b}.\qedhere
	\end{itemize}
\end{proof}

\noindent\rule{\textwidth}{1pt}
%%%%%%%%%%%%%%%%%%%%%%%%%%%%%%%%%%%%%%%%%%%%%%%%%%%
\paragraph{Lemma~\ref{lem:ba_for_xu}} \hypertarget{pf:lem:ba_for_xu}{}

\begin{proof}
	We distinguish cases h and l. It is easy to see that
	\begin{alignat*}{4}
		\hzph &= z &\qquad\Iff\qquad b_a &= y_a - \frac \alpha 2 y_z^2
		\\
		\hzpl &= z &\qquad\Iff\qquad b_a &= y_a - \theta y_z + \theta^2 \cdot \frac 1 {2\alpha}
		.
	\end{alignat*}
	% todo simplification: do we really need to plug it in?? We assume an identity!!
	By plugging the above formulas for $b_a$ into the condition $\alpha \Delta_a \le \frac 1 2 \theta^2$ (which defines cases h and l), we receive by simple algebraic transformations:
	\begin{alignat*}{4}
		\text{Case h applies}
		&\,\wedge \hzph = z
		&&\qquad\Then\qquad 1-\alpha y_z \ge \btheta
		\\
		\text{Case l applies}
		&\,\wedge \hzpl = z
		&&\qquad\Then\qquad 1-\alpha y_z \le \btheta
		.
	\end{alignat*}
	This implies the statement because exactly one of the cases on the right-hand side (or both in case of equality) must be satisfied.
	% todo dbl check above calcs. Should be fine tho
\end{proof}

\noindent\rule{\textwidth}{1pt}
%%%%%%%%%%%%%%%%%%%%%%%%%%%%%%%%%%%%%%%%%%%%%%%%%%%
\paragraph{Lemma~\ref{lem:precompute}} \hypertarget{pf:lem:precompute}{}

\begin{proof}
	1.
	First note that $b_a^{I/II}$ always exists because for $b_a \to 1$, we will eventually arrive in case I. By construction, both case I and II hold at $b_a^{I/II}$.
	Statement (a) now follows by monotonicity of the choice of parameters.
	Statement (b) follows by construction because case I applies at $b_a^{I/II}$.
	Statement (c) follows from statement (a) because, by strict monotonicity (Theorem~\ref{thm:mon-b0}), $b(x; b_a) \ge b(x; b_a^{I/II})$ iff $b_a \ge b_a^{I/II}$.

	The arguments for the other constructions are analogous.

	The equations for $b_a^{h/l}$ and $b_a^{H/L}$ immediately result from Propositions \ref{prop:zp} and \ref{prop:alpha}, respectively, by converting the conditions distinguishing their respecitive cases into equalities.
\end{proof}

\noindent\rule{\textwidth}{1pt}
%%%%%%%%%%%%%%%%%%%%%%%%%%%%%%%%%%%%%%%%%%%%%%%%%%%
\paragraph{Prop.~\ref{prop:precompute-runtime}} \hypertarget{pf:prop:precompute-runtime}{}
\begin{proof}
	The statement regarding operations is obvious. The algorithm requires the computation of one square root, in line~2.
	Note that no square root is required for the computation of $x_U^{h/l}$ because the calculation of $\hzp$ for case~l can be applied by choice of $b_a^{h/l}$.
	Candidate values for these numbers can be easily verified without taking a square root.
	To verify the only non-trivial value $x_L^{I/II}$, it is sufficient to check whether $p(x; b_a^{I/II}, y_a=1, \alpha=\balpha, x_U=\bzp) = r(x; b_a^{I/II}, y_a=1, \alpha=\balpha, x_U=\bzp)$, which does not require a square root.
\end{proof}

\noindent\rule{\textwidth}{1pt}
%%%%%%%%%%%%%%%%%%%%%%%%%%%%%%%%%%%%%%%%%%%%%%%%%%%
\paragraph{Theorem~\ref{thm:region-detection}} \hypertarget{pf:thm:region-detection}{}
\begin{proof}
	The statement regarding operations is obvious. Note in particular that output values of the function $b$ are only ever calculated while providing fixed values for the dynamic parameters $\alpha$, $x_U$, $x_L$.
	Correctness of detection of cases I--III and h/l and H/L immediately follows from Lemma~\ref{lem:precompute}.

	It remains to show that cases i--iii are detected correctly. We go through the different cases one-by-one.

	First assume that $(x, b, y)$ lies in case I.
	Obviously we're in case i iff $x\le x_U = \bzp$. Assume now that we are in case I but not in case I i.
	By choice of $x_L$ and definition of the regions ii and iii, we are in region ii iff $r \le p^U(x)$, where $p^U(x)$ is defined in \eqref{eq:rhou}. The expression $p^U(x)$ implicitly depends $b_a$ via the choice of $\alpha$ and $x_U$. However, since we are in region I, we know that $\alpha=\balpha$ and $x_U=\bzp$. Furthermore, $x \ge x_U$ and thus $p^U(x) = 1-\balpha(x-\bzp)$. Therefore, line~\ref{ln:region-detection:Iii} appropriately checks the condition $r \le p^U(x)$.

	Assume now that $(x, b, y)$ does not lie in case I. Case~iii is impossible in case II or III by Corollary~\ref{cor:iii}. Case~i is further trivial in case III because $x_U=0$ in this case. Thus, these cases do not have to be checked.
	It remains to correctly distinguish case i and ii given that we are in case II.

	Now assume that $(x, b, y)$ lies in case II h. It lies in case i iff $x \le x_U = \hzph = y_a - \sqrt {2 \frac {\Delta_a} \balpha}$, where the last equality follows from Proposition~\ref{prop:zp} and the fact that we are in case II h. It is easy to see that this is equivalent to
	\begin{equation}
		\Delta_a \le \frac \balpha 2 y^2
		\label{eq:IIhiTrue}
		.
	\end{equation}
	Line~\ref{ln:region-detection:IIhi} checks this condition, but with $\Delta_a = y_a - b_a$ replaced by $\Delta := y - b$. We show that the two conditions are equivalent, given that we are in case II h.
	First, note that
	\[
	\Delta_a - \Delta = b - (b_a - x) =
	\begin{cases}
		0 &\text{in case i}
		\\
		\frac \balpha 2 (x - x_U)^2 &\text{in case ii}.
	\end{cases}
	\]
	In any case, $\Delta_a-\Delta \ge 0$ and thus $\Delta \le \Delta_a$. Therefore, \eqref{eq:IIhiTrue} implies line~\ref{ln:region-detection:IIhi}.

	For the other direction, assume that \eqref{eq:IIhiTrue} does not hold, so that $\Delta_a > \frac \balpha 2 y^2$ and we are in case ii. We have, via the above, that
	% TODO Wondering if this can be proven more easily. E.g. by exploiting uniqueness or so.
	% I don't remember how I got these formulas back in July. Perhaps I took a more clever path back then. :/
	\begin{align*}
		\Delta &= \Delta_a - \frac \balpha 2 (x-x_U)^2
		= \Delta_a - \frac \balpha 2 \left(\sqrt{\frac 2 \balpha \Delta_a} - y\right)^2
		\\&= \Delta_a - \left(\Delta_a - \balpha \sqrt{\frac 2 \balpha \Delta_a} y + \frac \balpha 2 y^2\right)
		= \balpha \sqrt{\frac 2 \balpha \Delta_a} y - \frac \balpha 2 y^2
		\\&> \balpha y^2 - \frac \balpha 2 y^2 = \frac \balpha 2 y^2,
	\end{align*}
	where the last line follows by assumption.

	Assume finally that $(x, b, y)$ lies in case II l. It lies in case i iff $x \le x_U = \hzpl = y_a - \frac {\Delta_a} \theta - \frac \theta {2\alpha}$, where the last equality again follows from Proposition~\ref{prop:zp} and the fact that we are in case II l. This is equivalent to
	\begin{equation}
		\Gamma_a := b_a - x - \btheta y \ge \frac {\theta^2} {2\alpha}.
		\label{eq:IIliTrue}
	\end{equation}
	This follows by simple transformation, noting that $\theta y - \Delta_a = \Gamma_a$.
	Line~\ref{ln:region-detection:IIli} checks this condition, but with $\Gamma_a$ replaced by $\Gamma := b - \btheta y$. Again, we show that the two conditions are equivalent, given that we are in case II l.
	For the first direction of the equivalence, note that $\Gamma_a - \Gamma = \Delta - \Delta_a \le 0$ and thus $\Gamma_a \le \Gamma$ and \eqref{eq:IIliTrue} implies line~\ref{ln:region-detection:IIli}.

	For the other direction, assume that \eqref{eq:IIliTrue} does not hold, so that $\Gamma_a < \frac {\theta^2} {2\balpha}$ and we are in case ii.
	% TODO Again, I have the suspicion that this can be proven in a much simpler way and that I may actually have come here via that path. But I can't remember & it might as well have been a brainfart back then.
	Observe that
	\begin{align*}
		x-x_U &= x - y_a + \frac {\Delta_a} \theta + \frac \theta {2\balpha}
		= \frac 1 \theta \left(\frac {\theta^2} {2\balpha} + \Delta_a - \theta y\right)
		\\&
		= \frac 1 \theta \left(\frac {\theta^2} {2\balpha} - \Gamma_a\right)
		.
	\end{align*}

	We thus have
	\begin{align*}
		\Gamma &= \Gamma_a + \Delta_a - \Delta
		= \Gamma_a + \frac \balpha 2 (x-x_U)^2
		= \Gamma_a + \frac \balpha {2 \theta^2} \left(\frac {\theta^2} {2\balpha} - \Gamma_a\right)^2
		\\&= \Gamma_a + \frac {\theta^2} {8\balpha} - \frac 1 2 \Gamma_a + \frac {\balpha} {2\theta^2}\Gamma_a^2
		= \frac 1 2 \Gamma_a + \frac {\theta^2} {8\balpha} + \frac {\balpha} {2\theta^2}\Gamma_a^2
		\\&< \frac 1 2 \cdot \frac {\theta^2} {2\balpha} + \frac 1 4 \cdot \frac {\theta^2} {2\balpha} + \frac 1 4 \cdot \frac {\theta^2} {2\balpha}
		= \frac {\theta^2} {2\balpha}
		.
		\qedhere
	\end{align*}
\end{proof}

\noindent\rule{\textwidth}{1pt}
%%%%%%%%%%%%%%%%%%%%%%%%%%%%%%%%%%%%%%%%%%%%%%%%%%%
\paragraph{Theorem~\ref{thm:ba-reconstruction}} \hypertarget{pf:thm:ba-reconstruction}{}
\begin{proof}
	The statement regarding operations is obvious.
	Towards correctness, all calculations result from Propositions \ref{prop:b}, \ref{prop:xm}, \ref{prop:alpha}, and \ref{prop:zp}, by replacing the value of $\alpha$ and $x_U$ (and $x_L$ in case of case I iii) into the equation for $b(x)$ from Proposition~\ref{prop:b}.
	Within each region, this yields a smooth equation (i.e., without a case distinction or a maximum/minimum) that has degree 1 or 2 in $b_a$ can therefore be solved for $b_a$ easily.

	% TODO We don't wanna do the whole thing here, right??

	The only part that requires further discussion are the equations for $\Delta_a$ in cases II l ii and III L ii. These calculatinos result from a quadratic equation each, which has two solutions (unless the radicand is zero). We show why the equation always has a solution and only the respective chosen solution is possible as a choice of $\Delta_a$.

	First consider case II l ii. Here we have $\alpha=\balpha$, $x_U = \hzpl = y_a - \frac {\Delta_a} \theta - \frac 1 {2\alpha} \theta$, and $\alpha\Delta_a \ge \frac 1 2 \theta^2$. Replacing this into the equation for $b(x)$ in Proposition~\ref{prop:b} for case ii yields the following.
	\newcommand\bgamma{\bar\gamma}
	\begin{align*}
		&& 		b &= b_a - x + \frac \alpha 2 \left( x - y_a + \frac {\Delta_a} \theta + \frac 1 {2\alpha} \theta \right)^2
		\\
		\Iff && 0 &= \Delta_a^2 - 2 \theta\bgamma \Delta_a + \frac {2\theta^2} \alpha (y-b) + \theta^2 \gamma^2
		,
	\end{align*}
	where $\gamma := \frac \theta {2\alpha} - y$ and $\bgamma := \frac \theta {2\alpha} + y$.
	The second line follows by straightforward algebraic transformation.
	Let $p' = \theta\bgamma$ and $q = \frac {2\theta^2} \alpha (y-b) + \theta^2 \gamma^2$.
	By the quadratic formula, the solutions to this equation are
	\[
	\Delta_a = p' \pm \sqrt {p^{\prime 2} - q} = p' \pm \sqrt d,
	\]
	where $d = \theta^2 \frac 2 \alpha (b-\btheta y)$ and the second equality again follows by simple algebraic transformation.
	Note that $d > 0$ since, by assumption, $b/y > \btheta$, so the equation has two distinct solutions.
	Algorithm~\ref{alg:ba-reconstruction} chooses the “$-$” solution.
	We show that the “$+$” solution to the equation is not a feasible value of $\Delta_a$. To see this, assume towards a contradiction that $\Delta_a = p' + \sqrt d = \theta \left(\gamma + \sqrt {2/\alpha (b-\btheta y)}\right) =: \theta (\gamma + \delta)$. Note that $\delta > 0$ because $d > 0$. Then
	\begin{align*}
		x_U &= y_a - \frac {\Delta_a} \theta - \frac \theta {2\alpha}
		\\
		&= y_a - \gamma - \delta - \frac \theta {2\alpha}
		\\
		&= y_a - \frac \theta {2\alpha} + y - \delta - \frac \theta {2\alpha}
		= x - \frac \theta \alpha - \delta.
	\end{align*}
	Therefore,
	\begin{align*}
		\rho(x) &= 1 - \alpha (x-x_U) = 1 - \alpha (\frac \theta \alpha + \delta)
		\\
		&= 1 - \theta - \alpha \delta = \btheta - \alpha \delta < \btheta.
	\end{align*}
	Contradiction to case ii and choice of $x_L$.

	Next consider case III L ii.
	We have $\alpha = \frac 1 2 \frac {\theta^2} {b_a - \btheta y_a}$ and $x_U=0$, and $r_a \le \frac {1+\btheta} 2$. Again, replacing this into the equation from Proposition~\ref{prop:b} yields:
	\begin{align*}
		&& b &= b_a - x + \frac 1 2 \cdot \frac 1 2 \frac {\theta^2} {b_a - \btheta y_a} x^2
		\\
		\Iff&& 0 &= \Delta_a^2 - (y-b + \theta y_a)\Delta_a + (y-b)\theta y_a + \frac 1 4 x^2 \theta^2
	\end{align*}
	Let $p' = \frac 1 2 (y-b + \theta y_a)$ and $q = (y-b)\theta y_a + \frac 1 4 x^2 \theta^2$. Then
	\[
	\Delta_a = p' \pm \sqrt d
	,
	\]
	where $d := p^{\prime 2} - q = \frac 1 4 \left((b+x-\btheta y_a)^2 + \theta^2 x^2\right)$.
	We trivially have $d > 0$.\footnote{
		Note that we trivially always have $d \ge 0$. If $d=0$, then $x=0$ and we must have $0 = b+x-\theta y_a = b - \theta y$ and thus $b/y=\theta$; contradiction to the assumption of the algorithm.
	}
	Assume towards a contradiction that $\Delta_a = p' + \sqrt d$. Then
	\begin{align*}
		\Delta_a &\ge p' + \sqrt {\frac 1 4 (b+x-\btheta y_a)^2}
		\\&= p' + \frac 1 2 (b+x-\btheta y_a)
		\\&= \frac 1 2 \left(y-b + \theta y_a + b + x -\btheta y_a\right)
		\\&= \frac 1 2 \left(y_a + \theta y_a - \btheta y_a\right)
		= \frac 1 2 \cdot 2 \theta y_a = \theta y_a.
	\end{align*}
	For the first equality, note that by assumption $b_a \ge \btheta y_a$ and $b \ge b_a - x$ (since $b$ arises from $b_a$ by redeeming at a rate $\le 1$), so $b+x \ge b_a \ge \btheta y_a$, and thus the operand of the square is non-negative.

	Now, since $y_a - b_a = \Delta_a \ge \theta y_a$, we equivalently have $b_a \le \btheta y_a$, i.e., $r_a \le \btheta$.
	Contradiction to non-triviality.
\end{proof}

\section{Additional Technical Guarantees}\label{apx:sanity-lemmas}

The following technical guarantees may be useful for implementation, e.g., to understand error conditions or the data types required to store certain values. They are essentially redundant but stated explicitly here for additional clarity.
Note that numerical errors in any implementation can lead to slight violations of these properties in practice even though they are mathematically guaranteed.

The following lemma shows that Proposition~\ref{prop:xm} can be safely used to store $x_L$ in an unsigned integer of bounded width.

\begin{lemma}
	Assume that $b_a < y_a$, $r_a > \btheta$, and let $x_U \le y_a$ and $\alpha > 0$ be arbitrary. Assume that Inequality~\eqref{eq:xm-exists} holds. Then the formula for $x_L$ in Proposition~\ref{prop:xm} is well-defined and $0 \le x_L \le y_a$. Specifically, we have:
	\begin{align*}
		(y_a - x_U)^2 - \frac 2 \alpha (y_a - b_a) &\ge 0
		\\
		y_a - \sqrt {(y_a - x_U)^2 - \frac 2 \alpha (y_a - b_a)} &\in [0, y_a]
	\end{align*}
\end{lemma}
\begin{proof}
	The first inequality immediately follows from \eqref{eq:xm-exists}. Towards the second inequality, the expression is obviously $\le y_a$. To see that it is $\ge 0$, note that this is equivalent to
	\begin{align*}
		&& y_a &\ge \sqrt {(y_a - x_U)^2 - \frac 2 \alpha (y_a - b_a)}
		\\
		\Iff && y_a^2 &\ge (y_a - x_U)^2 - \frac 2 \alpha (y_a - b_a)
	\end{align*}
	and this inequality obviously holds since $y_a^2 \ge (y_a - x_U)^2$.
\end{proof}

The following lemma shows that Proposition~\ref{prop:zp} can be safely used to store $x_U$ in an unsigned integer of bounded width if $\alpha$ is chosen appropriately.
%%%
We do not consider the caps $\bzp$ and $\balpha$ explicitly here; these only make the bounds on $x_U$ tighter.

\begin{lemma}
	Assume that $b_a < y_a$ and $r_a > \btheta$ and let
	$\alpha \ge \halpha$ be arbitrary where $\halpha$ is like in Proposition~\ref{prop:alpha}.
	Then for the formula for $\hzp$ in Proposition~\ref{prop:zp} we have $0 \le \hzp(\alpha) \le y_a$.
\end{lemma}
\begin{proof}
	It follows immediately from the formulas that $\hzph(\alpha) < y_a\;\forall \alpha > 0$.

	It remains to show that $\bzp(\alpha) \ge 0$.
	By construction (see the proof of Proposition~\ref{prop:alpha}), $\hzp(\halpha)=0$. Furthermore, $\hzp(\alpha)$ is monotonically increasing in $\alpha$.	
	To see this, note that the expressions $\hzph$ and $\hzpl$ are both monotonically increasing as functions of $\alpha$. Continuity of $\hzp(\alpha)$ (see Proposition~\ref{prop:zp}) implies that $\hzp(\alpha)$ as a whole must be monotonically increasing.
\end{proof}

\begin{remark}
	In general (if $\alpha$ is not chosen according to Proposition~\ref{prop:alpha} but is arbitrary), we can have $\hzp(\alpha) < 0$. To see this, consider the limit $\alpha \to 0$. Then $\alpha \Delta_a \le \frac 1 2 \theta^2$, so $\hzp(\alpha) = \hzph = y_a - \sqrt {2 \frac {\Delta_a} \alpha} \to -\infty$.
\end{remark}

The following lemma shows that, to ensure that the boundary conditions of our method hold, it is enough to check the reserve ratio $b/y$ at the beginning of any redemption.

\begin{lemma}
	Consider Algorithm~\ref{alg:redemption} and assume that $\btheta < b/y < 1$. Consider the values $b_a$ and $y_a$ computed by the algorithm. Then also $\btheta < b_a / y_a < 1$.
\end{lemma}
\begin{proof}
	By construction, if there is a $b_a$ with $\btheta < b_a/y_a < 1$ such that $b(x; b_a, y_a) = b$, then the algorithm determines this $b_a$ and in particular it lies in the interval $(\btheta, 1)$ as required.

	Existence (and, in fact, uniqueness) of this $b_a$ can be seen explicitly as follows. Fix $x$ and $y$ and note that this fixes $y_a = y + x$. Consider the function $f(b_a) := b(x; b_a, y_a)$. We need to show that there is a unique $b_a \in (\btheta y_a, y_a)$ such that $f(b_a) = b$. This is implied by the following facts: $f$ is strictly monotonic in $b_a$ in the region of those $b_a$ values where $f(b_a)/y \in (\btheta, 1)$ (see Theorem~\ref{thm:mon-b0}); $f$ is continuous (see the propositions in Appendix~\ref{apx:calc-params}); $f(b_a)\to y$ for $b_a \to y_a$; and $f(b_a) \to \btheta y$ for $b_a \to \btheta y_a$.
	We explicitly prove the last two statements.

	For $b_a \to y_a$, first note that we eventually have $\alpha=\balpha$ and $x_U = \bzp$ constant (propositions \ref{prop:alpha} and \ref{prop:xm}) and furthermore $x_L \to x_U$ and $r_L \to 1$ (Proposition~\ref{prop:xm}). Now, via Proposition~\ref{prop:b}, it is easy to see that $f(b_a) \to y_a - x = y$.

	For $b_a\to\btheta y_a$, the statement is trivial for $x=0$, so assume $x > 0$.
	Observe that $\alpha = \halpha_L$ eventually and $\halpha_L \to \infty$ (Proposition~\ref{prop:alpha}), and furthermore $x_U=0$ eventually (by choice of $\alpha$ or via Proposition~\ref{prop:xm}) and these statements imply that $x_L\to 0$ (see Proposition~\ref{prop:xm}). Consider a $b_a$ sufficiently close to $\btheta y_a$ such that all of the eventual statements hold, $\alpha \ge \theta$, and $x_L \le x$, so that $f(b_a) = r_L y$. We show that\footnote{%
		This does not contradict Theorem~\ref{thm:mon-b0} because that theorem only stated strict monotonicity of $f$ in the region of $b_a$ where $f(b_a)/y \in (\btheta, 1)$. We do not usually have \emph{strict} monotonicity in the larger region of $b_a$ where $b_a/y_a \in (\btheta, 1)$.	
	}
	$r_L = \btheta$. This can be seen by explicit calculation as follows.
	To simplify notation, assume WLOG $y_a=1$.
	By Proposition~\ref{prop:xm} and $x_U=0$ we have $r_L = 1 - \alpha x_L$, so the statement is equivalent to showing $\alpha x_L = \theta$. This statement is equivalent to
	\begin{align*}
		&& \alpha \left( 1 - \sqrt{1 - \frac 2 \alpha \left(1-b_a\right)} \right) &= \theta
		\\
		\Iff && \alpha - \theta &= \alpha \sqrt{1 - \frac 2 \alpha \left(1-b_a\right)}
		\\
		\Iff && \left(\alpha - \theta\right)^2 &= \alpha^2 - 2 \alpha \left(1-b_a\right)
		\\
		\Iff && \theta^2 - 2\alpha \left(\theta - (1 - b_a)\right) &= 0
		\\
		\Iff && \theta^2 - 2\alpha \left(b_a - \btheta\right) &= 0
		.
	\end{align*}
	Recall now that $\alpha=\halpha_L = \frac 1 2 \frac {\theta^2} {b_a - \btheta}$. The statement now clearly holds.
\end{proof}

\clearpage

\paragraph{Disclaimer}
\textit{This paper is for general information purposes only. It does not constitute investment advice or a recommendation or solicitation to buy or sell any investment and should not be used in the evaluation of the merits of making any investment decision. It should not be relied upon for accounting, legal or tax advice or investment recommendations. This paper may contain experimental code and technical designs that may not be ready for general use. This paper reflects the current opinions of the authors and is not made on behalf of Superluminal Labs or its affiliates and does not necessarily reflect the opinions of Superluminal Labs, its affiliates or individuals associated with Superluminal Labs. The opinions reflected herein are subject to change without being updated.}
	
\end{document}